\documentclass[11pt]{article}
\usepackage[utf8]{inputenc}
\usepackage{graphicx} 
\usepackage{amsmath}
\usepackage[margin=1in]{geometry} 

\usepackage{hyperref}
\hypersetup{colorlinks=true,linkcolor=blue,citecolor=blue,urlcolor=blue}

\usepackage{amssymb}

\newtheorem{theorem}{Theorem}
\newtheorem{corollary}[theorem]{Corollary}

\newtheorem{lemma}[theorem]{Lemma}
\newtheorem{claim}[theorem]{Claim}
\newtheorem{definition}[theorem]{Definition}
\newtheorem{proposition}[theorem]{Proposition}
\newtheorem{observation}[theorem]{Observation}
\newtheorem{remark}[theorem]{Remark}

\newenvironment{proof}{\noindent\bf{Proof.}\rm}{\hfill$\blacksquare$\bigskip}

\newcommand{\items}{\mathcal{M}}

\usepackage[]{color-edits}
\addauthor{UF}{red}
\newcommand{\ufc}[1]{{\UFcomment{#1}}}
\newcommand{\ufe}[1]{{\UFedit{#1}}}

\addauthor{MB}{blue}
\newcommand{\mbc}[1]{{\MBcomment{#1}}}
\newcommand{\mbe}[1]{{\MBedit{#1}}}
\newcommand{\mbfuture}[1]{} 

\newcommand{\cout}[1]{}

\newcommand{\OLD}[1]{}

\newcommand{\prices}{{\cal{P}}}
\newcommand{\anypricei}{\anyprice{b_i}{v_i}{\items}}
\newcommand{\anyprice}[3]{AnyPrice\left(#1,#2,#3\right)}
\newcommand{\truncated}[3]{TPS\left(#1,#2,#3\right)}
\newcommand{\truncatedi}{\truncated{b_i}{v_i}{\items}}

\title{Share-Based Fairness for Arbitrary Entitlements}
\author{Moshe Babaioff\thanks{The Hebrew University of Jerusalem, {Moshe.Babaioff@mail.huji.ac.il}} \ and Uriel Feige\thanks{Weizmann Institute, {uriel.feige@weizmann.ac.il}}}

\begin{document}

\maketitle

\thispagestyle{empty}
\begin{abstract}
We consider the problem of fair allocation of indivisible items 
to agents that have arbitrary entitlements to the items. 
Every agent $i$ has a valuation function $v_i$ and an entitlement $b_i$, where entitlements sum up to~1. Which allocation should one choose in situations in which agents fail to agree on one acceptable fairness notion? We study this problem in the case in which each agent focuses on the value she gets,
and fairness notions are restricted to be {\em share based}. A {\em share} $s$ is an  function that maps every $(v_i,b_i)$ to a value $s(v_i,b_i)$, representing the minimal value $i$ should get, and $s$ is {\em feasible} if it is always possible to give every agent $i$ value of at least $s(v_i,b_i)$.



Our main result is that for additive valuations over goods there is an allocation that gives every 
agent at least half her share value, regardless of which feasible share-based fairness notion the agent wishes to use. Moreover, the ratio of half is best possible. More generally, we provide tight characterizations of what can be achieved, both ex-post  (as single allocations) and ex-ante (as expected values of distributions of allocations), 
both for goods and for chores. 
We also show that for chores one can achieve the ex-ante and ex-post guarantees simultaneously (a ``best of both world" result), whereas for goods one cannot. 

\end{abstract}

\pagebreak
\section{Introduction}
\label{sec:intro}
\setcounter{page}{1}

In this paper we aim to build a general theory of share-based fairness for allocation of indivisible items when agents have arbitrary entitlements to the items. Prior work \cite{BF22} has presented a theory of share-based fairness for agents with \emph{equal entitlements} to the items, yet in many interesting settings the entitlements of the agents are not necessarily equal. 
For example, when allocating students with slots for courses of limited capacity, some students may have a higher entitlement (say, as they are more senior). 
Another example is the setting of multi-generation inheritance: even if the  total entitlement of siblings is equally split, grand-children will have smaller entitlements than their aunts and uncles.
Yet another setting is one of allocating scarce medical treatments or medical equipment, for example during the global pandemic of COVID-19,
and the entitled agents (states, hospitals) might be of very different sizes, implying asymmetric entitlements.
{Before describing our results for fairness with arbitrary entitlements, we start with some background.}

\subsection{Background}\label{sec:intro-back}


We consider allocation of a set $\items=[m]$ of $m$ indivisible items to $n$ agents with arbitrary entitlements. 
Each agent $i\in [n]$ has a valuation function $v_i:2^{\items} \rightarrow \mathbb{R}$ over the items. 
{We assume that all valuation functions are normalized ($v_i(\emptyset)=0$ for every $i\in [n]$). 
A valuation $v$ is called a \emph{valuation over goods} if it is monotone non-decreasing ($v(S)\leq v(T)$ for $S\subseteq T$). We say that $C$ is a \emph{class of valuations over goods} if every $v\in C$ is a valuation over goods.
Valuation  $v_i$ is \emph{additive} if $v_i(S)=\sum_{j\in S} v_i(\{j\})$ (with $v_i(\emptyset)=0$).
An additive valuation over goods satisfies $v_i(\{j\})\geq 0$ for every item $j\in \items$, and such items are called \emph{goods}.
} 
A valuation $v$ is called a \emph{valuation over chores} if it is monotone non-increasing ($v(S)\geq v(T)$ for $S\subseteq T$). For an additive valuation $v_i$, items are \emph{chores} if  $v_i(\{j\})\leq 0$ for every item $j\in \items$. 
In this paper we study both the allocation of goods as well as the assignments of chores. 
An \emph{allocation} is a partition of $\items$ to $n$ disjoint sets $(S_1,S_2,\ldots,S_n)$, where agent $i$ receives the set $S_i$ for which she has value of $v_i(S_i)$.

Babaioff and Feige~\cite{BF22} (henceforth BF22)  have presented a general theory of share-based fairness for agents with \emph{equal entitlements} to the items. 
In contrast, our focus is on agents that do not necessarily have equal entitlements to the items, but rather their entitlements might be arbitrary. We assume that agent $i\in [n]$ has \emph{entitlement} $b_i> 0$ to the items, and that $\sum_{i\in [n]} b_i=1$. 

The theory of share-based fairness aims to study fairness when every agent focuses on the value she is guaranteed to receive (compare to her own valuation for $\items$), knowing only her entitlement and valuation.\footnote{{Envy-based fairness is an alternative approach in which agents also care about the partition of the items between others. We briefly discuss research on that approach in Section \ref{sec:related}.}} 
{Thus, a share $s$ maps a valuation $v$ and an entitlement $b$, to the value of the share $s(v,b)$, viewed as the guarantee on the value agent should get}. {Share functions are} anonymous, capturing the most fundamental fairness requirement: agents that are completely identical (have exactly the same valuation function and entitlement) should be treated exactly the same, and therefore all such agents  have the same share value.
{To capture ``entitlement'' as desirable, we assume that a share function must satisfy \emph{entitlement monotonicity}: $s(v_i,b_i)$ is non-decreasing in $b_i$.\footnote{Although our positive result for goods do not use entitlement monotonicity, we view this property as a fundamental one, and believe every share should 
satisfy it.} 
}

Many different shares have been suggested in the literature. 
The most basic one is the proportional share, defined as $b_i\cdot v_i(\items)$. When items are indivisible it is clearly impossible to give every agent that amount,\footnote{{That impossibility holds ``ex-post'' (for deterministic allocations). In contrast, randomized shares can give that amount in expectation (ex-ante).}} 
so other shares have been suggested.  
An example for a share for equal entitlement is the maximin share (MMS) \cite{Budish11}, which for valuation $v_i$ in an $n$ agent setting (corresponding to an entitlement of $1/n$), is defined to be the highest value an agent can obtain by partitioning the set of items to $n$ bundles, and getting the lowest value one. The value of this share is denoted by $MMS(v_i,1/n)$.  The literature has also presented multiple share definitions applicable for arbitrary entitlements, including the AnyPrice share (APS) \cite{BEF2021b}, the Truncated Proportional Share (TPS) \cite{BEF2021b}, {the Weighted MMS (WMMS)~\cite{farhadi2019fair}} and the Quantile share \cite{babichenko2023fair}. (See also Section \ref{sec:related}.) 


{
Fairness that is share-based is captured by the guarantee that every agent must get a bundle of value that is at least as high as her share.  This can always be done only if the share $s$ is \emph{feasible}, that is, only if for every profile of valuations $(v_1,v_1,\ldots,v_n)$ {and entitlements $(b_1,b_2,\ldots,b_n)$ of the agents, there exists some allocation $(S_1,S_2,\ldots,S_n)$ that gives each agent value that is at least her share: $v_i(S_i)\geq s(v_i,b_i)$ for all $i\in [n]$. Such an allocation is called an \emph{acceptable allocation} for $s$, or an \emph{$s$-allocation}.} 
(All these definitions apply to ex-post shares -- shares defined on deterministic allocations. They also naturally extend to randomized allocations in which shares are defined ex-ante (over distributions of allocations), and agents are risk neutral (so lotteries are evaluated using their expected value)).

With many alternative share definitions available, each of the agents might argue for a different feasible share to be used as the fairness criteria. In particular, agent $i$ with valuation $v_i$ and entitlement $b_i$ might argue in favor of a feasible share $s^i$ under which her own share value $s^i(v_i,b_i)$ is exceptionally large {(though she is constrained to suggest a feasible share)}.
If there was one feasible share $s$ that is at least as good as $s^i$ for every $i$ (that is, $s(v_i,b_i)\geq s^i(v_i,b_i)$ for every $i$), then the agents could have agreed to use that share. 
BF22 have shown that, unfortunately, even when the entitlements are equal, such a feasible share $s$ (one that is as large as any other feasible share) does not exist for additive valuations over goods. This clearly extends to the case that entitlements are not necessarily equal. Given that strong impossibility result, we consider a relaxed problem: guarantee every agent a large as possible fraction $\rho$ of any feasible share of her choice. 
A feasible share that is $\rho$-dominating will give the required relaxed guarantee: 

\begin{definition}
\label{def:dominating}
    For $\rho > 0$ and class $C$ of valuations over goods, a share function ${s'}$  $\rho$-dominates (or dominates if $\rho = 1$) share $s$ if ${s'}(v,b) \ge \rho \cdot s(v,b)$ for every valuation $v\in C$ and every entitlement $b$. Share $s'$ is \emph{$\rho$-dominating} for class $C$ if it $\rho$-dominates every feasible share $s$ for class $C$.
\end{definition}

Note that the space of shares when entitlements are arbitrary is much richer than the space of shares for equal entitlements: while a share for equal entitlements with $n$ agents only depends on the valuation, the share for  arbitrary entitlements can depend on the entitlement of the agent. Thus, it is much harder for a share to be $\rho$-dominating share. 
{For example, such a share will need to dominate the feasible share that for an agent with valuation $v_i$ and entitlement $b_i \ge 0.35$  gives value $MMS(v_i,\frac{1}{2})$, and nothing to agents of entitlements smaller than $0.35$. (This share is feasible, as there can be at most two agents with entitlement at least $0.35$.)}

We aim to $\rho$-dominate all feasible shares. 
Our positive results present values of $\rho$ for which there is a feasible share that $\rho$-dominates every feasible share, even shares that might look unnatural. 
In proving that a better $\rho$ is not possible, we do not want an impossibility that hinges on unnatural shares, and we aim for an impossibility that holds even if we only require domination of feasible shares that are {``nice''.} 
For example, it is natural to require that a share is monotone in the valuation: an increase in the value of some item should not decrease the value of the share. 
We aim for impossibilities that hold even when requiring to $\rho$-dominate only shares {that are ''nice'', satisfying} 
many strong properties, including 
item name independence, valuation monotonicity, continuity and scale-freeness (see discussion in Section \ref{sec:model-share-props}).

\subsection{Our Results}\label{sec:our-results}

We present a systematic study of feasible shares and of domination relation between shares for arbitrary entitlements.
We characterize necessary and sufficient conditions for a share to dominate every other feasible share. 
We then focus on additive valuations, and consider items that are goods as well as items that are chores. We prove tight domination results for ex-post shares, for ex-ante shares, and for their combination (known as ``best-of-both-worlds'' results). 
{All our positive results are proven by presenting allocation mechanisms in which every agent has a strategy that guarantee her an allocation of the value she aims for, no matter how others play. We view these mechanisms as natural and it seems plausible to us that these mechanisms could be used in practice. See Section \ref{sec:practical} for a discussion.}

\subsubsection{Allocation of Goods}

We  pursue the desiderata of dominating all feasible shares for \emph{arbitrary entitlements}. Previously, 
a similar quest was pursued only for the equal entitlement case~\cite{BF22}. 
Our main result is a tight bound on the largest $\rho$ such that there exists a feasible share that $\rho$-dominates every feasible share for arbitrary entitlements for agents with additive valuations over goods.

\begin{theorem}\label{thm:main-intro} 
Consider settings with additive valuations over goods and arbitrary entitlements. 
There exists a poly-time computable feasible ex-post share that 
{is at least $\frac{1}{2}$ of every feasible ex-post share (it is a $\frac{1}{2}$-dominating share).
As the share is feasible, there exists an allocation that gives every agent at least $\frac{1}{2}$ of every feasible ex-post share.
{In contrast, for any $\rho>1/2$ there exists an allocation instance in which in every allocation there is an agent that gets less than $\rho$-fraction of her chosen feasible ex-post share.}
} 

Additionally, for any $\varepsilon>0$ there exists a poly-time computable feasible share that $(\frac{1}{2}-\varepsilon)$-dominates every feasible share, and for which an acceptable allocation can be computed in polynomial time. 
\end{theorem}


Some remarks are in order.

First, while we obtain a tight {$\rho$-domination} result for arbitrary entitlements ($\rho=\frac{1}{2}$), 
{for the case of equal entitlements there is  a gap {is the known results} regarding the values of $\rho$ for which $\rho$-domination by a feasible share is possible. See Section \ref{sec:related} for details.
}

Second, from computational point of view, our results are almost tight: the best that we can hope for is $\rho$ of $\frac{1}{2}$, and we obtain a $\rho$ of $(\frac{1}{2}-\varepsilon)$ with both share computation and acceptable allocation computation being polynomial time. 

{Third, our proof of the impossibility result actually shows a stronger claim, showing that for any $\rho>1/2$ there is no feasible ex-post share that $\rho$-dominates only those feasible shares that are ``nice'' {(rather than all feasible shares)}. Such shares satisfy natural properties like name independence, {valuation} monotonicity and {valuation} continuity. See Section \ref{sec:model}.}


{Next we consider ex-ante domination {for agents with} additive valuations over goods.} 
\begin{theorem}\label{thm:intro-goods-ex-ante}
Consider settings with additive valuations over goods and arbitrary entitlements. 
For $\gamma \simeq 1.69103$ (so  $\frac{1}{\gamma} \simeq  0.591$)
there exists a poly-time computable feasible ex-ante share that $\frac{1}{\gamma}$-dominates every feasible ex-ante share {($\frac{1}{\gamma}$-dominating).}
{In contrast,} for any $\rho>\frac{1}{\gamma}$, there is no feasible ex-ante share that $\rho$-dominates every feasible ex-ante share. 
\end{theorem}


We next consider the problem of obtaining a \emph{Best-of-Both-Worlds} result for arbitrary entitlements.  
A Best-of-Both-Worlds (BoBW) result aims to present a randomized allocation rule that is both ax-ante fair as well as ex-post fair. That is, for every instance, it obtains some fairness in expectation, but is supported only on allocations that are fair {(ex-post)}. 
An ultimate result will be a ``double-domination'' BoBW result, getting (approximate) domination of every feasible ex-ante share, as well as (approximate) domination of every feasible ex-post share.  We prove a strong impossibility result for BoBW when agents have additive valuations over goods and arbitrary entitlements. 

\begin{proposition}\label{prop:intro-BoBW-goods}
    For any $n\geq 2$ and $m\geq 1$, 
    there are {allocation instances} 
    with $n$ agents of unequal entitlement that have  additive valuations over $m$  {goods} 
    in which ``double-domination'' BoBW result is impossible: Every randomized allocation that gives every agent some constant fraction of 
    {the ex-ante proportional share (which is a feasible ex-ante share)} must be supported on some allocation in which an agent does not get a constant fraction of some  feasible ex-post share. 
\end{proposition}

Note that this impossibility result for arbitrary entitlements stands in  sharp contrast to the positive result proven in \cite{BEF2021c} for the case of equal entitlements: for equal entitlements it is possible to find a randomized allocation that is ex-ante proportional (thus dominates every ex-ante feasible share) and also $1/2$-dominates every feasible ex-post share.   

\subsubsection{Assignment of Chores}


{We next move to consider items that are chores (undesirable).
The definitions in this case involve some natural modifications compared to the definitions for allocations of goods. 
When items are chores, all chores must be {assigned} to a set of agents that have different responsibilities for handling the chores: each agent $i$ has a \emph{responsibility} $b_i>0$, such that $\sum_i b_i=1$.} 
In the case of chores, every chore (item) has a non-positive value ($v_i(\{j\})\leq 0$ for every item $j\in \items$). To simplify the terminology, instead of using the value function, we consider the cost function which is the negation of the value function. That is, for chores, every agent $i$ has a \emph{cost function} $c_i$ over chores, which is normalized (the cost of the empty set is~0), non-negative, and weakly monotone non-decreasing. 
The agent prefers bundles of lower cost over those of higher cost. 
{A share $s$ determines the share cost $s(c_i,b_i)$ of agent $i$ with cost function $c_i$ and responsibility $b_i$, This is the maximal cost the agent should bear: a set of $S$ chores is \emph{acceptable} to agent $i$ if $c_i(S)\leq s(c_i,b_i)$.} An assignment  must assign all chores, and it  is acceptable if each agent is assigned an acceptable set of chores. {Additionally, a share for  chores need to be monotone in the responsibility, that is, the more responsibility, the larger the share: $s(c_i,b_i)\leq s(c_i,b^+_i)$ for any $b_i\leq b_i^+$ and any $c_i$.}

\begin{definition}
    {For $\rho > 0$  and class $C$ of cost functions, a share function ${s'}$  \emph{$\rho$-dominates} (or simply \emph{dominates} if $\rho = 1$) share $s$ if ${s'}(c,b) \le \rho \cdot s(c,b)$ for every cost function $c\in C$ and every responsibility $b$. Share $s'$ is \emph{$\rho$-dominating} if it $\rho$-dominates every feasible share $s$  for class $C$.}
\end{definition}

Our main result for chores is a tight bound on the smallest $\rho$ such that there exists a feasible share that $\rho$-dominates every feasible share for arbitrary responsibilities over additive chores.
Moreover, we can compute the feasible share that obtains our positive result  in polynomial time, and we can also find an acceptable assignment in polynomial time. 

\begin{theorem}
\label{thm:feasibleChores-intro}
    For additive costs and arbitrary responsibilities, the {\em Rounded{-responsibilities} Round Robin share} (RRR) {(see Definition \ref{def:RRR})}
   is polynomial time computable 
   and {with cost at most twice that of any feasible ex-post share ($2$-dominating).} 
   It is feasible, and a feasible assignment  can be computed in polynomial time. 

   {In contrast}, for any $\rho<2$  there is no feasible share that $\rho$-dominates every feasible share.
\end{theorem}

{We next consider ex-ante assignment of chores. 
We observe that there is no feasible ex-ante share that is $\rho$-dominating, if $\rho<2$. 
{Thus, one would like a feasible ex-ante share that is $2$-dominating, and this optimal ex-ante domination is easily obtained by the proportional share. Yet, ex-ante fairness by itself is a weak fairness guarantee, and we would like an ex-post guarantee as well.}
We present a best-possible double-domination Best-of-Both-Worlds result: a randomized allocation that 
dominates the proportional share in expectation (and $2$-dominates every feasible ex-ante share -- the best we can hope for), together with $2$-domination of every feasible ex-post share (again, the best we can hope for). 
}

\begin{theorem}
\label{thm:BoBWChores-intro}
{Consider assignment of chores to agents with additive cost functions and arbitrary responsibilities.
There is a polynomial time randomized assignment algorithm that for every input instance outputs a distribution over assignments for which
each agent suffers an expected cost that is at most her proportional share (and at most twice her share for any feasible ex-ante share), and is supported on allocations of cost that is
no more than twice her cost in any feasible ex-post share ($2$-dominating).} 

\end{theorem}


This strong positive result for chores stands in sharp contrast to our impossibility for double-domination Best-of-Both-Worlds result for the case of goods {(Proposition \ref{prop:intro-BoBW-goods})}.

{
\begin{remark}
    Our example showing impossibility of {ex-post $\rho$-domination for} $\rho < 2$ in Theorem~\ref{thm:feasibleChores-intro} involves an agent whose responsibility is very high, close to~1. One may ask whether a better value of $\rho$ is possible if all responsibilities are upper bounded by $\delta$, for some small positive $\delta < 1$. For arbitrarily small $\delta$ we design examples showing the impossibility of $\rho < \frac{3}{2}$, but were not able to design examples that show impossibility of {$\rho$ satisfying $\frac{3}{2}\leq \rho < 2$}. 
    
    {In contrast, our examples showing impossibility of $\rho < 2$ for ex-post $\rho$-domination in the Best-of-Both-Worlds setting with the ex-ante requirement being the proportional share (Theorem~\ref{thm:BoBWChores-intro}) hold for all values of $\delta$, no matter how small.}
\end{remark}
}

\subsection{Overview of Our Techniques}\label{sec:intro-tech} 

In this section we provide an overview of the techniques used in our proofs. As our main results concern additive valuations, we assume throughout this section that all valuations are additive, unless explicitly stated otherwise. Moreover, to simplify this overview, we only explain how we design allocation algorithms with the desired approximation ratios, and omit discussions of feasible shares that are implied by our approaches.

\subsubsection{Identifying tight dominating shares}

We start by considering goods. Recall the notion of a dominating share from Definition~\ref{def:dominating}. A key aspect of our work is that of identifying dominating shares $\hat{s}$ that are tight in the sense that for every setting for the valuation $v_i$ and entitlement $b_i$, there is some share $s$ that is both feasible and satisfies $s(v_i,b_i) = \hat{s}(v_i,b_i)$.

For the case of equal entitlements, it was already shown in~\cite{BF22} that the MMS is a tight dominating share ex-post, and it is easy to see that the proportional share $PS(v_i,b_i) = b_i \cdot v_i(\items)$ is a tight dominating share ex-ante. We define variations of MMS and PS that are tight for the case of arbitrary entitlements. This is done by ``rounding up" the entitlement $b_i$ to {the smallest} value that is of the form $\frac{1}{k}$, and then considering the corresponding share for the equal entitlement case with $k$ agents. 

\begin{definition}\label{def:unitGoodsIntro}
    For entitlement $0 < b \le 1$, let integer $k$ be such that $\frac{1}{k+1} < b \le \frac{1}{k}$. 
    We define the {\em unit upper bound} of $b$ to be $\frac{1}{k}$, and denote it by $\hat{b}$. 
    Define the ex-post share $\widehat{MMS}$ as $\widehat{MMS}({v_i},b_i) = MMS({v_i},\hat{b}_i)$ and the ex-ante share $\widehat{PS}$ as $\widehat{PS}({v_i},b_i) = PS({v_i},\hat{b}_i)$.
\end{definition}

The facts that $\widehat{MMS}$ and $\widehat{PS}$ are dominating follows by considering input instances in which there are $k = \lfloor \frac{1}{b} \rfloor$ agents with identical valuations $v_i$ and identical entitlements $b_i$.  Tightness of the share $\hat{s}$ follows by considering for each $(v_i,b_i)$ a {\em personalized share} $s_{v_i,b_i}$ that is feasible and satisfies $s_{v_i,b_i}(v_i,b_i) = \hat{s}(v_i,b_i)$. In our proofs we take the extra effort to ensure that the personalized shares $s_{v_i,b_i}$ are not some artificial shares that make no sense in practice, but rather shares that enjoy natural properties that one would like shares to enjoy.  They are {\em name-independent} (see Definition~\ref{def:orderedShare}), {\em self maximizing} in the sense of~\cite{BF22} (which implies among other things continuity and monotonicty in $v_i$), and they are monotone in the entitlement $b_i$. 

For chores, the corresponding tight shares involve rounding $b$ downwards instead of upwards.

\begin{definition}
\label{def:unitChoresIntro}
    For {responsibility} $0 < b < 1$, let $k$ be such that $\frac{1}{k+1} \le b < \frac{1}{k}$.
    Then the {\em unit lower bound} on $b$, denoted by $\check{b}$, is $\frac{1}{k+1}$. Define the ex-post share $\overline{MMS}$ as $\overline{MMS}(c_i,b_i) = MMS(c_i,\check{b}_i)$ and the ex-ante share $\overline{PS}$ as $\overline{PS}(c_i,b_i) = PS(c_i,\check{b}_i)$.
\end{definition}

\subsubsection{Identifying target approximation ratios, and achieving them ex-ante}

We have identified the tight dominating shares for goods, $\widehat{MMS}$ ex-post and $\widehat{PS}$ ex-ante, and for chores, $\overline{MMS}$ ex-post and $\overline{PS}$ ex-ante. None of these shares is feasible. There are two different sources for infeasibility. One is infeasibility of the corresponding equal entitlement share. This applies to the MMS~\cite{KurokawaPW18}, but not to the PS (as PS is feasible ex-ante in the equal entitlement setting). The other source is the rounding of 
{$b_i$, either to $\hat{b}_i$ (for goods) or to $\check{b}_i$ (for chores).} 
We first discuss the effect of this rounding issue.

By considering instances in which all agents have equal valuations, it is clear that the ratio between $\sum_i \hat{b}_i$ (or $\sum_i \check{b}_i$) and $\sum b_i$ gives bounds on the best approximation ratios that can be guaranteed relative to the corresponding tight dominating share notion $\hat{s}$. As we fix $\sum b_i = 1$, determining the bounds implied by this is the same as determining the maximum possible value of $\sum_i \hat{b}_i$ and the minimum possible value of $\sum_i \check{b}_i$. 

For chores, the inequality $  \frac{1}{2} b_i <\check{b}_i $ that holds for all $0 < b_i < 1$ implies that $\sum_i \check{b}_i > \frac{1}{2}$. The lower bound of $\frac{1}{2}$ is best possible in the sense that for every $\varepsilon > 0$, if $b_1 = 1 - \varepsilon$, then $\sum_i \check{b}_i < \frac{1}{2} + \varepsilon$. (We discuss in our paper also tightness of our bounds when no agent has responsibility close to~1. We obtain partial results in this case but some questions remain open. {See Section~\ref{sec:chores}.}) 

For goods, determining the maximum possible value of $\sum_i \hat{b}_i$ is more demanding. We determine this value to be $\gamma \simeq 1.69103$, which is a limit of a sum related to an integer sequence known as the {\em Sylvester sequence}. 

The above implies that the best possible ex-ante guarantees are $\frac{1}{\gamma}$-$\widehat{PS}$ for goods (Theorem~\ref{thm:intro-goods-ex-ante}) and $2$-$\overline{PS}$ for chores. Each of these ratios is trivially attainable ex-ante, by giving agent $i$ all of $\items$ with probability $\hat{b}_i$ (for goods) or $\check{b}_i$ (for chores).

We now move to consider also the effect of infeasibility of the MMS, which is relevant only to the ex-post setting. For chores we show that  $2$-$\overline{MMS}$ allocations always exist (see Section~\ref{sec:BoBWIntro}), and so all the loss in the approximation ratio can be attributed to the rounding of the $b_i$. This establishes the best possible ex-post guarantee for chores.

For goods the situation is different. For every $\varepsilon > 0$ we design allocation instances in which there is no $(\frac{1}{2} + \varepsilon)$-$\widehat{MMS}$ allocations. Moreover, these examples can be designed even if one fixes arbitrarily small $\delta > 0$, and requires all entitlements to be no more than $\delta$. 
This establishes that a {fraction of $\widehat{MMS}$ that is better than $\frac{1}{2}$} 
cannot be guaranteed for goods. The main and most innovative technical content of our work is to show that there are allocation algorithms that guarantee $\frac{1}{2}$-$\widehat{MMS}$. But before explaining our approach for proving this, we discuss {\em best of both worlds} (BoBW) results.

\subsubsection{Best of Both Worlds results}
\label{sec:BoBWIntro}

Recall that in BoBW results we wish to design randomized mechanisms that simultaneously achieve strong ex-ante guarantees and strong ex-post guarantees {(a guarantee that holds for every allocation in the support)}. Impossibility of BoBW results imply that the ex-ante goal and the ex-post goal conflict with each other, no {randomized} allocation mechanism can satisfy both requirements.

For goods, we establish an impossibility result (Proposition~\ref{prop:intro-BoBW-goods}). Its proof is by designing explicit instances in which {there is no randomized allocation that simultaneously provides constant approximation to $\widehat{PS}$ (ex-ante) and to $\widehat{MMS}$ (ex-post).} 

For chores, we establish the existence of a BoBW randomized assignment mechanism, and moreover, its approximation ratio with respect to both $\overline{PS}$ (ex-ante) and $\overline{MMS}$ (ex-post) is~2, which is best possible (Theorem~\ref{thm:BoBWChores-intro}). The randomized assignment mechanism that we use is one that was previously briefly discussed in the discussion session of~\cite{FH23arxiv}. The analysis showing that it gives $2$-$\widehat{MMS}$ is new. We now give some details.

Refer to a randomized assignment as {\em natural} if for every agent $i$ and chore $e$, the probability that $e$ is assigned to $i$ is exactly the responsibility $b_i$. (There is no assumption of the nature of the correlation among assignments of different items.) Observe that any natural randomized assignment assigns each agent exactly her PS ex-ante, and hence at most twice her $\overline{PS}$. Hence it remains to establish that there are natural randomized assignments that can be implemented as sampling from (deterministic) assignments that are $2$-$\overline{MMS}$. We prove this by starting with the natural fractional assignment in which each agent $i$ is assigned a $b_i$ fraction of each chore. The fractional assignment is ``rounded" to integral assignments by a rounding procedure that is randomized but correlated across items. The randomness of the rounding ensures that each agent $i$ is assigned every chore $e$ with probability exactly $b_i$. This establishes the ex-ante property. The correlations in the rounding ensure that for every deterministic assignment that may result from the rounding, the cost of the chores assigned to  agent $i$ exceed $b_i \cdot c_i(\items)$ (her PS, which is a lower bound on her MMS) by the cost of at most one chore. This immediately implies that the cost assigned to the agent it at most $2$-$MMS$, and we show that in fact it is at most $2$-$\overline{MMS}$.

\subsubsection{The bidding game for ex-post allocation of goods}

We now turn to describe our approach for proving {our main result, showing that every allocation instance has an $\frac{1}{2}$-$\widehat{MMS}$ allocation} 
(Theorem~\ref{thm:main-intro}). In is based on a bidding game, introduced in~\cite{BEF2021b}. In the version of the bidding game that we consider, each agent $i$ receives an initial budget $b_i$, equal to her entitlement. The game precedes in rounds. In each round, each agent places a bid of her choice, not larger than her remaining budget. 
{A bidder with the highest bid wins the round} 
(our results will apply to any tie breaking rule, and hence we do not specify an explicit tie breaking rule) 
and pays her bid, meaning that the bid value is subtracted from her remaining budget.  The winner selects an item of her choice among those items not already selected in previous rounds, and this ends the round. The game ends when all items are allocated (after $m$ rounds). 

A strategy for an agent $i$ in the bidding game is a function that at each round $r$, maps the information available to the agent at  round $r$ to a bid in round $r$, and also dictates which item to select if $i$ is the winner of the round.  (We assume that the information available to $i$ includes at least the set $\items$ of items, her valuation $v_i$ and entitlement (initial budget) $b_i$, which items were selected in previous rounds, by whom, and how much was the paid for each selected item.) For a target value $T$, such as $\frac{1}{2}$-$\widehat{MMS}$, we shall refer to a strategy $\sigma$ as $T$-{\em safe} (or just {\em safe}, if $T$ is clear from the context) if no matter how the initial entitlement is distributed among other agents, no matter what strategy is used by other agents, and no matter what tie breaking rule is used in determining the winner, an agent $i$ that uses $\sigma$ is guaranteed to end the game with a bundle of value at least $T$.

We prove that a $\frac{1}{2}$-$\widehat{MMS}$-safe strategy exists. This implies that a $\frac{1}{2}$-$\widehat{MMS}$ allocation exists, as each of the agents may use the safe strategy.

A main difference between our work and previous work~\cite{BEF2021b, bUF23}  that designed safe strategies for various versions of the bidding game and other fairness notions is that the development here seems to require a much better understanding of the strategy space of the bidding game. Consequently, in the process of developing our safe strategy, we present additional results about the bidding game that may be of independent interest.

\subsubsection{New insights about the bidding game}

For entitlements of the form $b_i = \frac{1}{n}$ where $n$ is an integer, is there an MMS-safe strategy? The answer is negative, because if there was an MMS-safe strategy $\sigma$ this would imply that MMS allocations always exist in the equal entitlement case (as all agents could use $\sigma$), whereas we know that this is not true~\cite{KurokawaPW18}. However, what about the case $n=2$, in which MMS allocations do exist? A-priori, the answer is not clear. {Plausible answers include {\em no}, {\em yes}, and {\em it depends} (for example, one may speculate that if ties among equal bids are consistently broken in favor of the agent the answer is {\em yes}, and otherwise it is {\em no}).}


In our work we provide a characterization for what we refer to as {\em worst-case-optimal} strategies. A \emph{worst-case-optimal} (or ``optimal'', for short) strategy for agent $i$ with valuation $v_i$ and entitlement $b_i$ in the bidding game is a strategy that maximizes the agent's value in a $0$-sum game against an {\em adversary} with entitlement $1 - b_i$ that aims to minimize the value of the agent (when all ties in bids are broken by the adversary). The optimal strategies in this case are given by a backward induction process, and we characterize the possible outcomes of this process. 
Among other things, our characterization implies that the optimal value $V_i(b_i)$ is a piece-wise constant function of the entitlement $b_i$, going over all subset values in increasing order. Additionally, for all entitlements except of the points in which the function changes value (e.g, $b_i = \frac{1}{2}$), it holds that $V_i(b_i) + V_i(1 - b_i) = v_i(\items)$. 

Using this characterization we determine that agents with entitlement $b_i = \frac{1}{2}$ do have an MMS-safe strategy. Consequently, if there are two agents with entitlement $\frac{1}{2}$ and the same valuation, if both use the MMS-safe strategy the outcome is an MMS partition. As computing MMS partitions is NP-hard (for two agents it is weakly NP-hard, by its equivalence to the weakly NP-hard problem of PARTITION), this implies that playing optimally in the bidding game is (weakly) NP-hard.

Our characterization also implies that for all $b_i > \frac{1}{2}$, there is a $\frac{1}{2} \cdot v_i(\items)$-safe strategy. This result is crucial for proving that $\frac{1}{2}$-$\widehat{MMS}$ allocations exist, because for agents with entitlements  $b_i > \frac{1}{2}$ we have  $\hat{b}_i = 1$, implying that $\widehat{MMS}(v_i,b_i) = v_i(\items)$, {and thus for this agent to obtain $\frac{1}{2}$-$\widehat{MMS}$ she must get at least $\frac{1}{2} \cdot v_i(\items)$.}

Unfortunately, our characterization of optimal  strategies does not provide sufficiently useful information on what value an agent can guarantee to herself when her entitlement is smaller than $\frac{1}{2}$. {This will be addressed next.} 

\subsubsection{New strategies for the bidding game}

To prove that there are $\frac{1}{2}$-$\widehat{MMS}$-safe strategies for the bidding game, we design a new strategy for which we can prove that it guarantees at least $\frac{1}{2}$-$\widehat{MMS}$, and in fact, also a stronger share-based benchmark of $\frac{1}{2}$-$\widehat{TPS}$ (see Definition~\ref{def:TPS}). 

Our new strategy can be thought of as being designed in an inductive manner. For integer $k \le 1$, let $I_k$ denote the interval $(\frac{1}{k+1}, \frac{1}{k}]$ (including $\frac{1}{k}$ but not $\frac{1}{k+1}$), and observe that for all entitlements $b \in I_k$ it holds that $\hat{b} = \frac{1}{k}$. Our construction of a $\frac{1}{2}$-$\widehat{TPS}$-safe strategy proceeds in an inductive manner, by induction of $k$. When designing a $\frac{1}{2}$-$\widehat{TPS}$-safe strategy for entitlements in $I_{k+1}$, we make use of the fact that we already have $\frac{1}{2}$-$\widehat{TPS}$-safe strategies for all entitlements in $\cup_{j=1}^k I_j$.

The base case of the induction is $k=1$, which corresponds to entitlements satisfying $b_i > \frac{1}{2}$. For the base case, we use the worst-case-optimal strategy, which as we have seen ensures at least $\frac{1}{2} v_i(\items) \ge \frac{1}{2} \widehat{TPS}(v_i,b_i)$. 

Proving the inductive step is the most technically involved part of our paper. In particular, the bidding strategy that we design involves  ``look ahead". The bid on the first round depends not only on $v_i(\items)$ and on the value of the most valuable item (the one that the agent will select if she wins the round), but also on the values of additional items. The first inductive step, handling the case that $k=2$, is the most challenging step. In that case the look ahead goes all the way up to the eighth most valuable item. Moreover, such extensive look head is necessary for the class of strategies that we consider (see Remark~\ref{ref:lookahead}). 
Section~\ref{sec:goods-pos-ex-post} provides a more extensive overview of our induction step {and presents the proof.}

\subsubsection{Polynomial time strategies}

The reason why our $\frac{1}{2}$-$\widehat{TPS}$-safe strategy does not run in polynomial time is because for the case $b_i > \frac{1}{2}$ (or more generally, at any intermediate step of the game in which the remaining budget of the agent is larger than half the total remaining budget) we need a strategy that gives the agent at least half the total (remaining) value. We know that the worst-case-optimal strategy offers such a guarantee, but this strategy is NP-hard to compute. The question of whether there is polynomial time strategy {(that is non-optimal yet) provides} 
such a guarantee, is open.


In our work, in the context of polynomial time strategies, we show that for every $\varepsilon > 0$ there is a polynomial time $(\frac{1}{2} - \varepsilon)$-$\widehat{TPS}$ safe strategy. For this purpose, we design a polynomial time strategy that gives the agent at least $(\frac{1}{2} - \varepsilon)\cdot v_i(\items)$ in the special case in which $b_i > \frac{1}{2}$. 
The design of our polynomial time strategy is based on the following approach. We first show that for a certain subclass of additive valuations (in which all but a constant number of items have the same value), the worst-case-optimal strategy can be implemented in polynomial time. Then we show that every additive valuation function $v_i$ can be replaced by a valuation $v'_i$ from this subclass, such that the worst case safe strategy for $v'_i$ guarantees a bundle whose value (with respect to $v_i$) differs by at most $\varepsilon \cdot v_i(\items)$ from the value guaranteed by  the worst case safe strategy for $v_i$. 






\subsection{Related Work}
\label{sec:related}

There is extensive research on fair allocation of resources, much more than we can survey in this section. The reader is referred to the books \cite{brams1996fair,moulin2004fair} and surveys \cite{AzizSurvey2022, fairSurvey2022} for background. We next focus only on the literature most related to our paper, and only on the case of additive valuations.

Our work concerns share-based notions for agents of arbitrary entitlements. In doing so, we make extensive use of the maximin share (MMS)~\cite{Budish11} which is a share for agents with equal entitlements. For additive valuations over goods, the MMS is no larger than the proportional share (though for non-additive valuations it may be larger). 
The MMS is not feasible: there exist additive valuations for which no allocation gives every agent her MMS~\cite{KurokawaPW18}.
This led to multiple works \cite{KurokawaPW18,amanatidis2017approximation,GM19,BK20,GhodsiHSSY18,garg2019approximating,GT20,FST21,FN22,akrami2023simplification,akrami2023breaking} trying to determine the best ratio $\alpha$ such that in instances with additive valuations, there always is an allocation that gives every agent a bundle of value at least a $\alpha$ fraction of 
her MMS {(an $\alpha$-MMS allocation)}.  
The optimal value of $\alpha$ is known to lie between $\frac{3}{4} + \frac{3}{3836}$~\cite{akrami2023breaking} and $\frac{39}{40}$~\cite{FST21}, but its exact value is not known. For the case of chores, 
the optimal ratio is known to lie between $\frac{13}{11}$~\cite{HuangS23} and $\frac{44}{43}$~\cite{FST21}.\footnote{Note that for goods the ratios are smaller that~1, whereas for chores they are larger than~1. This is because for goods agents wish to get bundles of high value, whereas for chores they which to get bundles of low cost.} 

Another share of much relevance to our work is the TPS~\cite{BEF2021c}. This share is defined only for goods, and assumes additive valuations. In the equal entitlement case, the TPS is at least as large as the MMS. Unlike the MMS, the best ratios with respect to the TPS are know, $\frac{n}{2n-1}$ in the equal entitlement case~\cite{BEF2021c}, which generalize to $\frac{1}{2 - b_i}$ in the arbitrary entitlement case~\cite{BEF2021b}.

We now survey work most related to ours but in the equal entitlement case, as this work serves as background to our work. Most related is the work of \cite{BF22}, which 
{introduced the agenda of dominating all feasible shares for equal entitlements (which we extend to arbitrary entitlements),} 
and identified the MMS as the minimal dominating (ex-post) share. Unfortunately, as stated above, there is a significant gap {in our knowledge} regarding the values of $\rho$ for which the $\rho$-MMS is feasible, implying the exact same gap {in our knowledge} for $\rho$-domination in the equal entitlements case. {In contrast, for $\widehat{MMS}$, 
which by our characterization is the minimal dominating (ex-post) share for arbitrary entitlements, we do get in our paper tight approximation ratios, implying  tight $\rho$-domination results for arbitrary entitlements.} As to ex-ante shares, in the equal entitlement case the proportional share is the minimal dominating share, and it is feasible (ex-ante). Previous work~\cite{BEF2021c, FH23} implies best of both worlds (BoBW, providing both ex-ante and ex-post guarantees) results, though similar to the situation with ex-post shares, the best ratios that such results can achieve are not known, neither for goods nor for chores. In contrast, in the arbitrary entitlement case, we show {a strong} impossibility for BoBW for goods,
and tight BoBW results for chores.



{We now discuss relations between our proof techniques and those of earlier work.}
To show our ex-post domination for goods, we design strategies for a bidding game that was introduced in~\cite{BEF2021b}. To achieve the tight benchmark of $\frac{1}{2}$-$\widehat{MMS}$ we embark on a systematic study of the strategy space for the bidding game, whereas previous strategies for this game were designed in a more ad-hoc manner. To show our ex-post domination for chores, we use an approach based on picking sequences. For chores, picking sequences were previously shown to give allocations with {quite good} fairness guarantees, 
both in the equal entitlement case and the arbitrary entitlement case~\cite{ALW22,FH23}. {However, in the equal entitlement case, their approximation for the MMS is no better than $\frac{3}{2}$~\cite{FH23}, whereas other approaches~\cite{HuangS23} provide significantly better approximation ($\frac{13}{11}$).} {In contrast,} in our work we establish that in the arbitrary entitlement case and with benchmarks such as $\widehat{MMS}$, {picking sequences} give BoBW results that are best possible. This may serve as  justification for using these relatively simple assignment mechanisms.

Several prior papers have studied settings with agents that have arbitrary entitlement. 
This setting led to definitions of additional types of shares, including the weighted maximin
share (WMMS)~\cite{farhadi2019fair} (which in fact does not qualify as a share under our definitions, as it depends on more than just the valuation and entitlement of the agent), the ``$\ell$-out-of-$d$'' share~\cite{BabaioffNT2020}, the AnyPrice Share (APS) {(presented in Appendix \ref{app:aps})} and the Truncated-Proportional Share (TPS)~\cite{BEF2021b} {(presented in Section \ref{sec:model})}. 
The Quantile share \cite{babichenko2023fair} is another share definition for agents with arbitrary entitlements (although that paper defines it only for equal entitlement, the definition naturally extends to the case of arbitrary entitlements).
While all these papers suggest and study some specific shares, 
we present a comprehensive study of the space of all shares that are feasible, 
and  systematically studied share-based fairness for agents with arbitrary entitlements.

As observed in the literature (e.g. \cite{BF22}), one of the attractive properties of the share-based approach to fairness is that algorithmic results directly imply similar results in strategic settings. Given any algorithm that always outputs an acceptable allocation for some feasible share, we can define a game that has the following two attractive properties: The allocation of every pure Nash equilibrium of the game is an acceptable allocation, and moreover, pure Nash equilibria indeed exist in that game. 
{Other approaches to strategic behaviour include designing truthful mechanisms that output fair allocations for agents with equal entitlements \cite{ABCM2017}, and introducing the notion of ``self-maximizing'' shares, a property of shares that provides incentives to ``pessimistic" agents to report valuation functions truthfully~\cite{BF22}.}


While we focus on share-based fairness, there is a prominent alternative approach to fairness definition in the literature that is ``envy-based'',  aiming to eliminate (or minimize) envy between the agents. Envy-free allocations are allocations in which no agent envies any other agents \cite{varian1973equity}. Unfortunately, for indivisible items such allocations need not exist. 
An extensive body of literature studies relaxations of envy freeness, such as envy-free-up-to-one-good (EF1)~\cite{Budish11,LMMS04} and envy-free up to any good (EFX)~\cite{CaragiannisKMPS19, PR2020, EFX3, amanatidis2020maximum}. 
{The envy-based approach was extended to arbitrary entitlements (``the weighted case'') in  \cite{farhadi2019fair}, studying two variants of weighted envy-freeness up to one item.}

\section{Model and Preliminaries}\label{sec:model}

We consider allocation of a set $\items=[m]$ of $m$ indivisible items to $n$ agents with arbitrary entitlements. 
Each agent $i\in [n]$ has a valuation function $v_i:2^{\items} \rightarrow \mathbb{R}$ over the items, with $v_i(S)$ denoting the value of $S\subseteq \items$ to agent $i$. We assume that all valuation functions are normalized ($v_i(\emptyset)=0$ for every $i\in [n]$) 
{A valuation $v$ is called a \emph{valuation over goods} if it is monotone non-decreasing ($v(S)\leq v(T)$ for $S\subseteq T$). 
We say that $C$ is a \emph{class of valuations over goods} if every $v\in C$ is a valuation over goods.} 
Up to Section \ref{sec:chores} we focus on the allocation of goods (that is, all valuations are over goods), while in Section \ref{sec:chores} we present our results for chores (monotone non-decreasing costs).
Our main results are for additive valuations. 
Valuation  $v_i$ is \emph{additive} if $v_i(S)=\sum_{j\in S} v_i(\{j\})$ (with $v_i(\emptyset)=0$).
An additive valuation over goods satisfies $v_i(\{j\})\geq 0$ for every item $j\in \items$, and we call such items \emph{goods}.

Our focus in this paper is on agents that do not necessarily have equal entitlements to the items, but rather their entitlements might be arbitrary. 
We assume that agent $i\in [n]$ has \emph{entitlement} $b_i> 0$ to the items, and that $\sum_{i\in [n]} b_i=1$. 
The equal entitlement case is the special case in which $b_i=1/n$ for every $i\in [n]$.

An \emph{allocation} $(S_1,S_2,\ldots,S_n)$ is a collection of $n$ disjoint subsets of $\items$ whose union equals $\items$. In this allocation agent $i$ receives the subset $S_i$, for which she has value of $v_i(S_i)$. A \emph{randomized allocation} is a distribution over allocations. A property that holds for an allocation is said to hold ex-post, 
while a property for randomized allocations is said to hold ex-ante.


\subsection{Share-Based Fairness} \label{sec:model-share-based} 
Babaioff and Feige~\cite{BF22} (henceforth BF22)  have presented a general theory of share-based fairness for agents with equal entitlements to the items. Babaioff, Ezra and Feige~\cite{BEF2021b} have studied some specific fairness results for agents with arbitrary entitlements (but did not present a general theory as we do here).  
We next present some background from these papers that is relevant to this paper.

A \emph{share} function (or simply share) $s$ for agents with arbitrary entitlements maps the valuation $v_i$ and an entitlement $b_i$, to the share value $s(v_i,b_i)$ of agent $i$. BF22 have required 
a share to be 
\emph{name-independent} (see Section \ref{sec:model-share-props} below). In this paper we drop this requirement 
and make appropriate adjustments that only strengthen results (as discussed later).
An allocation $A$ is \emph{acceptable} if for every agent $i\in [n]$ it holds that $v_i(A_i)\geq s(v_i,b_i)$, that is, every agent gets a bundle of value at least her share. Similarly, a randomized allocation is acceptable if the expected value for each agent is at least her share (agents are assumed to be risk neutral). 
A share is called an \emph{ex-post} share if it applies to deterministic allocations, and an \emph{ex-ante} share if it applies to randomized allocations.
An ex-post (ex-ante) share is \emph{feasible} for class of valuations $C$ if for every $(v_1,v_2,\ldots, v_n)$ such that $v_i\in C$ for all $i$, it holds that there exist an acceptable ex-post (ex-ante randomized) allocation. 

\subsubsection{Share domination}
Consider any class $C$ of valuations over goods, and any $\rho > 0$. A share function $\tilde{s}$ over goods  $\rho$-dominates (or simply dominates, if $\rho = 1$) share $s$, if for every valuation function $v\in C$ and every entitlement $b$, it holds that $\tilde{s}(v,b) \ge \rho \cdot s(v,b)$. The share $\tilde{s}$ over goods is  $\rho$-dominating for class $C$ if it $\rho$-dominates every  share $s$ that is feasible for $C$. 
BF22 shows that for equal entitlements, a share $\tilde{s}$ over goods is $\rho$-dominating if and only if $\tilde{s}$ $\rho$-dominates the MMS. 
In the next section we present a generalization of this result for entitlements that are not necessarily equal but rather can be arbitrary.  
BF22 have also shown that for equal entitlements, there does not exist a feasible share that dominates every feasible share.
This clearly implies that such a share does not exist also for arbitrary entitlements. 



\subsubsection{Share properties}\label{sec:model-share-props}

{Recall that a share function must be anonymous and entitlement monotone.  We next discuss several natural properties that one might like a share function to additionally satisfy. 
One property, defined in \cite{BF22}, is for a share to be \emph{self maximizing}: 
A share is {\em self-maximizing} if reporting the true valuation maximizes the minimum true value of a bundle that is acceptable with respect to the report (see \cite{BF22} for a formal definition).
Essentially, a self-maximizing share is a property that provides incentives to ``pessimistic" agents to report valuation functions truthfully. BF22 show that any self-maximising share function must be valuation monotone, {valuation} 
continuous (and even $1$-Lipschitz) and scale-free.}

Another share property is name-independence. BF22 
required a share to satisfy that property. We next present this notion and discuss its. 
To define that notion it will be important to distinct between two valuation functions being {\em identical}, and being {\em {equivalent}}.



\begin{definition}
    \label{def:same}
    Given a set $\items$ of items, we say that two valuation functions $v_1$ and $v_2$ are {\em identical} if for every $S \subseteq \items$ it holds that $v_1(S) = v_2(S)$. 
    {We say that $v_1$ and $v_2$ are {\em $\pi$-equivalent} if for permutation $\pi$ over $\items$ it holds that $v_1(S) = v_2(\pi(S))$ for every $S \subseteq \items$ (where $\pi(S)$ is the set of items to which the items of $S$ are mapped under $\pi$). We say that $v_1$ and $v_2$ are {\em equivalent} if they are $\pi$-equivalent for some permutation $\pi$.} 
\end{definition}


\begin{definition}
    \label{def:orderedShare}
    {A share function $s$ is {\em name-independent} if for every two {equivalent valuations} $v$ and $v'$ and for every entitlement $0 < b \le 1$ it holds that $s(v,b) = s(v',b)$.
    An \emph{unrestricted} share (or simply a \emph{share}) is a share that it is either name-independent or not.}
\end{definition}


\begin{definition}
    A share $s$ is \emph{nice} if it is name independent and self maximizing.
\end{definition}

The main theme of our work is to identify (in various contexts) a share $s$ that dominates every feasible share, and then to design allocation algorithms that guarantee for each agent at least some $\rho$ fraction of her share value under $s$. 
{For additive valuations, our positive results will present feasible shares that dominate every unrestricted feasible share. In contrast, our negative result will be proven even when one only asks the dominating feasible share to dominated the restricted subclass of \emph{nice} feasible shares, 
making our impossibility results stronger. }

\subsection{Some Central Shares} \label{sec:model-prior-shares} 


We start by presenting several shares for equal entitlements from prior literature that will be useful in this paper. 
Thus, 
the entitlement $b_i$ of an agent is of the form $\frac{1}{k}$ for some integer $k$ (as would happen in the equal entitlements case, where $k$ would be the number of agents). For the ex-post setting, we recall the maximin share (MMS)~\cite{Budish11}.

\begin{definition}
    \label{def:MMS}
    For positive integer $k$, a $k$-partition $S_1, \ldots, S_k$ of $\items$ is a collection of $k$ disjoint subsets of $\items$ whose union equals $\items$. {Denote the set of all $k$-partitions by ${{\cal S}_k}$.}
    Given a valuation function $v$ and entitlement $b$, where $b = \frac{1}{k}$ for some positive integer $k$, the value of the maximin share (MMS) is:
    
    $$MMS\left(v,\frac{1}{k}\right) = \max_{{(S_1, \ldots, S_k)\in {{\cal S}_k}}} \min_{j \in \{1, \ldots, k\}} v(S_j)$$
    where ${(S_1, \ldots, S_k)\in {{\cal S}_k}}$ ranges over all $k$-partitions of $\items$. A $k$-partition that maximizes the above expression is referred to as an $MMS(v,\frac{1}{k})$ partition. 
\end{definition}

For the ex-ante setting, we consider a share that we refer to as {\em max expected share} (MES). It was previously used in the context of dividing a homogeneous divisible good~\cite{FeigeT14}.

\begin{definition}
    \label{def:MES}
    Given a valuation function $v$ and entitlement $b$, where $b = \frac{1}{k}$ for some positive integer $k$, the value of the max expected share (MES) is:
    
    $$MES\left(v,\frac{1}{k}\right) = \max_{{(S_1, \ldots, S_k)\in {{\cal S}_k}}} E_{j \in \{1, \ldots, k\}} [v(S_j)]$$
    where ${(S_1, \ldots, S_k)\in {{\cal S}_k}}$ ranges over all $k$-partitions of $\items$, and $E[\cdot]$ denotes expectation, taken over a uniformly random choice of bundle among the ones of the chosen $k$-partition. A $k$-partition that maximizes the above expression is referred to as an $MES(v,\frac{1}{k})$ partition. 
\end{definition}

We remark that if $v$ is an additive valuation, then all $k$-partitions give the same expected value, and the MES equals the proportional share, which for entitlement $b$ is defined to be $PS(v,b) = b\cdot v(\items)$. Likewise, for the equal entitlements case with $k$ agents, $PS(v,\frac{1}{k} ) = \frac{1}{k}\cdot v(\items)$.

{Another share that we use in our proofs is the Truncated Proportional Share (TPS).}
The Truncated Proportional Share (TPS) was defined for equal entitlements in \cite{BEF2021c} and extended to arbitrary entitlements in \cite{BEF2021b}. Essentially, it is the proportional share after item values are minimally truncated such that no single item will have value above the post-truncation proportional share. 
\begin{definition}\label{def:TPS}
	The \emph{Truncated Proportional Share (TPS)} of agent $i$ with an additive valuation $v_i$ over the set of items $\items$ and entitlement $b_i$, denoted by $\truncatedi$, is:
	\begin{equation}
	\truncatedi = \max\{~z \mid b_i \cdot \sum_{j\in \items} \min(v_i(j),z) = z ~\}. \label{eq:tps}    
	\end{equation}
\end{definition}


As usual, for a share $Share$ and a constant $\alpha$, we use $\alpha-Share$ to denote the share $Share$ scaled by a factor of $\alpha$. For example, $1/2$-MMS denotes the share that is half the MMS.

\section{Share Domination for Valuations over Goods}\label{sec:domination}


In this section we present necessary and sufficient conditions for a share to $\rho$-dominate every feasible share for agents with valuations over goods and arbitrary entitlements. 
Our characterization can be viewed as a reduction from the arbitrary entitlements case to the equal entitlements case.   
{For additive valuations, this characterization will serve as a starting point for our results that determine} the largest $\rho$ for which it is possible to $\rho$-dominate every feasible share, {both in the ex-post case and in the ex-ante case.} 




{To illustrate the issue of share domination for arbitrary entitlements we present a simple example. Consider an agent $i$ with entitlement of $b_i = 0.3$ and valuation $v_i$. What may serve as an upper bound on the value $s(v_i,b_i)$ that holds for every feasible share $s$? There might be $\lfloor \frac{1}{0.3} \rfloor = 3$ agents (including agent $i$) with the same budget and valuation, in which case symmetry dictates that all three agents have the same share value. Feasibility dictates that there should be an allocation that gives this share value to each of the three agents, and so $s(v_i,b_i) \le MMS(v_i,\frac{1}{3})$. This idea is generalized in our characterization below.}   


\begin{definition}\label{def:unit}
    For entitlement $0 < b \le 1$, let $k = \lfloor \frac{1}{b} \rfloor$. 
    We define the {\em unit upper bound} of $b$ to be $\frac{1}{k}$, and denote it by $\hat{b}$. 
    For any share $s$ for equal entitlements, we define the share $\hat{s}$ to be $\hat{s}(v_i,b_i) = s(v_i,\hat{b}_i)$.
    In particular, $\widehat{MMS}(v_i,b_i) = MMS(v_i,\hat{b}_i)$, $\widehat{MES}(v_i,b_i) = MES(v_i,\hat{b}_i)$ and
    $\widehat{PS}(v_i,b_i) = PS(v_i,\hat{b}_i)$.
\end{definition}
Observe that $k = \lfloor \frac{1}{b} \rfloor$ is the unique integer $k$ such that $\frac{1}{k+1}< b \le \frac{1}{k}$.

\subsection{General Valuations over Goods}\label{sec:domination-general}
We present our characterizations for ex-post shares and for ex-ante shares.

\subsubsection{Ex-post Domination}
We start with our characterization for ex-post shares.

\begin{proposition}\label{prop:rounding-MMS-goods}
    
    For any class $C$ of valuations over goods,  $\widehat{MMS}$ dominates every feasible unrestricted ex-post share for arbitrary entitlements  (for class $C$). Moreover, it is the minimal share for arbitrary entitlements that dominates every feasible unrestricted ex-post share (for class $C$). 
\end{proposition}
\begin{proof} 
We first prove that $\widehat{MMS}$ dominates every feasible unrestricted share for arbitrary entitlements (for class $C$). To prove this, we need to show that for any feasible unrestricted share $s'$ it holds that $\widehat{MMS}(v,b)\geq s'(v,b)$ for any valuation $v\in C$ and any entitlement $b$. 
Assume in contradiction that for a feasible unrestricted share $s'$ there exist a valuation $v\in C$ and an entitlement $b$ such that $\widehat{MMS}(v,b)< s'(v,b)$.
Let $k=({\hat{b}})^{-1}$, for $\hat{b}$ which is the  unit upper bound on $b$. Consider an instance with $k$ agents having an entitlement $b$, and one agent with entitlement $1-b\cdot k<b$, and assume all have the same valuation $v$.
As the share $s'$ is feasible, there must be an allocation that gives each one of the $k$ agents at least $s'(v,b)> \widehat{MMS}(v,b) = {MMS}(v,1/k)$. That is, there is an allocation that gives each of $k$ agents more than ${MMS}(v,1/k)$. 
But this leads to a contradiction as by the definition of the MMS there is no acceptable allocation that  gives all $k$ agents (of the same entitlement) more than ${MMS}(v,1/k)$.

 Next, we show that $\widehat{MMS}$ is the minimal share  that dominates every feasible unrestricted share (for class $C$). 
 That is, we show that a share that dominates every feasible unrestricted share must dominate $\widehat{MMS}$ . 
 Let $\bar{s}$ be a share that dominates every feasible unrestricted share for arbitrary entitlements. 
 Assume in contradiction that $\bar{s}$ does not dominate $\widehat{MMS}$ . Then there exist a valuation $v\in C$ and an entitlement $b$ such that $\widehat{MMS}(v,b)> \bar{s}(v,b)$. By Proposition \ref{pro:verypersonalMMS}, the personalized MMS share $MMS_{v,b}$ is feasible for $C$ (ex-post). For this share it holds that $MMS_{v,b}(v,b) = \widehat{MMS}(v,b)$. As $\widehat{MMS}(v,b)> \bar{s}(v,b)$, the share $\bar{s}$ does not dominate the feasible share  $MMS_{v,b}$, a contradiction. 
 \end{proof}

\OLD{
For each valuation $v'\in C$, let $s_{v'}$ be the share for equally-entitled agents that is defined as follows. For every $k$, define $s_{v'}(v',1/k)=s(v',1/k)$, and define $s_{v'}(v,1/k)=0$ for $v\neq v'$ (observe that this constructs a share that is unrestricted). 
The share $s_{v'}$ is feasible, as $s$ is tightly-feasible for identical valuations, so for every $k$ there is an acceptable allocation when there are at most $k$ identically-entitled agents with valuation $v'$ (for any agent with a different valuation it is acceptable to get nothing according to the share $s_{v'}$).
As $s_{v'}$ is feasible for every $v'\in C$, to dominate every feasible share the share $\bar{s}$ must dominate $s_{v'}$ for every $v'\in C$.  Thus, for every $b$ with unit upper bound for $\hat{b}$ satisfying $k=({\hat{b}})^{-1}$, and for  $v\in C$ we have    
$\bar{s}(v,b)\geq s(v,1/k)= \hat{s}(v,b)$, a contradiction.
\mbc{finished. But we do not use tightness here, I find this strange. Is the proof wrong?} \ufc{The lemma assumes that $s$ is tightly feasible  for identical valuations, and that it dominates ever feasible share. These two things imply that $s$ is minimal. Hence, the assumption that $s$ is minimal is redundant, and thus need not be used in the proof.} \mbc{I have revised the proof. Again, it would be good to fully understand and make explicit how we are using the fact that the valuations are over goods.}
} 

\begin{remark}
\label{rem:MMSminimality}
Proposition~\ref{prop:rounding-MMS-goods} considers domination of unrestricted ex-post shares. Its conclusion that $\widehat{MMS}$ dominates unrestricted shares clearly extends also to dominating only the name-independent shares. However, the claim of minimality makes essential use of unrestricted shares.  For name-independent shares, minimality of $\widehat{MMS}$ holds for the class of additive valuations, but there are non-additive classes of valuations for which it does not hold (even when entitlements are equal). 
See Proposition \ref{prop:goods-additive-minimal} and Observation~\ref{obs:MMS-non-add}, respectively. 
\end{remark}


\subsubsection{Ex-ante Domination}
We continue with our characterization for ex-ante shares.

\begin{proposition}\label{prop:rounding-MES-goods}  
    For any class $C$ of valuations over goods,  $\widehat{MES}$ dominates every feasible unrestricted ex-ante share for arbitrary entitlements  (for class $C$). Moreover, it is the minimal ex-ante share for arbitrary entitlements that dominates every feasible unrestricted ex-ante share (for class $C$).
\end{proposition}

We omit the proof as it is essentially the same as the proof of Proposition \ref{prop:rounding-MMS-goods}. {Likewise, Remark~\ref{rem:MMSminimality} applies as is with $\widehat{MMS}$ replaced by $\widehat{MES}$.}


\subsection{Additive Valuations over Goods}
\label{sec:domination-additive}

We next focus on the important special case that the valuations over goods are additive. For such valuations, the  share $\widehat{MES}$ is simply the share $\widehat{PS}$. For additive valuations, we strengthen our results from Proposition \ref{prop:rounding-MMS-goods} and  Proposition \ref{prop:rounding-MES-goods} and show that  $\widehat{MMS}$ and $\widehat{PS}$ {remain the minimal dominating shares  even if domination is only with respect to} every feasible {\emph{nice} ({name-independent} and self-maximizing)} share (ex-post and ex-ante, respectively).    

\begin{proposition}\label{prop:goods-additive-minimal}
For the class of additive valuations over goods and arbitrary entitlements, the share $\widehat{MMS}$ is the minimal ex-post share  that dominates every feasible {\emph{nice}} 
ex-post share, and the share $\widehat{PS}$ is the minimal ex-ante share  that dominates every feasible  {\emph{nice}}
ex-ante share.
\end{proposition}
\begin{proof}
The fact that $\widehat{MMS}$ and $\widehat{PS}$ are dominating for ex-post and ex-ante shares respectively, was proved in Propositions~\ref{prop:rounding-MMS-goods} and~\ref{prop:rounding-MES-goods}, and this holds even with respect to unrestricted shares. It remains to prove minimality.

We first show that for the class of additive valuations over goods and arbitrary entitlements, the share $\widehat{MMS}$ is the minimal ex-post share that dominates every feasible {nice} 
ex-post share. 
That is, we show that a share $\bar{s}$ that dominates every feasible {nice} 
ex-post share must dominate $\widehat{MMS}$. 
 Assume in contradiction that $\bar{s}$ does not dominate $\widehat{MMS}$. Then there exists an additive valuation over goods $v$ and an entitlement $b$ such that ${MMS}(v,\hat{b})= \widehat{MMS}(v,b)> \bar{s}(v,b)$. 
 {Consider the share $OMMS_{(v,b)}$ as defined in Appendix \ref{sec:nice}. 
 \mbfuture{forward pointer, do we need to define it here?}
 By Corollary \ref{cor:nice}, $OMMS_{(v,b)}$ is a nice (name independent and self maximizing) ex-post share that is feasible for additive goods, and satisfies $OMMS_{(v,b)}(v,b) = \widehat{MMS}(v,b)$.
 As share $\bar{s}$ dominates every nice feasible share, it must holds that $\bar{s}(v,b)\geq \widehat{MMS}(v,b)$, a contradiction.  }

{The proof of the claim that for additive valuations over goods the share $\widehat{PS}$ is the minimal ex-ante share  that dominates every feasible nice ex-ante share is very similar to the proof for ex-post share presented above, replacing the personalized share $OMMS_{(v,b)}$ with $PS_{(v,b)}$. As the proof is a direct adaptation of the one above, it is omitted. }
\end{proof}

In the rest of the paper we assume that the valuations are additive. We first study items that are goods, and in Section \ref{sec:chores} move to consider chores. 

\section{The Bidding Game and Main Positive Result}\label{sec:bidding}

{In this section we obtain our main result showing that for every allocation instance with additive valuations {over goods}, there are allocations that give each agent at least $\frac{1}{2} \widehat{MMS}$, {and thus at least half of every feasible ex-post share}. In fact, our allocations provide a somewhat better guarantee, $\frac{1}{2} \widehat{TPS}$. The TPS is defined in Definition~\ref{def:TPS}, and recall that for share $s$ for equal entitlements, we define the share $\hat{s}$ to be $\hat{s}(v_i,b_i) = s(v_i,\hat{b}_i)$.
So $\widehat{TPS}(v_i,b_i) = TPS(v_i,\hat{b}_i)$. It is known that for additive valuations and entitlements of the form $\frac{1}{k}$, the TPS is at least as large as the MMS.} 

{Our allocations are obtained by designing strategies for a bidding game that was introduced in~\cite{BEF2021b}.}
We consider {the simplest}  
variant of the game, which in this paper we refer to as \emph{the bidding game}. 
The bidding game we consider is a game in which each agent gets a budget equal to her entitlement, and agents bid in $m$ rounds for the right to pick one of the remaining items {of their choice}, and the highest bidder wins and pays her bid. We first formalize the game (Section \ref{sec:goods-bidding-game}). 
A \emph{worst-case-optimal} (or ``optimal'', for short) strategy for agent $i$ with valuation $v_i$ and entitlement $b_i$ in the bidding game is a strategy that maximizes the agent value in a $0$-sum game against an {\em adversary} with entitlement $1 - b_i$ that aims to minimize the value of the agent (when all ties in bids are broken by the adversary).
We study the properties of worst-case-optimal strategies in the game (Section \ref{sec:OptSafeBiddingGame}). Worst-case-optimal strategies are not poly-time computable, so we discuss approximately worst-case-optimal strategies (Section \ref{sec:approx-optimal}). {We then use the optimal-worst-case strategies as a component in new strategies that we develop, and using our strategies we establish our tight result for ex-post domination (Section \ref{sec:goods-pos-ex-post}).}   

\subsection{Background and Our Results}\label{sec:background}


\cite{BEF2021b} have presented a bidding game and have derived results for some specific shares. It
was shown that any agent with an additive valuation function has a simple strategy that ensures that she receives a bundle of value at least half her TPS, 
no matter what the strategies of the other agents are. Moreover, for a slight variation on the bidding game (allowing the winner of a round to select multiple items and pay for each of them), it was shown that any agent $i$ with an additive valuation has strategies that gives her at least a $\frac{1}{2 - b_i}$ fraction of the TPS, and another strategy that gives her at least a $\frac{3}{5}$-fraction of her APS {(see Definition~\ref{def:APS})}.
This latter result implies in particular that if all agents have additive valuations, there is an allocation that gives each agent at least a $\frac{3}{5}$-fraction of her APS. There is no other known proof of this fact. The bidding game was used in~\cite{bUF23} in order to provide the strongest known approximations when allocating items to agents with submodular valuations, to the MMS in the case of equal entitlements and to the APS in the case of arbitrary entitlements.

Our proof of existence of a $\frac{1}{2}$-$\widehat{TPS}$ allocation
for agents with additive valuations is by designing strategies that achieve such a guarantee in the bidding game. Giving the effectiveness of the bidding game in achieving previous bounds, it is not surprising that we use this methodology as well. However, it was far from obvious that there are strategies for the bidding game that provide a $\frac{1}{2}$-$\widehat{TPS}$ allocation. It is instructive to compare our strategies with the previous strategy of~\cite{BEF2021b} that achieves $\frac{1}{2}$-${TPS}$ (but not $\frac{1}{2}$-$\widehat{TPS}$). That strategy was simple in several respects:

\begin{itemize}
    \item It was {\em myopic}: in any given round of the bidding game, the value of a bid depends on the value of the item that the agent would select at that particular round if she was to win the round, but not on values of items that will be selected in future rounds.
    \item It was {\em monotone}: the bids of the agent form a nonincreasing sequence as the round number grows.
    \item It was polynomial time computable, or even computable in constant time (the bid in round $r$ is simply the minimum between the budget that the agent has and a scaled version of the value of the most valuable item that remains).
\end{itemize}

In contrast, we show that no myopic strategy can guarantee $\frac{1}{2}$-$\widehat{TPS}$. See Section~\ref{sec:myopic}. 
Moreover, we show that no monotone strategy can guarantee $\frac{1}{2}$-$\widehat{TPS}$. See Proposition~\ref{prop:nonMonotone}. Finally, we do not known if there is a polynomial time strategy that achieves $\frac{1}{2}$-$\widehat{TPS}$. 
We compensate for this last aspect by designing for every constant $\varepsilon > 0$ a polynomial time strategy that guarantees $(\frac{1}{2} - \varepsilon)$-$\widehat{TPS}$.

In developing our strategy we also provide a characterization of optimal strategies for the bidding game (for agents with additive valuations). This allows us to infer previously unknown facts about this game, such as that computing an optimal strategy is NP-hard. (We leave open the question of the exact complexity of this problem, as we have no reason to believe that computing an optimal strategy is in NP.) More importantly for our main result, we establish (via a nonconstructive argument), that if the entitlement of an agent is strictly larger than $\frac{1}{2}$, then the agent has a strategy that guarantees at least $\frac{1}{2} \cdot v_i(\items)$. This is one of the ingredients of our strategy that achieves $\frac{1}{2}$-$\widehat{TPS}$, and this is the only ingredient in our strategy that we do not know how to implement in polynomial time.

\subsection{Formalizing the Bidding Game}\label{sec:goods-bidding-game}

The {\em bidding game} is an allocation mechanism that was introduced in~\cite{BEF2021b}. It is a game played by the $n$ agents, and resulting in an allocation of $\items$. Initially, each agent $i$ gets a budget $b_i^0$ equal to her entitlement $b_i$, and holds an empty bundle of items $A_i^0$.  Thereafter, the game proceeds in rounds, where in each round, one item gets allocated, and the agent receiving the item pays part of her budget for this item. We now describe the execution of a single round, say round $r$, where $1 \le r \le m$. 

At the beginning of round $r$, each agent $i$ holds a budget $b_i^{r-1} \ge 0$, the part of the budget that she had not spent in the first $r-1$ rounds, and a set $A_i^{r-1}$ of those items that were allocated to her in previous rounds.  There also is a set $\items^{r-1}$ that contains $m - (r-1)$ items, those items that have not been allocated in the first $r-1$ rounds. In round $r$, each agent $i$ submits a non-negative bid $p_i^r$ of her choice, subject to the constraint that $p_i^r \le b_i^{r-1}$. For concreteness, we assume that the bids are submitted sequentially, one agent after the other, in some arbitrary order. The agent submitting the highest bid 
{(our results will apply to any tie breaking rule, and hence we do not specify an explicit tie breaking rule)} 
is declared as the winner of the round. The winner, say, agent $i$, selects an item $e_r$ of her choice from $\items^{r-1}$.  The set of items received by $i$ becomes $A_i^r = A_i^{r-1} \cup \{e_r\}$, and the budget of $i$ is updated to $b_i^r = b_i^{r-1} - p_i^r$ (the agent pays her bid). All other agents do not receive an item in round $r$ and do not pay anything in round $r$. That it, for $j \not= i$, we have that $A_j^r = A_j^{r-1}$, and $b_j^r = b_j^{r-1}$. We also update $\items^r = \items^{r-1} \setminus \{e_r\}$, and the round ends.

The bidding game ends after $m$ rounds, at which point all items are allocated.

A strategy for an agent $i$ in the bidding game is a function that at each round $r$ takes as input the set $\items$, the vector $(b_1, \ldots, b_n)$ of entitlements, the valuation $v_i$ of the agent, and the relevant history of the game up to round $r$ {(this history includes the bids in every prior round, the winner of each of these rounds, and the item that the winner selected when paying her bid)}
and produces two outputs. One is a bid $p_i^r$ in round $r$. The other is an item $e_r \in \items^{r-1}$ that the agent will choose if she wins the round. (More generally, the choice of $e_r$ may depend also on the values of the losing bids, but such generality is not needed for our purpose.)  

\subsection{Worst-case-optimal Strategies for the Bidding Game}
\label{sec:OptSafeBiddingGame}

\mbfuture{we have discussed a $\varepsilon$-SM share based on optimal strategies, and getting a positive result of $\varepsilon$-SM using the PTAS. The loss will be additive - case we handle small entitlements to get a multiplicative bound?}

In this section we perform a systematic study of {worst-case-optimal} 
strategies for the bidding game. Consider a set $\items$ of items, and an agent $i$ with valuation $v_i$ and entitlement $b_i$. 
{We look for optimal (max-min) strategies in the zero-sum game in which the agent aims to maximizes her value against an adversary that holds an entitlement of $1-b_i$ (and controls tie breaking).}
There are two players, $P_1$ and $P_2$. {The reader may think of $P_1$ as representing an agent $i$ of interest, and of $P_2$ as an adversary that coordinates bids of all other agents in an attempt to minimize the value that agent $i$ receives.} {The adversary becomes most powerful if there is just one other agent that holds all the remaining budget, instead of the budget being spread across different agents, because then the adversary's 
{bidding strategy is least constrained}.} 
There is a known additive valuation function $v$ with non-negative values {($v$ is the valuation function $v_i$ of agent $i$)}, and items are named in non-increasing order of value ($v(e_1) \ge v(e_2) \ge \ldots \ge v(e_m)$). The players have non-negative budgets $b_1$ and $b_2$ with $b_1 + b_2 = 1$. In every round $r$, $P_1$ bids first and $P_2$ bids next {(after seeing the bid of $P_1$)}, where a bid cannot exceed the budget available to the bidder. The winner of a round is the player placing the higher bid, breaking ties in favor of $P_2$. The winner {of round $r$ receives item} $e_r$ and pays her bid. The payoff for $P_1$ is the sum of $v$ values of items that $P_1$ wins, whereas the payoff for $P_2$ is the negative of that of $P_1$. Equivalently (up to {shifting} by $v(\items)$), the payoff of $P_2$ is  the sum of $v$ values of items that $P_2$ wins. 

Observe that without loss of generality, {as in every round the adversary ($P_2$) bids after seeing the bid of $P_1$,} the strategy space for $P_2$ need not be that of bids, but rather that of bits. The bit~0 can be interpreted as allowing $P_1$ to win the round, whereas the bit~1 can be interpreted as placing a bid that makes $P_2$ win the round. Being a 0-sum game and the convention that ties are broken in favor of $P_2$ imply that in this latter case, $P_2$ makes the same bid as $P_1$, and pays the bid of $P_1$. 
Hence, in every round $r$, first $P_1$ places a bid (not higher than her available budget), and 
then, $P_2$ decides who wins the round ({unless her remaining budget is lower than that bid}). 
The winner pays the bid of $P_1$, and gets the item $e_r$.

{Being a zero-sum deterministic sequential game, the worst-case-optimal strategy for the agent is deterministic.}
The max-min value of the game is the maximum payoff that $P_1$ can guarantee to herself, when $P_2$ is using a strategy that attempts to minimize this payoff. 

Under $v$, two different bundles may have the same value. It will be convenient for us to enforce that this never happens. 
{We use $v'$ to denote {an additive valuation function in which there are no ties in values of bundles,} that is consistent with the partial order over subsets induced by $v$ and the following tie breaking rule: ties are broken lexicographically according to the indexes of the items. We call such a $v'$ an \emph{associated valuation} of $v$, and observe that the order over subsets is the same for any $v'$ associated with $v$. 
For the special case that $v$ is integer-valued we consider the following concrete valuation $v'$: we let $v'$ be defined as}
$v'(e_r) = v(e_r) + 2^{-r}$ for every $r \in [m]$. Note that $v'$ differs from $v$ by adding a fractional part to otherwise integer values, it preserves the order among items (though we now have strict inequalities, $v'(e_1) > v'(e_2) > \ldots > v'(e_m)$), and the sum of all fractions added to a bundle is strictly less than~1. No two bundles can have the same value under $v'$, {and it preserves the order among bundles:} for every two bundles $S$ and $T$, if $v(S) > v(T)$ then $v'(S) > v'(T)$. We may assume that payoffs for the players are computed according to $v'$ rather than according to $v$, as optimal strategies for $v'$ are also optimal for $v$.

\subsubsection{Worst-case-optimal strategies: a structural result}

{We present a structural result, describing how 
the value guaranteed by the worst-case-optimal strategy changes as one increases the fraction of total entitlements that the agent has.}

The valuation function $v'$ induces a strict total order on bundles (with no ties), based on increasing $v'$ value. We refer to this order as $\pi_v$ (with $v'$ being implicit from $v$ and the item names). For $j \in [2^m]$, we let $\pi_v(j)$ denote the bundle that is at location $j$ in this total order, and say that $j$ is the {\em rank} of this bundle. In particular, $\pi_v(1)$ is the empty bundle, and $\pi_v(2^m) = \items$. More generally, for every $j$, the bundles at location $\pi_v(j)$ and $\pi_v(2^m - j + 1)$ are complements of each other (they are disjoint and their union is $\items$). 

With the vector $\pi_v$ we associate another vector $T_v$ with $2^m$ coordinates referred to as the {\em threshold vector}. 
Our intention is to fill up the entries of the vector $T_v$ with ``budget quantities" that satisfy the following threshold property:

\begin{itemize}
    \item {$T_v(1) = 0$.} For every {$2 \le j \le 2^m$}, $P_1$ has a strategy that guarantees that she receives value at least $v'(\pi_v(j))$ (namely, $P_1$ can guarantee to herself a bundle at least as good as the bundle of rank $j$ in the order $\pi_v$) if and only if her budget $b_1$ satisfies  $b_1 > T_v(j)$. {(Note that for $j=1$ this strict inequality is replaced by a weak inequality, requiring us to define $T_v(1)$ separately.)} 
\end{itemize}

As a warm up, we determine some entries of the vector $T_v$. 

The bundle of rank~1 is necessarily the empty bundle. Any budget of $P_1$ suffices in order to win at least this bundle, so $T_v(1) = 0$. Note that in this case, $P_1$ can win the bundle of rank~1 even if  $b_1 = T_v(1)$ (the strict inequality $b_1 > T_v(1)$ can be made into a weak inequality). This is the only exception to the if and only if condition in the threshold property. 

The bundle of rank~2 is $\{e_m\}$ {(recall that $e_m$ is the least valuable item)}. For $P_1$ to be guarantee winning this item (or better), we need $b_1 > \frac{1}{m+1}$. If $P_1$ bids her full budget in each round (until winning a round), she is guaranteed to win one item. With $b_1 \le \frac{1}{m+1}$, $P_2$ can afford to bid $\frac{1}{m+1}$ in every round, and $P_1$ does not win any item. Hence $T_v(2) = \frac{1}{m+1}$.

The bundle of rank~$2^m$ is $\items$. For $P_1$ to be guarantee winning all items, we need $b_1 > \frac{m}{m+1}$. If $P_1$ bids $\frac{1}{m+1}$ in each round, she is guaranteed to win all rounds. With $b_1 \le \frac{m}{m+1}$, $P_2$ can afford to bid $\frac{1}{m+1}$ in every round, and $P_2$ wins at least one item. Hence $T_v(2^m) = \frac{m}{m+1}$.

We see that the cases of rank~2 and rank $2^m$ are complements of each other, with the optimal strategy for $P_1$ in one case being the optimal strategy for $P_2$ in the other case, and $T_v(2) + T_v(2^m) = 1$. This is not a coincidence, and it stems from the fact that the bundle at rank $2^m$ is the complement of the bundle at rank $2 - 1 = 1$. More generally, as we shall see, we will have the equality $T_v(j) + T_v(2^m - j + 2) = 1$ for all $2 \le j \le 2^m$.

\begin{theorem}
    \label{thm:budgetVector}    
    Let $v$ be a non-negative additive valuation function, with its associated total order $\pi_v$ 
    {(with respect to an associated valuation function $v'$ that induces a total order
    )} 
    over all bundles. 
    Then there is a unique threshold vector $T_v$ with the following properties.
    \begin{enumerate}
    \item {\em Threshold property.} For every $j \in [2^m]$, $P_1$ has a strategy that guarantees that she receives value at least $v'(\pi_v(j))$, if and only if her budget $b_1$ satisfies  $b_1 > T_v(j)$.  (There is one exception. For $j=1$, a budget of $b_1 \ge T_v(1)=0$ suffices in order to get $v'(\pi_v(1))=0$.)
    \item {\em Strict monotonicity.} $T_v(j+1) > T_v(j)$ for every  $j \in \{1, 2,\ldots ,2^{m}-1\}$.
        \item {\em Symmetry.} For every $j \in \{2, 3,\ldots ,2^{m}\}$ it holds that $T_v(j) + T_v(2^m - j+2) = 1$. In particular, $T_v(2^{m-1} + 1) = \frac{1}{2}$.
        \item {\em Rationality {and representation bound:}} 
        For every $j\in [2^m]$ the threshold $T_v(j)$ is a rational number, with numerator and denominator both upper bounded by $2^{2^m - 1}$.   
    \end{enumerate}
\end{theorem}

    \begin{proof} 
    The proof is by induction on $m$. The base case is $m=1$. In this case, $\pi_v$ contains two bundles, the empty one, and the one containing the single item. The vector $T_v$ has $2^m=2^1=2$ coordinates, and its value is $(0, \frac{1}{2})$ (recall that prior to the statement of Theorem~\ref{thm:budgetVector} we showed that $T_v(1) = 0$ and $T_v(2) = \frac{1}{m+1}$). One can readily see that this vector satisfies all properties of Theorem~\ref{thm:budgetVector}. 

    For the inductive step, suppose that $m>1$ and that the theorem is known to hold for $m' = m-1$. For an instance $I$ with $m$ items, let $I'$ be the instance on the set $\items' = \items \setminus \{e_1\}$ of items, and let $\pi'_v$ be the corresponding order on its $2^{m-1}$ bundles. By the inductive hypothesis, instance $I'$ has a threshold vector $T'_v$ (with $2^{m-1}$ coordinates) satisfying the properties of the theorem. We show now how to use $T'_v$ of $I'$ in order to construct $T_v$ for $I$.
    \medskip

Consider first 
an index $j$ such that the bundle $B_j$ of rank $j$ in $\pi_v$ does not contain $e_1$. We explain how to set the value of $T_v(j)$. We assume that $j > 1$, because for $j=1$ the respective bundle is the empty bundle, and $T_v(1) = 0$.

Let $h$ denote the rank of $B_j$ in $\pi'_v$. 
Let $B' \subset \items'$ be a bundle not containing $e_1$ for which the rank of $B' \cup \{e_1\}$ in $\pi_v$ 
is above $j$, and which has the lowest rank in $\pi'_v$ among all such bundles. Denote the rank of $B'$ in $\pi'_v$ by $\ell$. Note that necessarily $\ell < h$, because removing any item from $B_j$ and adding $e_1$ to the resulting bundle necessarily gives a bundle ranked above $B_j$ in $\pi_v$. 

Suppose that $P_1$ has budget $b_1$ and wishes to ensure that she wins a bundle of rank at least $j$ under $\pi_v$. What should her first bid $\alpha$ satisfy? If she loses the first bid, she is in the situation that she needs to win a bundle of rank at least $h$ in $\pi'_v$. To meet this need it is required that $\frac{b_1}{1 - \alpha} > T'_v(h)$, giving $\alpha > \frac{T'_v(h) - b_1}{T'_v(h)}$.   If she wins the first bid, she is in the situation that she needs to win a bundle of rank at least $\ell$ in $\pi'_v$. Hence, $\frac{b_1 - \alpha}{1 - \alpha} > T'_v(\ell)$, giving $\alpha < \frac{b_1 - T'_v(\ell)}{1 - T'_v(\ell)}$. The inequalities that we derived imply that $\frac{T'_v(h) - b_1}{T'_v(h)} < \frac{b_1 - T'_v(\ell)}{1 - T'_v(\ell)}$, giving $b_1 > \frac{T'_v(h)}{1 + T'_v(h) - T'_v(\ell)}$. Hence $T_v(j) = \frac{T'_v(h)}{1 + T'_v(h) - T'_v(\ell)}$.

Let us check now the rationality property. Writing $T'_v(h)$ as the reduced rational number $\frac{p_h}{q_h}$, and $T'_v(\ell)$ as $\frac{p_{\ell}}{q_{\ell}}$, we get that $T_v(j) = \frac{p_h q_{\ell}}{q_{\ell} q_h + p_h q_{\ell} - p_{\ell}q_h}$. By the inductive hypothesis, each term $p_h, q_h, p_{\ell}, q_{\ell}$ is at most $2^{2^{m-1} - 1}$, which readily implies that both numerator and denominator of $T_v(j)$ are bounded by $2^{2^m - 1}$.
\medskip

\mbfuture{STOPPED HERE}

Consider now an index $j$ such that the bundle $B_j$ of rank $j$ in $\pi_v$ contains $e_1$. We decompose this bundle as $B_j = B \cup \{e_1\}$. We explain how to set the value of $T_v(j)$. 

Let $\ell$ denote the rank of $B$ in $\pi'_v$. 
Let $B' \subset \items'$ be a bundle not containing $e_1$ for which the rank of $B' \cup \{e_1\}$ in $\pi_v$ is above $j$, and which has the lowest rank in $\pi'_v$ among all such bundles. Such a bundle $B'$ need not exist (for example, it might hold that $v'(e_1) > v'(\items')$). If $B'$ does exist, denote its rank in $\pi'_v$ by $h$. Note that necessarily $h > \ell$.

Suppose that $P_1$ has budget $b_1$ and wishes to ensure that she wins a bundle of rank at least $j$ under $\pi_v$. What should her first bid $\alpha$ satisfy? There are two cases to consider, depending on whether $B'$ exists or not.

If $B'$ does not exist, $P_1$ is in the situation that she needs to win $e_1$, and later also win a bundle of rank at least $\ell$ in $\pi'_v$. To win $e_1$, $P_1$ must bid more than $1 - b_1$ in the first round. Thereafter, the budget that she has left (just below $2b_1 - 1$) needs to satisfy $\frac{2b_1 - 1}{1 - (1 - b_1)} > T'_v(\ell)$, giving $b_1 > \frac{1}{2 - T'_v(\ell)}$, and $T_v(j) = \frac{1}{2 - T'_v(\ell)}$.   Writing $T'_v(\ell)$ as $\frac{p_{\ell}}{q_{\ell}}$, we get that $T_v(j) = \frac{q_{\ell}}{2 q_{\ell} - p_{\ell}}$. By the inductive hypothesis, each term $p_{\ell}, q_{\ell}$ is at most $2^{2^{m-1} - 1}$, implying that both numerator and denominator of $T_v(j)$ are bounded by $2^{2^m - 1}$.
\medskip

If $B'$ does exist, what should $P_1$'s first bid $\alpha$ satisfy? If $P_1$ looses the first bid, she is in the situation that she needs to win a bundle of rank at least $h$ in $\pi'_v$. If she wins the first bid, she is in the situation that she needs to win a bundle of rank at least $\ell$ in $\pi'_v$. We already saw that these considerations lead to $T_v(j) = \frac{T'_v(h)}{1 + T'_v(h) - T'_v(\ell)}$, and that both numerator and denominator of $T_v(j)$ are bounded by $2^{2^m - 1}$.

By the above, we designed a vector $T_v$ satisfying the threshold property, and showed that is satisfies the rationality property. We now show that strict monotonicity holds. 

\begin{claim}\label{claim:threshold-strict}
    The threshold vector $T_v$ satisfies the strict monotonicity property. 
\end{claim}

\begin{proof}
Observe that we established a formula for $T_v(j)$, namely, $T_v(j) = \frac{T'_v(h)}{1 + T'_v(h) - T'_v(\ell)}$, with the following interpretations of $h$ and $\ell$.

\begin{enumerate}
    \item If $e_1$ is not in the bundle $B_j$ of $\pi_v$ rank $j$, then $h$ is the $\pi'_v$ rank of this bundle $B_j$, and $\ell$ is the $\pi'_v$ rank of the lowest bundle $B' \subset \items'$ whose union with $e_1$ gives a bundle of $\pi_v$ rank above $j$.
    \item If $e_1$ is in the bundle $B_j$ of $\pi_v$ rank $j$, we decompose $B_j = B \cup \{e_1\}$. Then, $\ell$ is the $\pi'_v$ rank of $B$, and $h$ is the $\pi'_v$ rank of the lowest bundle $B' \subset \items'$ of $\pi_v$ rank above $j$. 
    \item If in case~2 above there is no such $B'$, then the two occurrences of $T'_v(h)$ in the expression $T_v(j) = \frac{T'_v(h)}{1 + T'_v(h) - T'_v(\ell)}$ are each replaced by~1, giving $T_v(j) = \frac{1}{2 - T'_v(\ell)}$. Note that $1 > T'_v(i)$ for every {$i \in \{1, \ldots, 2^{m-1}\}$,} and hence we may think of this case as corresponding to taking $h = 2^{m-1}+1$ (larger than the number of coordinates in $T'_v$) and keep the expression $T_v(j) = \frac{T'_v(h)}{1 + T'_v(h) - T'_v(\ell)}$.
\end{enumerate}

Suppose without loss of generality that $i > j$. Let $h_i, h_j, \ell_i, \ell_j$ denote the corresponding $h$ and $\ell$ values for $i$ and $j$. To show strict monotonicity of $T_v$ it suffices to show that $h_i \ge h_j$, $\ell_i \ge \ell_j$, and that at least one of these inequalities is strict. Verifying monotonicity within each of the three cases separately is straightforward. Likewise, verifying monotonicity when $i$ is in case~3 and $j$ is in case~2 is straightforward. Note also that it cannot be that $i$ is in case~2 and $j$ is in case~3, because $i > j$.

Suppose that $i$ is in case~1 whereas $j$ is in case~2. Note that $B_i$ can serve as $B'$ (showing in particular that $j$ cannot be in case~3). Hence $h_i \ge h_j$. Moreover, $\ell_i > \ell_j$, because adding $e_1$ to $B'$ of $i$ leads to a bundle of $\pi_v$ rank higher than $i$, whereas adding $e_1$ to $B$ of $j$ leads to $B_j$ which has has $\pi_v$ rank $i$.

Suppose that $i$ is in case~2 or~3 whereas $j$ is in case~1. Then necessarily $h_i > h_j$ (because $B'$ of $i$ has $\pi_v$ rank higher than $i$, whereas $B_j$ has rank $j$), and $\ell_i \ge \ell_j$ (because $B$ of $i$ is a feasible choice for $B'$ of $j$). 
\end{proof}

We now prove the symmetry property, showing that for every $j \in \{2, 3,\ldots ,2^{m}\}$ it holds that $T_v(j) + T_v(2^m - j+2) = 1$, or equivalently, that for every $j \in \{1, 2, \ldots ,2^{m} - 1\}$ it holds that $T_v(j+1) + T_v(2^m - j+1) = 1$. Consider the bundle $B_j$ of $\pi_v$ rank $j$. Its complement bundle $B_{j'}$ has rank $j' = 2^m - j + 1$. We need to show that that $T_v(2^m - j+1) = 1 - T_v(j+1)$. Observe that necessarily $T_v(j+1) + T_v(2^m - j+1) \ge 1$, as otherwise, {with one budgets slightly above $T_v(j+1)$ and the other slightly above $T_v(2^m - j+1)$,} one agent could guarantee a bundle of rank $j+1$ and the other a bundle of rank $2^m - j+1$, which is strictly better than the complement of a bundle of rank $j+1$ (as the complement has rank $2^m - j$). Hence, it suffices to show that $T_v(2^m - j+1) \le 1 - T_v(j+1)$, or in other words, that with a budget larger than $1 - T_v(j+1)$ player $P_1$ can guarantee the complement bundle $B_{j'}$.

If $b_1 \le T_v(j+1)$  then $P_1$ does not have a strategy that is guaranteed to give a bundle better than $B_j$. In other words, for any strategy $s_1$ for $P_1$ there is a strategy $s_2$ for $P_2$ that guarantees to herself a bundle at least as good as the complement, $B_{j'}$. Recall that strategies of $P_2$ (and hence, $s_2$) are of the following form. In every round $r$, the strategy sets a threshold $t_r$ that depends on the history of the game up to round $r$. If $P_1$ bids $t_r$ (or below), then $s_2$ bids $t_r$ and wins the round (as ties are broken in favor of $P_2$). If $P_1$ bids above $t_r$, then also in this case $P_2$ bids $t_r$, and this lets $P_1$ win the round (and pay her bid).  In other words, strategy $s_2$ dictates a bid $t_r$ for every round $r$ (where $t_r$ depends on the history of the game up to round $r$).

Now, swap the budgets between $P_1$ and $P_2$, and also change the names of the players to $P'_1$ and $P'_2$ (so no confusion arises with the original $P_1$ and $P_2$). After this swap, can $P'_1$ (with a budget of $1 - T_v(j+1)$) guarantee to herself a bundle at least as good as $B_{j'}$, by using the bidding strategy $s_2$? Not quite, because now ties in bids are broken in favor of $P'_2$, so the execution of the game with $P'_1$ and $P'_2$ is not the mirror image of the original execution with $P_1$ and $P_2$. However, suppose that the budget of $P'_1$ is strictly larger than $1 - T_v(j+1)$. This is equivalent to a setting in which budgets do not sum up to~1 and the budget of $P'_1$ has two components, a {\em main budget} of exactly $1 - T_v(j+1)$ and an auxiliary budget of $\varepsilon > 0$, whereas $P'_2$ has a budget of exactly $T_v(j+1)$.

Now, in each round $r$, $P'_1$ can consult $s_2$ regarding which bid $t_r$ to give, and then instead of bidding $t_r$, $P'_1$ can bid $t_r + \frac{\varepsilon}{m}$, where $t_r$ is taken from $P'_1$'s main budget, and  $\frac{\varepsilon}{m}$ is taken from $P'_1$'s auxiliary budget (ensuring that neither budget can run out prematurely during the game). When $P'_1$ uses this strategy, then if $P'_2$ lets $P'_1$ win the round, $P'_1$ does not spend more from her main budget than $P_2$ would spend in the game against $P_1$. If $P'_2$ wins the round, then $P'_2$ pays strictly more that $t_r$, and with such a bid $P_2$ would have let $P_1$ win the round. Hence $P'_2$ cannot win for herself a bundle better than $B_j$, implying that $P'_1$ wins a bundle at least as good as $B_{j'}$. This implies that $T_v(j+1) + T_v(2^m - j + 1) = 1$, establishing the symmetry property.

We remark that an alternative way of establishing the symmetry property is by induction on $m$, showing that symmetry of $T'_v$ (with $m-1$ items) implies symmetry of $T_v$ (with $m$ items). This involves using the formula $T_v(j) = \frac{T'_v(h)}{1 + T'_v(h) - T'_v(\ell)}$, establishing relations between $h$ for $j$ and $\ell$ for ${j'}$, and  between $\ell$ for $j$ and $h$ for ${j'}$, and using the symmetries that $T'_v$ is assumed to have. Details of this proof are omitted.
\end{proof}

\subsubsection{Worst-case-optimal strategies: fairness}\label{sec:optimal-fair}

Theorem~\ref{thm:budgetVector} has the following immediate consequences {for agents with a high entitlement (at least half). Fairness properties for low entitlements (below half) are not immediate from the structural result, but will be derived from our results on safe strategies presented in Section \ref{sec:goods-pos-ex-post}}.

\begin{corollary}
    \label{cor:bidding}
    In the bidding game, an agent $i$ {with an additive valuation over goods} 
    has a strategy with the following guarantees:
    \begin{itemize}
        \item If $b_i = \frac{1}{2}$, the agent receives at least her MMS.
        \item If $b_i > \frac{1}{2}$, the agent receives value of at least $\frac{v_i(\items)}{2}$.
    \end{itemize}
\end{corollary}

\begin{proof}
    Let $v_i$ be the valuation function of agent $i$. 
    
    Consider the two player version of bidding game described above with valuation function $v = v_i$ and with $b_1 = \frac{1}{2}$. By the proof of Theorem~\ref{thm:budgetVector}, $P_1$ has a strategy that gives her a bundle of value at least that of the bundle of $\pi_v$ rank $2^{m-1}$. Recall that bundles that are complements of each other appear in ranks $j$ and $2^m - j + 1$ under $\pi_v$, for some $j$. Consider now the optimal MMS partition into two bundles with respect to $v_i$. As these bundles are complements of each other, at least one of them has rank no higher than $2^{m-1}$, and by optimality of the MMS partition, its rank would be exactly $2^{m-1}$. Hence $P_1$ can win a bundle of value at last that of the MMS of agent $i$. By using this strategy of $P_1$ in the $n$ agent bidding game, agent $i$ can guarantee to herself a bundle of value at least her MMS.

    A similar argument as above, but with $b_i > \frac{1}{2}$ shows that agent $i$ has a strategy that gives her a bundle of value no worse than that of $\pi_v$ rank $2^{m-1} + 1$. As this bundle has value at least as high as its complement, this gives agent $i$ at least her proportional share.
\end{proof}

\subsubsection{Worst-case-optimal strategies: computational complexity}\label{sec:optimal-comp}

Theorem~\ref{thm:budgetVector} also has interesting consequences concerning the complexity of computing {optimal strategies} in the bidding game. Its proof implicitly describes a procedure for computing bids of a worst-case-optimal strategy, by backward induction. However, 
expressing these bids appears to require $2^m$ bits of precision. The theorem establishes this as an upper bound, and we do not have any significantly better upper bounds. If this precision is really necessary, then it might be that the optimal strategies require space exponential in $m$. Moreover, as the value of the game is very sensitive to small perturbations in the budget (for example, if there is only one item, a budget of $\frac{1}{2} + \varepsilon$ gives much higher value than a budget of $\frac{1}{2} - \varepsilon$), it might be difficult to even approximate the value of the bidding game in subexponential space. 

Though we do not know if computing the value of the bidding game is in NP, it is clearly NP-hard.

\begin{corollary}
    \label{cor:NPhard}
    When agents have 
    {additive integer-valuations over goods} 
    determining the {max-min value} of the bidding game is NP-hard. This holds for every tie breaking rule among bids.
\end{corollary}

\begin{proof}
    Consider two agents, each with entitlement $\frac{1}{2}$. By Corollary~\ref{cor:bidding}, each agent can ensure to herself her MMS. Hence the outcome of the bidding game is an allocation in which each agent gets at least her MMS. Finding such allocations is NP-hard (this is a known result, by a straightforward reduction from the NP-hard problem {\em partition}).
\end{proof}

\subsubsection{Worst-case-optimal strategies: non-monotonicity}\label{sec:optimal-non-mon}
{As part of our study of optimal strategies for the bidding game, we note another aspect (beyond Corollary~\ref{cor:NPhard}) that indicates that optimal strategies do not have a simple structure. We say that a bidding strategy is {\em monotone} if the sequence of bids dictated by the strategy in non-increasing. That is, for every round $r$, the bid in round $r+1$ cannot be larger than the bid in round $r$. Winning in round $r$ is more valuable than winning in round $r+1$, in the sense that the set of items available for selection at round $r+1$ is a strict subset of those available at round $r$. {Nevertheless,} {Proposition~\ref{prop:nonMonotone} 
shows that it may be desirable to bid higher in round $r+1$ than in round $r$: every monotone strategy gives less than $\frac{v_i(\items)}{2}$, while the optimal strategy gives at lease that much (and thus must be non-monotone.)}

\begin{proposition}\label{prop:nonMonotone}
    There are allocation instances with additive valuations over goods in which for some agent $i$ with entitlement $b_i > \frac{1}{2}$, every monotone strategy in the bidding game gives value less than $\frac{v_i(\items)}{2}$
(yet the optimal strategy gets at least $\frac{v_i(\items)}{2}$, so it is necessarily not monotone)
\end{proposition}

{We break the proof of Proposition~\ref{prop:nonMonotone} into two parts. The first part is presented in Lemma~\ref{lem:nonMonotone}.}
   
\begin{lemma}
\label{lem:nonMonotone}
There are allocation instances with additive valuations {over goods} in which the optimal strategy in the bidding game is not monotone.
\end{lemma}

\begin{proof}
    Consider an instance with two agents, with $b_1 = \frac{1}{3}$ and $b_2 = \frac{2}{3}$. Both agents have the same valuation function $v$ with $v(e_1) = v(e_2) = v(e_3) = \frac{1}{4}$. The remaining $m-3$ items each has value $\frac{1}{4(m-3)}$, and $m$ is assumed to be very large. To simplify the presentation (but without effecting the correctness), we omit negligible terms such as the value $\frac{1}{4(m-3)}$ of a single item, or small perturbations in bid value that may be required in order to ensure that ties in bids are broken in the direction of our choice. 

    Consider first a monotone bidding strategy for agent~1. If  her first bid is smaller than $\frac{2}{9}$, monotonicity forces also her next two bids to be smaller than $\frac{2}{9}$. This allows agent~2 to {win} 
    the first three items, and agent~1 cannot get a value larger than $\frac{1}{4}$. Alternatively, if the first bid of agent~1 is at least $\frac{2}{9}$, agent~2 lets agent~1 win this bid. Thereafter, agent~1 can bid at most $\frac{1}{9}$ on each of the next two items, and agent~2 can win both of them. After three items are consumed, the budget left for agent~2 is $\frac{4}{9}$ and the budget left of agent~1 is $\frac{1}{9}$. By bidding $\frac{1}{m-3} \cdot \frac{5}{9}$ on each of the remaining items, each agent guarantees to herself a fraction of the items equal to her fraction of the budget (and hence no agent can improve her guarantee). The total value obtained by agent~1 is $\frac{1}{4} + \frac{1}{5} \cdot \frac{1}{4} = \frac{3}{10}$. 

    We now show a non-monotone bidding strategy that guarantees for agent~1 a value larger than $\frac{3}{10}$. For this, it suffices to show that agent~1 has a strategy that ensures that she wins one of the first three items, and moreover, that if she wins exactly one of these items the budget that she has left is {some constant fraction of the total budget left that is strictly larger than} one fifth {(enabling her to get a value larger than $ \frac{1}{4}+ \frac{1}{5}\cdot \frac{1}{4}= \frac{3}{10}$, {as ``bidding the value'' strategy will ensure getting one fifth of the remaining value, up to one remaining item, which has a tiny value)}}. 
    Agent~1 bids $\frac{2}{9} - \frac{1}{45}$ in round~1. If agent~1 wins the bid, she can bid $\frac{1}{9} + \frac{1}{45}$ in each of the next two rounds. Agent~2 is forced to win these two items (as otherwise agent~1 reaches a value of $\frac{1}{2}$), and the budget left for agent~1 {(which is $\frac{1}{9} + \frac{1}{45}{=\frac{6}{45}}$)} is a fraction of $\frac{1}{4}$ of the total remaining budget {of $1-\left(\frac{2}{9} - \frac{1}{45}\right) -2\cdot \left(\frac{1}{9} + \frac{1}{45}\right) = \frac{24}{45}$}. Alternatively, if agent~2 wins the first item, then in round~2 agent~1 can bid $\frac{2}{9}$ (here we used non-monotonicity). If agent~1 wins the bid, she can still bid $\frac{1}{9}$ in round~3, reaching a situation in which the budget left for agent~1 {(which is $\frac{1}{9}= \frac{5}{45}$)} is a fraction of $\frac{5}{21}$ of the total remaining budget {of $1-\left(\frac{2}{9} - \frac{1}{45}\right) - \frac{2}{9} - \frac{1}{9} = \frac{21}{45}$}. If agent~2 wins the second round {(leaving her with a budget of $\frac{2}{3}- \left(\frac{2}{9} - \frac{1}{45}\right) - \frac{2}{9} = \frac{2}{9} + \frac{1}{45}$)}, then agent~1 wins the third round with a bid 
    {slightly above} {$\frac{2}{9} + \frac{1}{45}$} (agent~2 does not have enough budget left to bid higher). After three rounds, the budget left for agent~1 {(which is $\frac{1}{9}- \frac{1}{45}= \frac{4}{45}$)} is a fraction of $\frac{4}{15}$ of the total remaining budget {of $1-\left(\frac{2}{9} - \frac{1}{45}\right) - \frac{2}{9} - \left(\frac{2}{9} + \frac{1}{45}\right) = \frac{15}{45}$}. In the worst of these three cases, agent~1 has a fraction of $\frac{5}{21}$ of the remaining budget, so she can guarantee for herself a value of $\frac{1}{4} + \frac{5}{21} \cdot \frac{1}{4} \simeq 0.31{>\frac{3}{10}}$.
\end{proof}


{Armed with Lemma~\ref{lem:nonMonotone}, we now prove Proposition~\ref{prop:nonMonotone}.}


\begin{proof}
    Consider the instance presented in the proof of Lemma~\ref{lem:nonMonotone}. Recall that in that instance, $m$ is assumed to be very large (say, $m > 1000$) so that the value of every item $e_j$ with $4 \le j \le m$ is much smaller than $0.01$. Change the entitlements to be $b_1 = \frac{1}{2} + \varepsilon$ and $b_1 = \frac{1}{2} - \varepsilon$, where $\varepsilon < \frac{1}{m}$. Also, add to this instance a single item $e_0$ of value $0.39$. Hence, now $v(\items) = 1.39$, and if agent~1 gets a value smaller than $0.695$ she fails to get a value of $\frac{v(\items)}{2}$. We show that to get value at least  $0.695$, agent~1 must use a non-monotone strategy.

    Agent~2 bids $\frac{1}{4} + 3\varepsilon$ in round~1. 
    
    If agent~1 bids higher and wins the round, she gets item $e_0$ of value $0.39$. But now, agent~1 holds only a $\frac{1}{3}$ fraction of the remaining total budget, and by the proof of Lemma~\ref{lem:nonMonotone}, she does not have a monotone strategy that will give her additional value above $0.3$. Hence, {with any monotone strategy} agent~1 {reaches a value of at most} 
    ${0.39+0.3=0.69}<\frac{v(\items)}{2}$.

    If agent~2 wins round~1, then she selects $e_0$, and still has a $\frac{1}{3} - O(\varepsilon)$ fraction of the remaining total budget. By the non-monotone strategy of Lemma~\ref{lem:nonMonotone}, agent~2 can ensure to herself an additional value of $0.31 - O(\frac{1}{m})$, leaving agent~1 with a value of $0.69 + O(\frac{1}{m}) < 0.695$, where the last inequality holds when $m$ is sufficiently large.
\end{proof}

\subsubsection{Approximately Optimal Strategies}\label{sec:approx-optimal}

{We have seen that computing optimal strategies is NP-hard (Corollary~\ref{cor:NPhard}). We now show that near optimal strategies can be computed in polynomial time. Here, near optimality is in an additive sense, suffering a loss of value of at most $\varepsilon \cdot v_i(\items)$ compared to the optimal strategy. Such additive approximation suffices for our current paper (specifically, in order to prove Corollary~\ref{cor:PTAS}). We leave the question of the existence of a $(1 - \varepsilon)$ multiplicative approximation to future work.}



\begin{lemma}
    \label{lem:PTAS}
    There is a function $f$ (that will satisfy $f(x) = 2^{O(x)}$) such that in the biding game, for every $\varepsilon > 0$, there is a strategy running in time $f(\frac{1}{\varepsilon})$ times a polynomial in the input size that guarantees to an agent $i$ with entitlement $b_i$ and an additive valuation $v_i$ {over goods,} a value of at least $OPT(b_i,v_i) - \varepsilon {\cdot v_i(\items)}$. Here, $OPT(b_i,v_i)$ is the value that the optimal strategy guarantees against an adversary who holds all the remaining budget, when ties in bids are broken in favor of the adversary.
\end{lemma}

\begin{proof}
    We assume without loss of generality that items are sorted in decreasing order of their $v_i$ value, and that $v_i(\items) = 1$.

     {Suppose first that $v_i(e_1) \le \delta$ for some $\delta > 0$ (we shall later fix $\delta = \frac{\varepsilon}{4}$). In this case, agent $i$ can guarantee to herself a value equal to at least her proportional share, up to the value of one item, by following the ``bid your value" strategy. Thus ``bid your value" guarantees to the agent a value of at least $b_i - \delta$. Likewise, the adversary can follow the ``bid your value" strategy, and by this  limit the agent to a value of at most $b_i + \delta$. Hence  $OPT(b_i,v_i) \le b_i + \delta$. It follows that by using the polynomial time ``bid your value" strategy, the agent gets value at least  $OPT(b_i,v_i) - 2\delta$.}




Suppose now that $v_i(e_1) > \delta$. Let {$q$} 
be the largest index for which $v_i(e_q) > \delta$, and note that $q < \frac{1}{\delta}$. Consider an instance with a set $\items'$ of $t \le \lceil \frac{1}{\delta} \rceil$ items, and valuation $v'_i$ defined as follows. For $j \le q$ we define $v'_i(e_j) = v_i(e_j)$, whereas for $q < \ell \le t$ we define $v'_i(e_{\ell}) = \frac{1 - \sum_{j=1}^q v'_i(e_i)}{t-q}$. Observe that $v'_i(\items') = 1$, and that $v'_i(e_{\ell}) \le \delta$ for every $\ell > q$.

Observe that $v'_i$ is identical to $v_i$ on the prefix $\{e_1, \ldots, e_q\}$, differs from $v_i$ only in values of items in the suffix and they have value at most $\delta$, and the sum of values in both valuations is~1. It follows that in the bidding game,  $|OPT(b_i,v_i) - OPT(b_i,v'_i)| < 2\delta$, because once the prefix is consumed and a certain budget is left, for both suffixes the optimal value that can be achieved is the proportional share with an error of at most $\pm \delta$.

Consider the following strategy for agent $i$ with valuation $v_i$. For the first $q$ rounds, follow the optimal strategy for $v'_i$. Then, scale the budgets so that they sum up to~1, and $v_i$ so that the value of the remaining items sum up to~1, and follow the ``bid your value" strategy in the remaining rounds. The loss compared to the optimal strategy for $v_i$ is at most an additive term of $4\delta$ ($2\delta$ due to using $v'_i$ instead of $v_i$, and $2\delta$ for using ``bid your value" instead of the optimal strategy on the prefix).

The ``bid your value" strategy can be computed in polynomial time. The optimal strategy with respect to $v'_i$ can be computed in time polynomial in $2^{O(t)}$ by Theorem~\ref{thm:budgetVector}. (As there are $t$ items, the dimension of the vector $T_v$ is $2^t$, and specifying budgets requires at most $2^t$ bits of precision.) Choosing $\delta = \frac{\varepsilon}{4}$, the lemma follows. 
    \end{proof}



\subsection{Positive Results for ex-post Domination}\label{sec:goods-pos-ex-post}


{We show that for additive valuations {over goods}, there is an allocation that gives each agent at least half of her $\widehat{MMS}$. 
We later show (Proposition~\ref{pro:badExampleGoods}) that the share $(\frac{1}{2} + \varepsilon)$-$\widehat{MMS}$ is not feasible, thus
this is the best possible guarantee. Nevertheless, we offer a guarantee that is stronger for some valuation functions. Namely, we provide a guarantee of $\frac{1}{2}$-$\widehat{TPS}$ rather than just $\frac{1}{2}$-$\widehat{MMS}$. 
This guarantee is stronger because for additive valuations, it is always true that the TPS is at least as large as MMS. Moreover, for entitlements of the form $\frac{1}{k}$, the ratio between TPS and MMS may be as large as $2 - \frac{1}{k}$ {\cite{BEF2021c}}. 
Hence, the $\frac{1}{2}$-$\widehat{TPS}$ guarantee implies a $\rho_i$-$\widehat{MMS}$ guarantee to agent $i$, for $\frac{1}{2} \le \rho_i \le 1 - \frac{1}{2}\cdot \hat{b}_i$, 
where the value of $\rho_i$ depends only on the valuation $v_i$ of the agent.}

{In Section~\ref{sec:OptSafeBiddingGame} we have studied worst-case-optimal strategies for the bidding game. While we now have better understanding of the structure of these strategies, the analysis so far did  not provide a ratio $\rho$ such that these strategies imply a $\rho$-domination result for all feasible shares. To obtain such a result we design a $\frac{1}{2}$-$\widehat{TPS}$-safe strategy for the bidding game (which of course implies that the worst-case-optimal strategies of Section~\ref{sec:OptSafeBiddingGame} are $\frac{1}{2}$-$\widehat{TPS}$-safe as well), proving the existence of a $\frac{1}{2}$-$\widehat{TPS}$ allocations. Moreover, unlike the worst-case-optimal strategies, our $\frac{1}{2}$-$\widehat{TPS}$-safe strategy has polynomial complexity, except for one aspect, which is its use of the worst-case-optimal strategies when $b_i > \frac{1}{2}$.
Before describing our strategy, let us explain what difficulties our strategy needs to overcome.}

Consider agent $i$ with entitlement $b_i$ satisfying $\frac{1}{k+1} < b_i \le \frac{1}{k}$ {for some integer $k\geq 1$. Note that $\hat{b}_i=\frac{1}{k}$ and that $\widehat{TPS}(v_i,b_i)= TPS(v_i,\frac{1}{k})$}. Scale the valuation function $v_i$ so that $TPS(v_i,\frac{1}{k}) = \frac{1}{k}$, implying in particular that $v_i(\items) \ge 1$. We need to show that agent $i$ has a strategy that gives her value at least $\frac{1}{2k}$. This is a fraction of at most $\frac{1}{2k}$ of the total value, whereas the agent has a fraction of more than $\frac{1}{k+1}$ of the total budget. Hence, if the agent could get one unit of value for each unit of budget, she will meet her goal, even with some slackness (get a value more than $\frac{1}{k+1}$ instead of just $\frac{1}{2k}$). This suggests the ``bid your value" bidding strategy. At each round, agent $i$ bids the value of the highest value remaining item. If the agent wins, the value that she receives is equal to the budget that she spends. If the adversary wins the bid, the budget of the adversary is consumed at least at the same extent as value {(according to $v_i$)} is consumed. 
The adversary cannot win more value than its total budget, and hence the agent will have value at least $\frac{1}{k+1}$ left for her to win.

The above argument fails, because the strategy of ``bid your value" is not feasible. Below we explain the sources of this infeasibility. 
{The reason why 
the new bidding strategy that we later present is able to} 
cope with them is the slackness that we have, as the agent has a budget higher than $\frac{1}{k+1}$, but only needs to reach a value of $\frac{1}{2k}$. 
{As we shall see from the explanations below, our bidding strategy will treat  all values of $k \ge 3$ in a uniform way, 
but will need to offer a special treatment to the case $k=1$, and a different special treatment to the case that  
$k=2$.} 


\begin{itemize}
    \item There may be ``huge" items of value larger than the initial budget of the agent, and then the agent cannot bid their values. This is handled by the agent bidding her full budget as long as there are huge items. It is true that the adversary may win these items while paying for them less than their value, but this turns out to be acceptable because we consider the truncated proportional share (but would not be acceptable if we were to consider the proportional share). Consequently, we will be able to assume that there are no huge items. This, together with $TPS(v_i,\frac{1}{k}) = \frac{1}{k}$, implies that $v_i(\items) = 1$ (rather than $v_i(\items) \ge 1$).

    \item There may be several ``large" items of value slightly below $\frac{1}{2k}$. Winning one such item does not suffice for the agent, but winning two of them does. If agent $i$ wins a large item, the budget she has left is slightly above $\frac{1}{k+1} - \frac{1}{2k} = \frac{k-1}{k+1} \cdot \frac{1}{2k}$, which is insufficient in order to bid (the value) on the remaining large items. Hence the adversary can win the remaining large items at a discounted price, and might consume all value before exhausting its budget. We handle this problem by departing from the ``bid your value" strategy. Already on the first of the large items, the agent does not bid its value, but rather places a somewhat lower bid of $\frac{1}{2(k+1)}$. It may not be clear why this is helpful, as the adversary can still win all large items at a discount. 
    {Indeed, if $k=1$ (meaning that $b_i > \frac{1}{2}$) this approach fails. (For example, if $v_i(e_1) = \frac{2}{5}$, $v_i(e_2) = v_i(e_3) = v_i(e_4) = \frac{1}{5}$ and $\frac{1}{2} < b_i < \frac{9}{16}$ then bidding only $\frac{1}{4}$ on $e_1$ will result in agent $i$ getting a value of only $\frac{2}{5} < \frac{1}{2}$. {The adversary will win $e_1$, and then will bid $\frac{3}{16}$ in each of the remaining rounds, winning an additional item.})} 
    However, we show that when $k \ge 2$, this approach can be made to work. 
   {For the case of $k=1$ this approach fails, but in this case $b_i > \frac{1}{2}$, for which we already proved (see Corollary~\ref{cor:bidding}) that the worst-case-optimal strategy achieves at least $\frac{v_i(\items)}{2} \ge \frac{1}{2}$-$\widehat{TPS}$. Hence, here and in the following two bullets we assume that $k \ge 2$.}
   
    

    \item There may be several ``medium" items of value slightly below $\frac{1}{4k}$. Winning two such items does not suffice for the agent, but winning three of them does. If agent $i$ wins two of these items, the budget she has left is slightly above $\frac{1}{k+1} - 2 \cdot\frac{1}{4k} = \frac{2(k-1)}{k+1} \cdot \frac{1}{4k}$. When $k\ge 3$, {the budget left is at least $\frac{1}{4k}$} and suffices in order to keep on bidding on medium items. However, for $k = 2$, the budget left does not suffice, and hence the case $k=2$ needs special care, beyond that of $k \ge 3$.

    \item There may be several items of value slightly below $\frac{1}{6k}$. Winning three such items does not suffice for the agent, but winning four of them does. If agent $i$ wins three of these items, the budget she has left is slightly above $\frac{1}{k+1} - 3 \cdot \frac{1}{6k} = \frac{3(k-1)}{k+1} \cdot \frac{1}{6k}$. 
    When $k\ge 2$, {the budget left is at least $\frac{1}{6k}$} and suffices in order to keep on bidding on these items, so this does not cause a problem. 
\end{itemize}

Summarizing, there are three types of obstacles for implementing the ``bid your value" strategy. The first obstacle is that of huge items. This obstacle can be overcome because we are considering the TPS and not the proportional share. The second obstacle is that of large items. For this we have a strategy that works when $k \ge 2$, but not when $k=1$. Consequently, the case of $k=1$ is handled via {the worst-case-optimal strategy.} 
The third obstacle is that of medium items. This is problematic only if $k=2$, causing the strategy for $k=2$ to be more elaborate than the one for $k \ge 3$. 
}

{We are finally ready to present our main result, which implies Theorem \ref{thm:main-intro} (our main positive result).}

\begin{theorem}
\label{thm:1/2} 
{There is a bidding strategy in the bidding game that for additive valuation $v_i$ {over goods} and entitlement $b_i$ gives agent $i$ a set of value at least $\frac{1}{2}$ of $TPS(v_i, \hat{b}_i)$.  } {Thus, the share $\frac{1}{2}$-$\widehat{TPS}$, which can be computed in polynomial time, is feasible and $\frac{1}{2}$-dominating.}
\end{theorem}

The proof of the theorem can be found in Section \ref{sec:proof-thm1/2}.
The positive result of Theorem \ref{thm:main-intro} follows from the combination of Theorem \ref{thm:1/2}, the fact that $TPS(v_i, \hat{b}_i)\geq MMS(v_i, \hat{b}_i)= \widehat{MMS}(v_i, {b}_i)$, and Proposition \ref{prop:rounding-MMS-goods}. The negative result of Theorem \ref{thm:main-intro} (proving tightness) follows from Proposition \ref{pro:badExampleGoods} (actually, it is a stronger negative result, showing impossibility even when only requiring $\rho$-domination of feasible shares that are nice).

\mbfuture{future work can  consider submodular valuations over goods.}    



\subsubsection{Poly-time approximate implementation}

Theorem \ref{thm:1/2} gives an strategy that is not poly-time computable. We next show that for any $\varepsilon>0$ it is possible to obtain a {$(1-\varepsilon)$} 
{multiplicative} approximation in polynomial time. 
Our $\frac{1}{2}$-$TPS(v_i,\hat{b}_i)$-safe strategy is poly-time computable, except for the case that the  entitlement of the agent is more than half the leftover entitlement, {in which the strategy switches to use the worst-case-optimal strategy}. To derive the approximation result we use Lemma~\ref{lem:PTAS}. Though Lemma~\ref{lem:PTAS} considers additive approximations and not multiplicative ones,
{for the case that $b_i > \frac{1}{2}$, the additive approximation provided implies a multiplicative approximation,  as the share is large (at least $\frac{v_i(\items)}{2}$).} 
\begin{corollary}
    \label{cor:PTAS}
    There is a function $f$ (that will satisfy $f(x) = 2^{O(x)}$) such that in the biding game, for every $\varepsilon > 0$, {there is a strategy running in time $f(\frac{1}{\varepsilon})$ times a polynomial in the input size} that guarantees to an agent $i$ with entitlement $b_i$ and an additive valuation $v_i$ {over goods}, a value of at least $(\frac{1}{2} - \varepsilon)TPS(v_i,\hat{b}_i)$. In particular, for every constant $\varepsilon > 0$, the bidding strategy runs in polynomial time. 
\end{corollary}

\begin{proof}
    We assume without loss of generality that items are sorted in decreasing order of their $v_i$ value, and that $v_i(\items) = 1$.
    
    We show that agent $i$ has a polynomially computable strategy $\sigma_i$ that determines for $i$ the bid in each round, and guarantees that by the end of the bidding game, agent $i$ gets a bundle of value at least $(\frac{1}{2} - \varepsilon)TPS(v_i,\hat{b}_i)$. The strategy is based on a combination of the strategies presented in the proofs of Theorem~\ref{thm:1/2} and Lemma~\ref{lem:PTAS}. The reason why we get a loss of only $\varepsilon \cdot TPS(v_i,\hat{b}_i)$ compared to the bound in Theorem~\ref{thm:1/2}, and not a loss of $\varepsilon$ (as seems to be implied by Lemma~\ref{lem:PTAS}), is that the strategy of Lemma~\ref{lem:PTAS} is used only at a stage in which the agent $i$ holds more than half the total budget and needs to get half of the total remaining value. Hence one may assume that the remaining value at this point is at most $TPS(v_i,\hat{b}_i)$. Scaling that value to~1, the loss of $\varepsilon$ due to Lemma~\ref{lem:PTAS} translates to a loss of $\varepsilon \cdot TPS(v_i,\hat{b}_i)$ compared to the original valuation.

Strategy $\sigma_i$ will not attempt to keep track of the budget that each agent has left after each round.  Attempting to do so might be complicated and involve storing values at very high precision, if bids of agents who win items and pay their bids are given in very high precision. So instead, in each round, $\sigma$ only keeps track of the bids of agent $i$ at that round, and whether agent $i$ won the bid or not (one bit of information). This suffices in order to update the budget of agent $i$ from round to round. As to the budget of other agents, $\sigma_i$ only maintains an easily computable upper bound of the total budget remaining for all other agents combined. Initially, this upper bound is $1 - b_i$, and after each round in which agent $i$ does not win, the upper bound is decreased by the bid of agent $i$ in that round (even if the actual bid that won the round was strictly larger). The agent employs the strategy of the proof of Theorem~\ref{thm:1/2}, where in each round she pretends that the budget remaining for all other agents is this upper bound, rather than trying to keep track of their exact budgets. This still guarantees that the agent receives at least $\frac{1}{2}TPS(v_i,\hat{b}_i)$.
    
    If at any point, the budget left for agent $i$ becomes larger than the upper bound that the agent computes on the total budget remaining for other agents, then the agent switches to bid according to the strategy of Lemma~\ref{lem:PTAS}. At that point, half of the total value remaining {suffices for the agent in order to reach $\frac{1}{2}TPS(v_i,\hat{b}_i)$}, whereas that strategy gives her at least a fraction of $(\frac{1}{2} - \varepsilon)$ of that value. 

    It remains to show that each bid of agent $i$ can be computed in time polynomial in $n$, $m$ and $f(\frac{1}{\varepsilon})$, and is a rational number whose numerator and denominator are polynomial in the input parameters for the agent (the numerators and denominators of $b_i$, and of $v_i(e)$ for the various items $e \in \items$). Inspection of the strategy in the proof of Theorem~\ref{thm:1/2} shows that this holds (with no dependence on $\varepsilon$) as long as the agent holds not more than half of the remaining budget. When the agent holds more than half the budget, she switches to the strategy in the proof of Lemma~\ref{lem:PTAS}, and inspection of that proof shows that this still holds (this time, with dependence on $\varepsilon$). 
\end{proof}

\subsubsection{Proof of Theorem \ref{thm:1/2}}\label{sec:proof-thm1/2}
In this section we present the proof of Theorem \ref{thm:1/2}, showing that agent  
$i$ with additive valuation $v_i$ over goods and entitlement $b_i$ has a bidding strategy that guarantees the agent a bundle of value at least $\frac{1}{2}$ of $TPS(v_i, \hat{b}_i)$.

\begin{proof}[Theorem \ref{thm:1/2}]
    Consider an agent $i$ with entitlement $b_i$ satisfying $\frac{1}{k+1} < b_i \le \frac{1}{k}$ for integer $k$. Hence, $\hat{b}_i = \frac{1}{k}$. Scale the valuation function of the agent so that $TPS(v_i,\frac{1}{k}) = \frac{1}{k}$. 
    We now modify $v_i$ so that $TPS(v_i,\frac{1}{k}) = PS(v_i,\frac{1}{k})$. This is done by reducing the value of items $e$ that satisfy $v_i(e) > \frac{1}{k}$ to $v_i(e) = \frac{1}{k}$. 

    Hence, we are now in a situation in which  $v_i(\items) = 1$, 
    no item has value above $\frac{1}{k}$, and we need to show the existence of a bidding strategy that gives agent $i$ value at least $\frac{1}{2k}$. 
    Our proof has two parts. One handles the special case of $k=1$, and the other handles the remaining cases of $k \ge 2$. 

    The case $k=1$, meaning that $b_i > \frac{1}{2}$, is proved in Corollary~\ref{cor:bidding} {(for every $m$)}. Interestingly, for this case we do not present an explicit bidding strategy that gives the agent at least value $\frac{1}{2}$. Instead, we use a nonconstructive argument 
    that proves that such a bidding strategy exists. 

    For the case $k \ge 2$, our proof is by induction on the number $m$ of items. We note that the fact that the case of $k=1$ is handled separately for all values of $m$ assists in carrying out the induction step, because in some cases the induction step will reduce an allocation instance with $k \ge 2$ to a new allocation instance with $k=1$. 

    The facts that $v_i(\items) = 1$ and $v_i(e) \le \frac{1}{k}$ for every item $e$ imply that $m \ge k$. Hence, $m=k$ serves as the base case for our induction. In this case, each of the items has value $\frac{1}{k}$. The bidding strategy of bidding $\frac{1}{k+1}$ at every round ensures that the agent wins at least one item, because the budget of the adversary is strictly smaller than $\frac{k}{k+1}$. Hence the agent gets value of at least $\frac{1}{k} > \frac{1}{2k}$, as desired.

    

    We now explain how to handle the induction step.  Recall that $\frac{1}{k+1} < b_i \le \frac{1}{k}$ for $k \ge 2$, that  $v_i(\items) = 1$, that there is no item of value above $\frac{1}{k}$, and the goal of agent $i$ is to get a bundle of value at least $\frac{1}{2} TPS(v_i, \frac{1}{k}) = \frac{1}{2k}$. We assume {without loss of generality} that items are named in order of decreasing $v_i$ value, $v_i(e_1) \ge v_i(e_2) \ge \ldots \ge v_i(e_m)$. 
    We call a setting a {\em proportional setting} if it holds that $TPS(v_i, \frac{1}{k}) = PS(v_i, \frac{1}{k})$. Note that a setting with $v_i(\items)=1$ is proportional if and only if $v_i(e_1)\leq 1/k$.
    The setting we start with is proportional. {In our induction steps, when we reduce an instance with $m$ items to a new instance with fewer items, we will prove that if the original instance is in the proportional setting, then so is the new instance. This will allow us to use the inductive hypothesis on the new instance.}
    

    {We first show that {in the proportional setting} we can assume that there are no {\em huge items}, of value above $\frac{1}{2k}$.} 

{
\begin{claim}
    \label{claim:hugeItems}
    {Consider an agent $i$ with entitlement $b_i$ satisfying $\frac{1}{k+1} < b_i \le \frac{1}{k}$ for integer $k$.}
    Suppose that in the proportional setting, $v_i(e_1) \ge \frac{1}{2k}$. Then there is a strategy that gives agent $i$ a value of at least $\frac{1}{2k}$.
    \end{claim}}

\begin{proof}
    If $v_i(e_1) \ge \frac{1}{2k}$, agent $i$ bids $\frac{1}{k+1}$. If the agent wins the bid, she gets a value of $\frac{1}{2k}$, as desired. If the agent loses the bid, then the adversary pays at least $\frac{1}{k+1}$, and gets to choose an item of value at most $\frac{1}{k}$. For the allocation instance that remains on the set $\items' = \items \setminus \{e_1\}$ of items, scale the sum of entitlements so that they add up to~1, and scale $v_i$ (by at most $\frac{k}{k-1}$) so that $v'_i(\items') = 1$. One can readily see the the scaled entitlement $b'_i$ satisfies  $\frac{1}{k'+1} < b_i \le \frac{1}{k'}$ for $k' = k-1$, that $v'_i(\items') = 1$, and that no item $e'$ has $v'_i$ value above $\frac{1}{k'}$. As $|\items'| = m-1 < m$ {and the new instance is also in a} 
    proportional setting, the inductive hypothesis implies that agent $i$ has a strategy that gives her a bundle $B$ with $v'_i(B) = \frac{1}{2k'}$, implying that $v_i(B) \ge \frac{1}{2k}$, as desired.
    \end{proof}

    {Claim~\ref{claim:hugeItems}} implies that {in the proportional setting} we may assume that $v_i(e_1) < \frac{1}{2k}$. The following claim shows that we can further assume that $v(e_2) + v(e_3) < \frac{1}{2k}$. 
    
    \begin{claim}
    \label{claim:largeItems}
    {Consider an agent $i$ with entitlement $b_i$ satisfying $\frac{1}{k+1} < b_i \le \frac{1}{k}$ for integer $k$.}
    Suppose that 
    $v_i(e_1) < \frac{1}{2k}$ and $v_i(e_2) + v_i(e_3) \ge \frac{1}{2k}$. Then there is a strategy that gives agent $i$ a value of at least $\frac{1}{2k}$.
    \end{claim} 
    
    \begin{proof}
    Observe that in the setting of the claim, $v_i(e_2) = \frac{1}{2k} - \delta$ for some $0 < \delta \le \frac{1}{4k}$, {as $v_i(e_2)\geq v_i(e_3)$ and $v_i(e_2) + v_i(e_3) \ge \frac{1}{2k}$ implies that  $v_i(e_2)\geq \frac{1}{4k}$.}
    
    
    We offer the following bidding strategy for $i$. Bid $\frac{1}{2(k+1)}$ in each of the first two rounds. The adversary cannot let agent $i$ win both rounds, as then agent $i$ reaches a value of at least $\frac{1}{2k}$, as desired. The adversary {will also fail if it wins} 
    both rounds, as then it pays $\frac{1}{k+1}$ and takes a value of at most $2\cdot \frac{1}{2k} \le \frac{1}{k}$. In  such a case, 
    {a proof similar to that of Claim~\ref{claim:hugeItems}}
    implies that the agent reaches a value of $\frac{1}{2k}$. Hence, the adversary must win exactly one of the first two rounds, and the worst case is when the adversary wins $e_1$ and the agent wins $e_2$. 

    Now, in the third round, the agent again bids $\frac{1}{2(k+1)}$. The adversary is forced to win the third round, as we assume that $v(e_2) + v(e_3) \ge \frac{1}{2k}$. 
    
    In subsequent rounds, as long as there is an item of value above $\delta$, the agent continues bidding $\frac{1}{2(k+1)}$. The adversary is forced to win all these items as getting an additional value of $\delta$ is sufficient for the agent to obtain value $\frac{1}{2k}$ and win. Let $r \ge 3$ be the total number of items won in this phase. The adversary wins $e_1$ and $e_3, \ldots, e_r$ and pays $\frac{r-1}{2(k+1)}$, whereas the agent wins $e_2$ and pays $\frac{1}{2(k+1)}$. Note that $v_i(e_j) < \frac{1}{2k}$ for all $j$, $v_i(e_2) = \frac{1}{2k} - \delta$, and $v_i(e_j) \le v_i(e_2)$ for all $j \ge 3$. Note also that $r-1 \le 2k-1$ (the adversary cannot afford to pay for more items {as its budget is less than $\frac{k}{k+1}$}), implying that $r \le 2k$, and thus items of positive value still remain. 
 
    Scaling the entitlements of the agents so that they sum up to~1, the new entitlement of agent $i$ is $b'_i$ which now satisfies $b'_i \ge \frac{1}{2k+2-r}$.  The target additional value that agent $i$ needs to achieve is $\delta$.  If the new instance is a member of the proportional setting, then by the inductive hypothesis (which can be used, as the number of items decreased), the agent should be able to obtain a fraction of $\frac{1}{2(2k+1-r)} = \frac{1}{4k-2r+2}$ of the total value that remains. {We denote the set of remaining items by $\items'$.} 
    We now demonstrate that indeed we are in the proportional setting, and that $\frac{1}{4k-2r+2} \cdot v_i(\items') \ge \delta$, or equivalently, $(4k-2r+2)\cdot \delta \le v_i(\items')$.
    
    The value of the remaining items satisfies $v_i(\items') > 1 - \frac{1}{2k} - (r-1)(\frac{1}{2k} - \delta) = 1 - \frac{r}{2k} + (r-1) \delta)$. Hence we wish to demonstrate that $(4k-2r+2)\cdot \delta \le 1 - \frac{r}{2k} + (r-1) \delta)$, which simplifies to $(4k-3r+3)\cdot \delta \le 1 - \frac{r}{2k}$. If $4k-3r+3 \le 0$ the inequality surely holds. Hence we may assume that $4k-3r+3 > 0$, and then the left hand side is maximized when $\delta$ is maximized. As $\delta < \frac{1}{4k}$. Hence it suffices to show that  $\frac{4k-3r+3}{4k} \le 1 - \frac{r}{2k}$, which holds for all $r \ge 3$. And indeed, $r > 3$.
    
    Finally, note that indeed we are in the proportional setting. This is because every item in $\items'$ has value at most $\delta$, and $\delta \le \frac{1}{4k-2r+2} \cdot v_i(\items') \le b'_i \cdot v_i(\items')$.
    \end{proof}

    Given Claim~\ref{claim:largeItems}, 
    we can assume that we have the following three {\em item value bounds}:
    
    \begin{enumerate}
        \item $v_i(e_1) < \frac{1}{2k}$.
        \item $v_i(e_2) + v_i(e_3) < \frac{1}{2k}$.
        \item $v_i(e_j) < \frac{1}{4k}$ for all $j \ge 3$. (This is implied by bound 2 above.) 
    \end{enumerate}
    
    Given these item value bounds, the ``bid your value" bidding strategy (recall the discussion of this strategy in the introductory text preceding the statement of Theorem~\ref{thm:1/2}) can handle all cases of $k \ge 3$. 
    The significance of $k \ge 3$ is that then the inequality $\frac{1}{k+1} \ge \frac{3}{4k}$ holds, whereas for $k = 2$ this inequality fails {(so we present a different argument for that case, see Claim~\ref{claim:k2})}. Among other things, this inequality implies (in combination with the item value bounds) that $v_i(e_1) + v_i(e_j) {< \frac{1}{2k}+\frac{1}{4k}=
    \frac{3}{4k}}\le \frac{1}{k+1}$ for all $j \ge 3$. 


{When the agent bids below her value, the adversary can win items ``cheaply" and accumulate a ``surplus''. We define the {\em surplus} of the adversary from consuming item $e_j$ that she won to be the difference between $v_i(e_j)$ (the value according to $v_i$ of the item $e_j$), and the adversary's payment for item $e_j$. The surplus for all items the adversary consumed is simply the sum of surpluses of these items.}
        

    \begin{claim}
\label{claim:largek}
{Consider an agent $i$ with entitlement $b_i$ satisfying $\frac{1}{k+1} < b_i \le \frac{1}{k}$ for integer $k$.}
If the item value bounds hold, then there is a strategy that gives agent $i$ a value of at least $\frac{1}{2k}$ when $k\ge 3$.        
    \end{claim}

    \begin{proof}
        Agent $i$ uses the ``bid your value" strategy. When $k \ge 3$, this strategy is feasible for all rounds until agent $i$ gets a value of $\frac{1}{2k}$, except possibly for one single round, which is round~2. 
        The only situation in which  ``bid your value" is not feasible is if $v_i(e_1) + v_i(e_2) > b_i > \frac{1}{k+1}$. In this case, agent $i$ might win $e_1$ but then will not be able to afford to bid the full value of $e_2$. (Note that agent $i$ will be able to afford to later bid the full value of $e_3$, because  $v_i(e_1) < \frac{1}{2k}$ {and $v_i(e_3) < \frac{1}{4k}$, implying that $v_i(e_1) + v_i(e_3) \le \frac{1}{k+1}$ for $k \ge 3$}.) Consequently, the adversary might be able to accumulate a {\em surplus}, which is access value {(according to $v_i$)} that the adversary {consumes} 
        on $e_2$, compared to the adversary's payment.  
        As $v_i(e_2) < \frac{1}{2k}$ whereas the bid of agent $i$ {on $e_2$} is at least $b_i - v_i(e_1) \ge b_i - \frac{1}{2k}$, the surplus of the adversary is at most ${v_i(e_2) - \left( b_i - \frac{1}{2k}\right)} < \frac{1}{2k} - b_i + \frac{1}{2k} = \frac{1}{k} - b_i$.

        It follows that the total value that the adversary can consume before agent $i$ reaches a value of $\frac{1}{2k}$ is $1 - b_i + \frac{1}{k} - b_i < 1 + \frac{1}{k} - \frac{2}{k+1} = 1 - \frac{k-1}{k(k+1)} < 1 - \frac{1}{2k}$ (where the last inequality holds for $k \ge 3$). This does not suffice in order to prevent agent $i$ reaching a value of $\frac{1}{2k}$.
    \end{proof}

It remains to handle the case of $k=2$. {We defer the proof of the following claim to Appendix \ref{app:proof-k2}.} 

    \begin{claim}
\label{claim:k2}
{Consider an agent $i$ with entitlement $b_i$ satisfying $\frac{1}{3} < b_i \le \frac{1}{2}$.
If the item value bounds hold (with $k=2$), then there is a strategy that gives agent $i$ a value of at least $\frac{1}{4}$.}   
    \end{claim}
We summarize here the proof of Theorem~\ref{thm:1/2}.    
We show that there is a bidding strategy in the bidding game that for additive valuation $v_i$ {over goods} and entitlement $b_i$ gives agent $i$ a set of value at least $\frac{1}{2}$ of $TPS(v_i, \hat{b}_i)$.
Consider agent $i$ with entitlement $b_i$ satisfying $\frac{1}{k+1} < b_i \le \frac{1}{k}$ {for some integer $k\geq 1$. Note that $\hat{b}_i=\frac{1}{k}$ and that $\widehat{TPS}(v_i,b_i)= TPS(v_i,\frac{1}{k})$}. Scale the valuation function $v_i$ so that $TPS(v_i,\frac{1}{k}) = \frac{1}{k}$, implying in particular that $v_i(\items) \ge 1$. We need to show that agent $i$ has a strategy that gives her value at least $\frac{1}{2k}$.
If $k=1$ (and so $b_i > \frac{1}{2}$), the claim follows from Corollary \ref{cor:bidding} {(for every $m$)}: the agent can receive value of at least $\frac{v_i(\items)}{2}\geq \frac{1}{2}\cdot TPS(v_i, \hat{b}_i)$.
For $k\geq 2$, {the proof is by induction on $m$. For convenience of the presentation (though this is not strictly necessary for the proofs) the induction hypothesis is strengthened to assume that we are in the proportional setting, in which no single item has value larger than $\frac{1}{k}$. This indeed can be assumed to hold for the original input instance, and the inductive steps verify that this property holds for those sub-instances on which we wish to apply the induction hypothesis.}

The case in which $v_i(e_1) \ge \frac{1}{2k}$ is handled in Claim \ref{claim:hugeItems}, and the case that $v_i(e_1) < \frac{1}{2k}$  yet $v_i(e_2) + v_i(e_3) \ge \frac{1}{2k}$ is handled in Claim \ref{claim:largeItems}. All other cases are handled in Claim \ref{claim:largek} (when $k\geq 3$), and in Claim \ref{claim:k2} (when $k=2$).

Finally, as the TPS can be computed in polynomial time \cite{BEF2021b}, $\widehat{TPS}$ can be computed in polynomial time as well.  The share $\frac{1}{2}$-$\widehat{TPS}$ is feasible as we shown that there is a bidding strategy that for any agent $i$ with additive valuation $v_i$ {over goods} and entitlement $b_i$ gives agent $i$ a set of value at least $\frac{1}{2}$ of $TPS(v_i, \hat{b}_i)$. 

\end{proof}

\subsection{The Bidding Game: Practical Considerations}\label{sec:practical}

As we view the bidding game as a natural allocation mechanism that might be used in practice, we briefly discuss some practical considerations. 

Recall that every round has two components. One is a bidding process that decides which agent wins the round and the payment for that agent. The other is a selection process by which the winning agent selects an item of her choice. We discuss each of these components separately.

For the bidding process, though we described this process as a first-price auction, all that our analysis uses is the assumption that the highest bidder wins (breaking ties arbitrarily), and that the payment is at most the bid of the winner, and at least the highest bid of the non-winners. Beyond this, it does not matter how the bidding process is implemented, and one may choose the implementation that one is most comfortable with (either first price or second price or anything in between, agents may submit bids simultaneously or according to some order determined by the mechanism, possibly random, these can be sealed bids or open bids, ties can be broken at random or in some consistent manner, such as always in favor of the earlier bidder, etc.).

For tie breaking to be ex-ante fair, it makes sense to pick an order over agents uniformly at random, and use that order for tie breaking. This makes the procedure symmetric ex-ante.  As our results hold for any tie breaking, they hold for this as well. 

{We now discuss the item selection process. The bidding game is most natural if the agents have valuations under which there is a plausible ``best" item to choose if the agent wins the round. This includes valuations such as additive, unit demand and budget additive, in which items can be ordered by their values, and a natural item to select is the item of highest {stand-alone} value among those remaining. As an example, a unit-demand agent can participate in the bidding game (say, bidding her entire budget at every round) alongside additive agents. Our guarantee for any additive agent still holds, as it holds regardless of how other agents bid. It is also easy to see that a unit-demand agent can guarantee her $\widehat{MMS}$ (which is the best possible) in the bidding game. The bidding game can also plausibly be used when agents have submodular valuations, in which a natural (though not necessarily best) item to select is the one with highest marginal value compared to the set of items that the agent already holds. In contrast, for XOS valuations and beyond, the bidding game is not a recommended allocation mechanism, as it gives only poor guarantees. See~\cite{bUF23}.}

\section{Additional Results for Additive Valuations over Goods}\label{sec:goods}


In this section we present some additional results for additive valuations of goods. 
We first show that our result for ex-post shares, showing that $\frac{1}{2}\cdot \widehat{TPS}$ is feasible, is tight (Section \ref{sec:goods-neg-ex-post}).
We then present tight upper and lower bounds for feasible ex-ante shares, showing that for $\delta \simeq 0.591$ 
{(a constant that arises in the context of some integer sequence discussed below),}
the share $\delta \cdot \widehat{PS}$ is a feasible share ex-ante, and any larger fraction of $\widehat{PS}$ is not feasible ex-ante (Section \ref{sec:MESgoods}). Finally, we prove a strong impossibility result for ``best-of-both-worlds", showing that any randomized allocation that gives every agent some constant fraction of her proportional share, is supported on some allocations in which an agent does not get a constant fraction of her $\widehat{MMS}$  (Section \ref{sec:BoBW-goods}). 

\subsection{Tight Impossibility for ex-post Domination}\label{sec:goods-neg-ex-post}
We have shown (Theorem \ref{thm:1/2}) that {there is a bidding strategy in the bidding game that for additive valuation $v_i$ {over goods} and entitlement $b_i$ gives agent $i$ a set of value at least $\frac{1}{2}$ of $TPS(v_i, \hat{b}_i)$.  } 
This implies that the share 
$\frac{1}{2}\cdot TPS(v_i, \hat{b}_i)$ is feasible. We next show that any fraction of $TPS(v_i, \hat{b}_i)$ that is larger than $\frac{1}{2}$ is not feasible. 

\begin{proposition}
\label{pro:badExampleGoods}
Consider additive valuations over goods and arbitrary entitlements.
For every $\varepsilon > 0$ the share $(\frac{1}{2} + \varepsilon)$-$\widehat{MMS}$ is not feasible.
Thus, there is no feasible ex-post share that $(\frac{1}{2} + \varepsilon)$-dominates every {nice} 
ex-post share.
\end{proposition}

\begin{proof} We first show that the share $(\frac{1}{2} + \varepsilon)$-$\widehat{MMS}$ is not feasible.
    {Consider $m=2n-2$ identical items and $n \ge 2$ agents, with additive valuations. That is, each item has value~1 for each agent.} The entitlements satisfy $\frac{1}{n+1} < b_n < \frac{1}{n}$, and $\frac{1}{n} < b_i < \frac{1}{n-1}$ for all agents $i \not=n$. The $\widehat{MMS}$ of agent~$n$ is~1, whereas for every other agent the $\widehat{MMS}$ 
    is~2. 
    In every allocation, {either some agent gets no item, in which case the allocation is not dominating for any factor, or every agent gets at least one item. In the latter case, some agent $i \not=n$, for which $\widehat{MMS}$ is~2, 
    gets only one item of value~1. Hence, no domination with a factor larger than $1/2$ is possible.}

    {Combining the fact that the share $(\frac{1}{2} + \varepsilon)$-$\widehat{MMS}$ is not feasible with Proposition \ref{prop:goods-additive-minimal}, we conclude that there is no feasible ex-post share that $(\frac{1}{2} + \varepsilon)$-dominates every 
    nice ex-post share.}
\end{proof}

\mbfuture{We should check what can be said about the reciprocal entitlements case (or the iterative fair inheritance case). }


\subsection{Tight Impossibility for ex-ante Domination}
\label{sec:MESgoods}

In this section {we present a tight bound on feasible ex-ante shares for additive valuations. Concretely, for additive valuations we show that for $\gamma \simeq 1.69103$ (explicitly defined below) the share $\frac{1}{\gamma} \cdot \widehat{PS}$ is a feasible share ex-ante, while for every $\epsilon > 0$  the share $(\frac{1}{\gamma} + \epsilon) \cdot \widehat{PS}$ is not (Theorem~\ref{thm:MES})}. 

Before proving the theorem, we present some background on sequences of integers that we shall use in the proof.

The following sequence of integers $q_1, q_2, \ldots$, is known as a {\em Sylvester sequence}:

\begin{align*}
q_{1} & =2\\
q_{k} & =1+\prod_{i=1}^{k-1}q_{i}
\end{align*}

The Sylvester sequence grows at a double exponential rate. The first six terms in this sequence are 2, 3, 7, 43, 1807, 3263443. 

A related sequence $a_1, a_2, \ldots$ is defined by $a_i = q_i - 1$ for every $i$. Hence its first six terms are 1, 2, 6, 42, 1806, 3263442. It satisfies the following recurrence:

\begin{align*}
a_{1} & =1\\
a_{k} & = a_{k-1}(a_{k-1} + 1)
\end{align*}

One may observe that for every $n$,

\begin{equation}
\label{eq:sylvester}
\sum_{i=1}^n \frac{1}{q_i} = 1 - \frac{1}{a_{n+1}}     
\end{equation}

Let $s_n$ denote $\sum_{i=1}^{n} \frac{1}{a_i}$. Thus the first four terms of the sequence $s_1, s_2, \ldots$ are 1, $\frac{3}{2}$, $\frac{5}{3}$, $\frac{71}{42}$.  The sequence $s_1, s_2, \ldots$ will be useful for us due to the following lemma. {(We suspect that this lemma is known, but as we could not find a reference for it, so we state it and prove it below).} 

\begin{lemma}
\label{lem:Sylvester}
Let $2 \le k_1 \le \ldots \le k_n$  be a non-decreasing sequence of integers satisfying $\sum_{i=1}^n \frac{1}{k_i} < 1$. Then $\sum_{i=1}^n \frac{1}{k_i - 1} \le s_n$, for $s_n$ as defined above.
\end{lemma}

\begin{proof}
    Recall the sequences $q_1, q_2, \ldots$ and $a_1, a_2, \ldots$ defined above.  We need to prove that $\sum_{i=1}^n \frac{1}{k_i - 1} \le \sum_{i=1}^{n} \frac{1}{a_i}$. Assume for the sake of contradiction that $\sum_{i=1}^n \frac{1}{k_i - 1} > \sum_{i=1}^{n} \frac{1}{a_i} = s_n$. Then $k_1 \le \ldots \le k_n$ are not the first $n$ terms of the Sylvester sequence $q_1, q_2, \ldots$. Let $t$ be the first index for which $k_t \not= q_t$. Necessarily $k_t > q_t$, as otherwise $\sum_{i=1}^n \frac{1}{k_i} \ge \sum_{i=1}^{t-1} \frac{1}{q_i} + \frac{1}{a_t} = 1$ (the equality follows from Equation (\ref{eq:sylvester})), contradicting the requirement that  $\sum_{i=1}^n \frac{1}{k_i} < 1$. Given that $k_t > q_t$, we infer that $t < n$.

    The requirement $\sum_{i=1}^n \frac{1}{k_i} < 1$ implies that $\sum_{i=t}^n \frac{1}{k_i} < \frac{1}{a_t}$, and $k_t > q_t$ implies that $\sum_{i=t}^n \frac{1}{k_i-1} \le \frac{a_t + 2}{a_t + 1}\sum_{i=t}^n \frac{1}{k_i} = \frac{a_t + 2}{(a_t + 1)a_t}$. We also have $\sum_{i=t}^n \frac{1}{q_i-1} \ge \sum_{i=t}^{t+1} \frac{1}{q_i-1} = \frac{1}{a_t} + \frac{1}{a_t(a_t+1)}=\frac{a_t + 2}{(a_t + 1)a_t}$. Hence, $\sum_{i=1}^n \frac{1}{k_i - 1} \le \sum_{i=1}^{n} \frac{1}{a_i}$, reaching a contradiction to our assumption.
\end{proof}

\begin{corollary}
    \label{cor:sylvester}
    For  $n \ge 2$ and $s_n$ as defined above, for every vector $b_1, \ldots, b_n$ of entitlements (satisfying $\sum_{i=1}^n b_i = 1$), the respective unit upper bounds satisfy $\sum_{i=1}^n \hat{b}_i \le s_n$.
\end{corollary}

\begin{proof}
    If {for every $i$ it holds that $b_i=\frac{1}{k_i}$} 
    for some integer $k_i$, then we have that $\sum_{i=1}^n \hat{b}_i = \sum_{i=1}^n {b}_i = 1 < s_n$, {as $n\geq 2$}. Hence, we may assume that for some $i$ there is an integer $k_i$ for which $\frac{1}{k_i + 1} < b_i < \frac{1}{k_i}$. This implies that $\sum_{i=1}^n ((\hat{b}_i)^{-1} + 1)^{-1} < \sum_{i=1}^n {b}_i = 1$, and the Corollary follows from Lemma~\ref{lem:Sylvester}.
\end{proof}

Let $\gamma = \lim_{n \rightarrow \infty} s_n$ be the limit of $s_n$ as $n$ grows. We have that $\gamma \simeq 1.69103$ (one can verify this numerically, because the $a_i$ grow at a double exponential rate).  Hence, $\frac{1}{\gamma} \ge 0.591$.

\begin{theorem}
    \label{thm:MES}
    For $\gamma$ as defined above (with $\frac{1}{\gamma} \ge 0.591$) and every instance in which agents have additive valuations {over goods},  $\frac{1}{\gamma} \cdot \widehat{PS}$ is a feasible share ex-ante. For every $\epsilon > 0$, there are allocation instances with additive valuations {over goods} for which $(\frac{1}{\gamma} + \epsilon) \cdot \widehat{PS}$ is not feasible ex-ante.
\end{theorem}

\begin{proof}
To show feasibility of $\frac{1}{\gamma} \cdot \widehat{PS}$, allocate 
{all of $\items$}
to each agent $i$ with probability {$\frac{1}{\gamma} \cdot \hat{b}_i$.}
By Corollary~\ref{cor:sylvester}, {$\sum_{i=1}^n \hat{b}_i \le s_n < \gamma$, and so the grand bundle is allocated at most once.}
With the remaining probability the items can be allocated to an arbitrary agent. In expectation, each agent (with an additive valuation function) gets a value of at least a $\frac{1}{\gamma}$ fraction of her $\widehat{PS}$. 

To show that $(\frac{1}{\gamma} + \epsilon) \cdot \widehat{PS}$ is not feasible, let $n$ be such that $\frac{1}{\gamma} + \epsilon > \frac{1}{s_n}$ (the smallest $n$ for which this inequality holds satisfies $n \le O(\log\log \frac{1}{\epsilon})$), and suppose that all $n$ agents have the same additive valuation function $v$. With the vector of entitlements in which $b_i = \frac{1}{q_i} + \frac{1}{n \cdot a_{n+1}}$,  we have (by Eq.~(\ref{eq:sylvester})) that $\sum_{i=1}^n {b}_i = 1$ and that $\sum_{i=1}^n \hat{b}_i = s_n$.
Hence, the sum of values that agents expect to receive is $(\frac{1}{\gamma} + \epsilon)\cdot s_n \cdot v(\items) > v(\items)$, and this cannot be satisfied.
\end{proof}

{Theorem \ref{thm:intro-goods-ex-ante} follows from the combination of Theorem \ref{thm:MES}, the fact that for additive valuations $\widehat{PS}=\widehat{MES}$, and Proposition \ref{prop:rounding-MES-goods}.} 
{We remark that the impossibility can be strengthen to rule out $(\frac{1}{\gamma} + \epsilon)$-domination even of only feasible shares that are nice, using personalized shares defined in  Appendix \ref{sec:nice}, similarly to the proof of Proposition \ref{prop:goods-additive-minimal}. We omit the details.}

\subsection{Best of Both Worlds: an Impossibility}\label{sec:BoBW-goods}
A Best-of-Both-Worlds (BoBW) result aims to present a randomized allocation rule that is both ax-ante fair as well as ex-post fair. That is, for every instance, it obtains some fairness in expectation, but is supported only on allocations that are fair. 

We extend our share-based approach for deterministic allocation to the BoBW paradigm. An ultimate result will be a ``double-domination'' BoBW result, getting (approximate) domination of every feasible share, both in the ax-ante sense as well as the ex-post sense.   
That is, we now are looking for two shares, a share $s_R$ for randomized allocations, and a share $s_D$ for deterministic allocation, such that:
1) in expectation, each agent gets an acceptable randomized allocation with respect to her share $s_R$, and
2) for every realization, each agent gets an acceptable (deterministic) bundle with respect to her deterministic share $s_D$,
3) the share $s_R$ dominates ($\rho_R$-dominates with as large as possible $\rho_R$) every feasible share for randomized allocations,  
4) the share $s_D$ dominates ($\rho_D$-dominates with as large as possible $\rho_D$) every feasible share for deterministic allocations.

For additive valuations, the ex-ante proportional share is clearly a randomized share that is feasible (by giving every agent all the items with probability $b_i$). {By Proposition \ref{prop:goods-additive-minimal},  $\widehat{MMS}$ is a dominating share and is the minimal dominating share (every dominating share must dominate $\widehat{MMS}$).} Thus, to obtain a ``double-domination'' BoBW result, we need to $\rho_R$-dominates the ax-ante proportional share for some constant $\rho_R>0$, and to $\rho_D$-dominates the $\widehat{MMS}$ share  for some constant $\rho_D>0$. 

We observe that  ``double-domination'' BoBW result is impossible, even for additive valuations. {Proposition \ref{prop:intro-BoBW-goods} follows from the next claim and Proposition \ref{prop:goods-additive-minimal}.}
\begin{proposition}
    For any $n\geq 2$ and any $m\geq 1$, there is a setting with $n$ agents that have identical additive valuations over $m$ identical {goods} 
    in which ``double-domination'' BoBW result is impossible: Every randomized allocation that gives every agent some constant fraction of her proportional share, is supported on some allocation in which an agent does not get a constant fraction of her $\widehat{MMS}$. 
\end{proposition}
\begin{proof}
    We first prove the claim for  {$1\leq  m < n$}.
    {Consider $m$ identical items and {$n > m$} agents, with identical  additive valuations. That is, each item has value~1 for each agent.} 
    The agents are divided into two sets. The set $K$ of the first $m$ agents are more entitled, each agent $i\leq m$ has an entitlement of {$\frac{1}{m+1} < b_i < \frac{1}{m}$}. 
    Each of the other agents has an entitlement of {$0 < b_i \le \frac{1}{m+1}$}.
    The proportional share, which is feasible ex ante, is positive for every agent. 
   The $\widehat{MMS}$ of each of the agent in $K$ (the $m$ more entitled agents) is~1, so to get any positive fraction of that share ex post, each of these agents must get an item in every ex-post allocation. Thus, in any allocation in which the $\widehat{MMS}$ of every agent in $K $ is approximated within any finite factor, every agent not in $K$ gets nothing, and thus, with only such allocations in the support, it is impossible to give such an agent any constant fraction of her proportional share ex ante (note that such an agent exists as $K$ includes at most $n-1$ agents). 

   The construction for $m\geq n$ is based on adding $m-n+1$ items of value $0$ to $n-1$ items of value as in the case that $m=n-1$ above.  
\end{proof}

Thus, it is impossible to get a BoBW result with double-domination. A natural problem for future work is to get single-domination, either ex-ante or ex-post, coupled with a ``resonable'' share for the other requirement. 
For example, getting a single-domination ex-post by obtaining  constant domination of the $\widehat{MMS}$, while relaxing the ex-ante domination requirement to a constant fraction of the ex-ante {${TPS}$, instead of $\widehat{PS}$.}


\section{Assignment of Chores}\label{sec:chores}


In this section we consider the assignment of chores (undesired items). 
The definitions in this case involve some natural modifications compared to the definitions for allocations of goods. Every agent $i$ has a cost function $c_i$ (instead of a valuation function $v_i$), which is a set function over chores which is normalized (the cost of the empty set is~0), nonnegative, and weakly monotone non-decreasing. 
The agent prefers bundles of lower cost over those of higher cost. 
A share $s$ determines the share cost $s(c_i,b_i)$ of agent $i$ with cost function $c_i$ and responsibility $b_i$, This is the maximal cost the agent should bear: a set of $S$ chores is acceptable to agent $i$ if $c_i(S)\leq s(c_i,b_i)$.
An assignment  
must assign all chores, and it 
is acceptable if each agent is assigned an acceptable set of chores. 
Additionally, a share for  chores needs to be monotone in the responsibility, that is, the more responsibility, the larger the share: $s(c_i,b_i)\leq s(c_i,b^+_i)$ for any $b_i\leq b_i^+$ and any $c_i$. (For chores, we use this monotonicity condition in our proofs.
{As we consider it an important property, we require a similar monotonicity condition also for goods, although we do not use it in our proofs.)}


{The notion of domination was defined for the setting of goods in Definition~\ref{def:dominating}. For chores, agents try to minimize cost instead of maximizing value. This requires a natural adjustment to the notion of domination.}

\begin{definition}
Consider any class $C$ of cost functions over chores, and any $\rho > 0$. A share function $\tilde{s}$ over chores $\rho$-dominates (or simply dominates, if $\rho = 1$) share $s$, if for every cost function $c\in C$ and every responsibility $b$, it holds that $\tilde{s}(c,b) \le \rho \cdot s(c,b)$. The share $\tilde{s}$ is  $\rho$-dominating (or simply \emph{dominating} if $\rho = 1$) for class $C$ if it $\rho$-dominates every share $s$ that is feasible for $C$. 

\end{definition}


{For items that are good,  the characterization we have presented involves the unit upper bound of the entitlement. For chores, the same role is played by the unit lower bound, as we define next. } 

\begin{definition}
\label{def:unitLowerBound}
    For {responsibility} $0 < b < 1$, let $k$ be such that $\frac{1}{k+1} \le b < \frac{1}{k}$.
    Then the {\em unit lower bound} on $b$, denoted by $\check{b}$, is $\frac{1}{k+1}$. Define the ex-post share $\overline{MMS}$ as $\overline{MMS}(c_i,b_i) = MMS(c_i,\check{b}_i)$ and the ex-ante share $\overline{MES}$ as $\overline{MES}(c_i,b_i) = MES(c_i,\check{b}_i)$.
\end{definition}

{Note that in the context of chores, MMS is the minimax share (of cost), which is the same as the maximin share (of value) if we were to replace the non-negative cost function by a non-positive valuation function.}
 
\subsection{Share Domination}\label{sec:chores-share-dom}

{We next present characterizations for share domination for general classes of valuations over chores, both for ex-post shares and for ex-ante shares. The claims are parallel to the claims we have proven for goods (Proposition \ref{prop:rounding-MMS-goods}, Proposition \ref{prop:rounding-MES-goods} and Proposition \ref{prop:goods-additive-minimal}).} 
{As the proofs are similar to the proofs of the parallel claims for good, they are deferred to Appendix \ref{app:chores}.}
\begin{proposition}\label{prop:rounding-chores}
    For any class $C$ of cost functions  over chores,  $\overline{MMS}$ dominates every feasible unrestricted ex-post share for arbitrary responsibilities  (for class $C$). Moreover, it is the maximal ex-post share for arbitrary responsibilities that dominates every feasible unrestricted ex-post share (for class $C$).

    For any class $C$ of cost functions  over chores,  $\overline{MES}$ dominates every feasible unrestricted ex-ante share for arbitrary responsibilities  (for class $C$). Moreover, it is the maximal ex-ante share for arbitrary responsibilities that dominates every feasible unrestricted ex-ante share (for class $C$).

    For the class of additive cost functions  over chores and arbitrary responsibilities, the share $\overline{MMS}$ is the maximal ex-post share  that dominates every feasible {nice} ex-post share, and the share $\overline{PS}$ is the maximal ex-ante share  that dominates every feasible 
    {nice} 
    ex-ante share. 
\end{proposition}   
{In the rest of the section we move to consider additive cost functions.}

\subsection{Tight ex-post Domination Result for Additive Chores}\label{sec:chores-ex-post}

{We consider the assignment of chores to agents with additive cost functions and arbitrary responsibilities.
We first} introduce a share, the {\em Rounded{-responsibilities}  Round Robin share} (RRR), that is applicable for chores when the cost functions is additive. {We then show that this share obtains the best $\rho$-domination ($\rho=2$).}

\begin{definition}
    \label{def:RRR}
    Consider an agent with additive cost function $c_i$, {responsibility} $0<b_i<1$, and the integer $k$ for which  $\frac{1}{k+1} \le b_i < \frac{1}{k}$. Suppose that chores are ordered, where chore $e_1$ has highest cost and $e_m$ has lowest cost (under $c_i$). Let $\items_{1|k}$ denote the set of items whose index equals~1 modulo $k$, namely, $e_1, e_{k+1}, \ldots$. Then the {\em Rounded{-responsibilities}  Round Robin share} (RRR) is defined as:
    
    $$RRR(c_i,b_i) = \sum_{j\in \items_{1|k}} c_i(e_j).$$  
\end{definition}

The name Rounded{-responsibilities}  Round Robin share is because first the {responsibility} of the agent, which satisfies $\frac{1}{k+1} \le b_i < \frac{1}{k}$, is ``rounded" to $\frac{1}{k}$, and then the cost of the share is equal to that of the respective Round Robin share.\footnote{For the rounding to make more sense, it seems appropriate to consider the range $\frac{1}{k+1} < b_i \le \frac{1}{k}$ instead of $\frac{1}{k+1} \le b_i < \frac{1}{k}$. Our results hold whichever definition is used. We choose to use the former range over the latter range for compatibility with the definition of the $\overline{MMS}$.} That is, it equals to the cost that agent $i$ can limit her bundle to have if chores where assigned by a Round Robin protocol that involves $k$ agents, and agent $i$ is in the worst picking position in this protocol (picking 
in those rounds that include the last round, and hence getting the worst chore, $e_1$). 


\begin{theorem}
\label{thm:feasibleChores} 
    For additive costs {and arbitrary responsibilities}, the {\em Rounded{-responsibilities}  Round Robin share}
    (RRR)
   is polynomial time computable, {nice} 
   and 2-dominating. It is feasible, and a feasible {assignment} can be computed in polynomial time.
\end{theorem}

\begin{proof}
Clearly, $RRR(c_i,b_i)$ is computable in polynomial time. Moreover, being a picking order share, it is nice. (Name independence holds for all picking order shares. Likewise, being self maximizing holds for all picking order shares, as was proved in~\cite{BF22}.) 

    To see that $RRR(c_i,b_i)$ is 2-dominating, we {shall use (in the last inequality when bounding $RRR(c_i,b_i)$ below)} the fact that for chores and {responsibility} that is the inverse of an integer, the MMS 
    is at least the cost of the most costly item, and at least the proportional share.

{Consider agent $i$ with {responsibility} $b_i$, such that $\frac{1}{k+1} \leq b_i < \frac{1}{k}$ and $\check{b}_i=\frac{1}{k+1}$}. It holds that 
\begin{multline*}
RRR(c_i,b_i) = \sum_{j\in \items_{1|k}} c_i(e_j) \le c_i(e_1) + \frac{1}{k}c_i({\items \setminus \{e_1\}}) \\ = \frac{k-1}{k}c_i(e_1) + \frac{1}{k}c_i(\items) \le \frac{k-1}{k}\overline{MMS}(c_i,\check{b}_i) + \frac{k+1}{k}\overline{MMS}(c_i,\check{b}_i) = 2\cdot \overline{MMS}(c_i,\check{b}_i)
\end{multline*}


    {Feasibility follows from an assignment algorithm that is based on a picking sequence. It is easier to describe the reverse of this picking sequence, which means that at round $r$ the number of chores remaining is $r$, and the agent picks the chore of lowest cost among them. For this reverse sequence $\pi$, we shall have the following {\em spreading} property, which states that the set of rounds in which an agent of {responsibility} $b_i < \frac{1}{k}$ gets to pick is not worse than $1, 1+k, 1+2k, \ldots$. This implies that for every $j$, the number of chores that agent $i$ is assigned among her $j\cdot k$ most costly chores is at most $j$. Hence the cost that agent $i$ suffers is at most $\sum_{j\in \items_{1|k}} c_i(e_j) = RRR(c_i,b_i)$, showing feasibility of RRR.   
    
    We now describe how we construct $\pi$, the reverse picking sequence. As will be evident from our description, there is a lot of flexibility in the construction of $\pi$, and many different $\pi$ would work. The construction is done in rounds. In each round $r$, for each agent $i$, we maintain a {\em debt} $d_{i,r}$. Initially, $d_{i,1} = b_i$ for each agent $i$. At every round $r$, the agent that gets to pick in round $r$ is selected arbitrarily from the set of agents that have positive debt in round $r$. For example, in round~1 we may select any of the agents to make the first pick. We refer to the agent selected in round $r$ as $i_r$. Thereafter, the debts for round $r+1$ are computed as follows. For every agent $j \not= i_r$, we set $d_{j,r+1} = d_{j,r} + b_j$. For agent $i_r$ we set $d_{i_r,r+1} = d_{i_r,r} + b_{i_r} - 1$. The term $-1$ signifies that agent $i_r$ covered one unit of debt by picking a chore at round $r$.
    
    Observe that in every round, the sum of debt added is $\sum_i b_i = 1$, and the amount of debt consumed is~1. Hence at the beginning of a round the sum of debts is~1, and there is an agent of positive debt, and indeed there is an agent that can be selected. Moreover, at the end of every round, each agent has debt strictly larger than $-1$. This implies that an agent $i$ with $b_i \le \frac{1}{k}$ could not have picked in more than $j$ rounds among the first $j \cdot k$ rounds, as then her debt at the end of round $j \cdot k$ would have been at most $b_i \cdot j \cdot k - (j+1) \le \frac{1}{k} \cdot j \cdot k - (j+1) \le -1$.  This establishes the spreading property referred to above.} 
\end{proof}

{Theorem \ref{thm:feasibleChores-intro} directly follows from Theorem \ref{thm:feasibleChores} and Proposition \ref{pro:badExampleChores}.}



\subsection{Tight Best of Both Worlds Result for Additive Chores}\label{sec:chores-BoBW}

In the proof of Theorem~\ref{thm:feasibleChores}, we showed that there is a deterministic {assignment} algorithm that computes an {assignment} that is feasible with respect to the RRR share. A randomized version of that {assignment} algorithm can be used in order to give a {\em best of both worlds} (BoBW) result -- a randomized {assignment} algorithm that assigns each agent not more than her proportional share ex-ante, and not more than her RRR share ex-post. This is the content of Theorem~\ref{thm:BoBWChores}. Its proof follows principles that appear in several previous works. Among them, two are most relevant to our work. The general structure of our randomized {assignment} algorithm follows a pattern introduced in~\cite{aziz2020simultaneously}. The use of this pattern simplifies the proof of correctness (compared to some other alternatives that also work). The randomized {assignment} algorithm that we use was already sketched in Section 1.6 of~\cite{FH23arxiv} (the arxiv version of~\cite{FH23}). There it was explained that with this algorithm, each agent suffers a cost that is at most her proportional share in expectation, and no more than twice her anyprice share (APS) ex-post. The proof of Theorem~\ref{thm:BoBWChores} observes that the same algorithm also meets the additional benchmark of giving each agent not more than her RRR share ex-post.

\begin{theorem}
\label{thm:BoBWChores}
{Consider assignment of chores to agents with additive cost functions and arbitrary responsibilities.}
There is a polynomial time randomized {assignment} algorithm that for every input instance 
outputs a distribution over {assignments} such that under this distribution each agent suffers a cost that is at most her proportional share in expectation, and no more than her RRR share ex-post. 
\end{theorem}

\begin{proof}
The various steps of randomized {assignment} algorithm are described below. 

\begin{enumerate}
    
    \item Introduce $m$ main {\em coupons}, {numbered from~1 to $m$,} and $t < n$ additional auxiliary coupons,{numbered from $m+1$ to $m+t$.} The value of $t$ is chosen as $\sum_{i=1}^n (\lceil b_i \cdot m \rceil - b_i \cdot m)$. In particular, if $b_i \cdot m$ is an integer for every agent $i$, there is no need for auxiliary coupons.
    
    \item Allocate the coupons fractionally to the agents as follows. For each main coupon, agent $i$ gets a $b_i$ fraction of it. For each {of the $t$} auxiliary coupons, agent $i$ gets a $\frac{\lceil b_i \cdot m \rceil - b_i \cdot m}{t}$ fraction of it. Observe that the total fractions allocated for every individual coupon is exactly~1, and that the total fraction of coupons received by each agent $i$ is $\lceil b_i \cdot m \rceil$. (The reason for introducing the auxiliary coupons was so that this latter total will be an integer.)
    
    \item Replace each agent $i$ by $\lceil b_i \cdot m \rceil$ ``subagents" $a_1^i, \ldots a_{\lceil b_i \cdot m \rceil}^i$. The fractions of coupons that agent $i$ holds are distributed consecutively {(one by one, starting with coupon~1)} among her subagents, 
    so that the total fractions held by each subagent is exactly~1 ({if $b_i$ is not the inverse on an integer, this constraint entails that  there are coupons whose fractions are split among two consecutive subagents}). For example, suppose that $b_i \cdot m = 2$, implying that total fractions that agent $i$ holds is~2, and these are fractions of main coupons. Suppose further that $m$ is odd, implying that $\frac{m}{2}$ is not an integer. Then subagent $a_1^i$ gets a $b_i$ fraction of each of the first $\lfloor \frac{m}{2} \rfloor$ coupons, subagent $a_2^i$ gets a $b_i$ fraction of each of the last $\lfloor \frac{m}{2} \rfloor$ main coupons, and each of the subagents gets a fraction of $\frac{b_i}{2}$ of coupon $\lceil \frac{m}{2} \rceil$.
    
    \item Following the previous steps, we are now in a situation in which we have a perfect fractional matching $F$ between the coupons and the subagents. Using the Birkhoff-von Neumann theorem, decompose this fractional matching into a distribution over integral matchings. Following this step, 
    we now have $T$ 
    matchings $M^1, \ldots , M^T$ and associated nonnegative weights $p_1, \ldots, p_T$ with $\sum_{j=1}^T p_j = 1$, satisfying $F = \sum_{j=1}^T p_j \cdot M^j$. In each matching $M^j$, every subagent gets one coupon, and every coupon is allocated to one subagent.
    
    \item For each matching $M^j$, each agent $i$ collects the main coupons of her subagents. (The auxiliary coupons are discarded.) Observe that the coupons that an agent holds enjoy the following {\em spreading} property: for integer $k$ that satisfies $\frac{1}{k+1} \le b_i < \frac{1}{k}$, for every positive integer $\ell$, agent $i$ holds at most $\ell$ coupons among the first $\ell \cdot k$ coupons. 
    
    \item With each matching $M^j$ we associate an {assignment} $A^j$ as follows. Chores are {assigned} in $m$ rounds. When $r$ chores remain (which happens in round $m - r + 1$), the agent $i_r$ that holds coupon $r$ exchanges her coupon for a chore. The chore {assigned to agent $i$}
    is the one of lowest cost (according to $c_{i}$, breaking ties arbitrarily) among those chores that still remain. Its cost is no worse then the chore of rank $r$ in the decreasing cost order for that agent. The spreading property then implies that for every agent, the total cost of chores {assigned} is no more than her RRR share.
    
    \item For every $j$, the randomized {assignment} algorithm outputs {assignment} $A^j$ with probability $p_j$ (the probability of the respective matching $M^j$ in the Birkhoff-von Neumann decomposition). Hence ex-post, every agent {suffers a} cost not larger than her RRR share. Ex-ante, each agent $i$ receives each coupon $r$ with probability exactly $b_i$. As in exchange for this coupon agent $i$ is never {assigned} a chore that is more costly than the chore of rank $r$ (according $c_i$), it follows that the cost ex-ante is not larger than the proportional share.
\end{enumerate}

This randomized {assignment} algorithm runs in polynomial time because the Birkhoff-von Neumann decomposition (step~4 of the algorithm) has polynomially many matchings, and they can all be found in polynomial time. (The following polynomial time algorithm can be used, running in $T = O(mn)$ iterations. In every iteration $j$ it holds a bipartite graph $G^j$ with coupons on one side, subagents on the other side, and an edge between a subagent and a coupon if the subagent holds a positive fraction of the coupon. $G^j$ satisfies Hall's condition, and hence has a perfect matching $M^j$, which can be found in polynomial time. The weight $p_j$ given to $M^j$ is the minimum among the fractions that a subagent has of its matched coupon. For every matched edge, $p_j$ is subtracted from the fraction of the coupon that the subagent has. Hence at least one positive fraction becomes~0. Thus, the number of edges in $G^{j+1}$ is strictly smaller than that in $G^j$. As the number of edges in $G^1$ is $O(mn)$, the decomposition has $O(mn)$ matchings.)
\end{proof}

\subsection{Impossibility Results}\label{sec:chores-LB}

{The results of Theorem~\ref{thm:feasibleChores} and Theorem~\ref{thm:BoBWChores}} are best possible in a sense described in Proposition~\ref{pro:badExampleChores}.  Observe that if integer $k$ is such that $\frac{1}{k+1} \le b_i < \frac{1}{k}$ then $\check{b}_i = \frac{1}{k+1} > \frac{k}{k+1} \cdot b_i$. Hence the ratio between $\check{b}_i$ and $b_i$ might be close to $\frac{1}{2}$ when $k = 1$, but approaches~1 as $k$ grows (and $b_i$ becomes smaller). For this reason, approximation ratios compared to $\overline{MMS}(c_i,b_i) = MMS(c_i,\check{b}_i)$ that are excluded when $b_i$ is allowed to be close to~1, are not necessarily excluded when $b_i$ is small. Consequently, Proposition~\ref{pro:badExampleChores} distinguishes between the case that {responsibilities} can be arbitrarily large, and the case that 
all responsibilities are smaller than some $\delta>0$, which can be arbitrarily small.

\begin{proposition}
\label{pro:badExampleChores}
For additive costs, there are the following negative examples. 

\begin{enumerate}
\item  For every $\epsilon > 0$, there are {assignment} instances for which in every {assignment}, some agent $i$ is assigned  a bundle of cost at least $(2 - \epsilon) \cdot \overline{MMS}(c_i,b_i)$, {and in every randomized assignment is assigned  a bundle of expected cost at least $(2 - \epsilon) \cdot \overline{PS}(c_i,b_i)$}. (In this respect, the ratio of~2 in Theorem~\ref{thm:feasibleChores} is best possible.)
    \item For every $\delta > 0$, there are {assignment} instances in which every agent has {responsibility} smaller than $\delta$, for which in every {assignment}, some agent $i$ is assigned a bundle of cost at least $\frac{3}{2}\cdot \overline{MMS}(c_i,b_i)$. 
    \item For every $\delta > 0$, there are {assignment} instances in which every agent has {responsibility} smaller than $\delta$, for which in every randomized {assignment that assigns} each agent at most her proportional share ex-ante, there must be some agent $i$ that has strictly positive probability of being assigned (ex-post) a bundle of cost at least  $2\cdot \overline{MMS}(c_i,b_i)$. (In this respect, the ratio of~2 implied by Theorem~\ref{thm:BoBWChores} is best possible, {even if no agent has a high responsibility}.)
\end{enumerate}
\end{proposition}

\begin{proof}
In each of the {assignment} instances that are used in order to prove this proposition, all agents have the same cost function.

To prove item~1, {given $\epsilon > 0$,} fix integer $t \ge {\frac{2}{\epsilon}}$, 
and consider $m = 2t$ identical chores (each of cost~1) and two agents with $b_1 = \frac{2t-1}{2t}$ and $b_2 = \frac{1}{2t}$. Then {$\check{b}_1= \frac{1}{2}$ and $\check{b}_2= \frac{1}{2t}$, and} 
$\overline{MMS}(c_1,b_1) = t$ and $\overline{MMS}(c_2,b_2) = 1$. In every {assignment}, either agent~2 is assigned at least two chores and a cost of at least~$2=2\cdot \overline{MMS}(c_2,b_2)$, or agent~1 is assigned $2t-1$ chores and a cost of at least {$2t-1=(2 - \frac{1}{t}) \cdot t\geq   (2 - \varepsilon)\cdot \overline{MMS}(c_1,b_1)$}. {Additionally, 
${PS}(c_1,b_1) = 2t-1$ and ${PS}(c_2,b_2) = 1$, while
$\overline{PS}(c_1,b_1) = t$ and $\overline{PS}(c_2,b_2) = 1$. Thus, in every {randomized assignment}, either the expected number of chores agent~2 is assigned at least 2, and so her expected cost is at least $2\cdot \overline{PS}(c_2,b_2)$, or the expected number of chores agent~1 is assigned is at least $2t-2$, with cost of at least {$2t-2=(2 - \frac{2}{t}) \cdot t\geq   (2 - \varepsilon)\cdot \overline{PS}(c_1,b_1)$}.
}

To prove item~2, {given $\delta > 0$,} fix $n \ge \frac{1}{\delta} +1$. Consider $2n$ identical chores (each of cost~1) and $n$ agents. The {responsibilities} satisfy $b_n = \frac{1}{2n}$, and $\frac{1}{n} < b_i=\frac{1}{n-1}\cdot \frac{2n-1}{2n} < \frac{1}{n-1}$ 
for all agents $i \not=n$. The $\overline{MMS}$ of agent~$n$ is~1, whereas for every other agent the $\overline{MMS}$ is~2. In every {assignment}, either agent $n$ is assigned {at least two chores of cost~1 and suffering a cost of at least twice her $\overline{MMS}$, or an agent with $\overline{MMS}$ of~2 is assigned at least three chores of cost~1, suffering a cost that is at least $\frac{3}{2}$ of her $\overline{MMS}$.}

To prove item~3,  {given $\delta > 0$,}  let $m$ be such that $m-1 \ge \frac{1}{\delta}$. There are $m$ identical chores, each of cost~1. Some agent $i$ has {responsibility} $b_i$ satisfying $\frac{1}{m} < b_i < \frac{1}{m-1} \le \delta$. The $\overline{MMS}$ of agent~$i$ is~1. Other agents may have arbitrary {responsibility} (smaller than $\delta$).
As all agents have the same cost function, every randomized {assignment} that {assigns to every} 
agent at most her proportional share ex-ante must {assign to every} 
agent exactly her proportional share ex-ante. Consequently, {as $\frac{1}{m} < b_i$} agent~$i$ must have positive probability of {being assigned} more than one chore, and hence, at least twice  her $\overline{MMS}$.
\end{proof}

{Theorem \ref{thm:BoBWChores-intro} follows from Theorem \ref{thm:BoBWChores} and Proposition \ref{pro:badExampleChores}. }

\mbfuture{nice. We should check what can be said about the reciprocal \mbe{responsibilities} case (or the iterative fair inheritance case). }


\subsection*{Acknowledgments} 
 
Moshe Babaioff's research is supported in part by a Golda Meir Fellowship.
Uriel Feige's research is supported in part by the Israel Science Foundation (grant No. 1122/22).

\bibliographystyle{alpha}



\begin{thebibliography}{ABFRV22}

\bibitem[ABCM17]{ABCM2017}
Georgios Amanatidis, Georgios Birmpas, George Christodoulou, and Evangelos
  Markakis.
\newblock Truthful allocation mechanisms without payments: Characterization and
  implications on fairness.
\newblock In {\em The {ACM} Conference on Economics and Computation
  {(ACM-EC)}}, pages 545--562, 2017.

\bibitem[ABF{\etalchar{+}}20]{amanatidis2020maximum}
Georgios Amanatidis, Georgios Birmpas, Aris Filos{-}Ratsikas, Alexandros
  Hollender, and Alexandros~A. Voudouris.
\newblock Maximum nash welfare and other stories about {EFX}.
\newblock In {\em Proceedings of the Twenty-Ninth International Joint
  Conference on Artificial Intelligence, {IJCAI} 2020}, 2020.

\bibitem[ABFRV22]{fairSurvey2022}
Georgios Amanatidis, Georgios Birmpas, Aris Filos-Ratsikas, and {Alexandros A.}
  Voudouris.
\newblock {Fair Division of Indivisible Goods: A Survey}.
\newblock In {\em The International Joint Conference on Artificial Intelligence
  - Survey Track (IJCAI)}, pages 5385--5393, July 2022.

\bibitem[AG23]{akrami2023breaking}
Hannaneh Akrami and Jugal Garg.
\newblock Breaking the $3/4$ barrier for approximate maximin share.
\newblock {\em arXiv preprint arXiv:2307.07304}, 2023.

\bibitem[AGST23]{akrami2023simplification}
Hannaneh Akrami, Jugal Garg, Eklavya Sharma, and Setareh Taki.
\newblock Simplification and improvement of {MMS} approximation.
\newblock {\em arXiv preprint arXiv:2303.16788}, 2023.

\bibitem[ALMW22]{AzizSurvey2022}
Haris Aziz, Bo~Li, Herv\'{e} Moulin, and Xiaowei Wu.
\newblock {Algorithmic Fair Allocation of Indivisible Items: A Survey and New
  Questions}.
\newblock {\em SIGecom Exch.}, 20(1):24–40, November 2022.

\bibitem[ALW22]{ALW22}
Haris Aziz, Bo~Li, and Xiaowei Wu.
\newblock Approximate and strategyproof maximin share allocation of chores with
  ordinal preferences.
\newblock {\em Mathematical Programming}, pages 1--27, 2022.

\bibitem[AMNS17]{amanatidis2017approximation}
Georgios Amanatidis, Evangelos Markakis, Afshin Nikzad, and Amin Saberi.
\newblock Approximation algorithms for computing maximin share allocations.
\newblock {\em ACM Transactions on Algorithms (TALG)}, 13(4):1--28, 2017.

\bibitem[Azi20]{aziz2020simultaneously}
Haris Aziz.
\newblock Simultaneously achieving ex-ante and ex-post fairness.
\newblock In {\em Web and Internet Economics - 16th International Conference
  {(WINE)}}, 2020.

\bibitem[BEF21]{BEF2021b}
Moshe Babaioff, Tomer Ezra, and Uriel Feige.
\newblock Fair-share allocations for agents with arbitrary entitlements.
\newblock In {\em The {ACM} Conference on Economics and Computation
  ({ACM-EC})}, 2021.

\bibitem[BEF22]{BEF2021c}
Moshe Babaioff, Tomer Ezra, and Uriel Feige.
\newblock On best-of-both-worlds fair-share allocations.
\newblock In {\em Web and Internet Economics {(WINE)}}, 2022.

\bibitem[BF22]{BF22}
Moshe Babaioff and Uriel Feige.
\newblock {Fair Shares: Feasibility, Domination and Incentives}.
\newblock In {\em {ACM} Conference on Economics and Computation (ACM-EC)},
  2022.

\bibitem[BFHN23]{babichenko2023fair}
Yakov Babichenko, Michal Feldman, Ron Holzman, and Vishnu~V. Narayan.
\newblock Fair division via quantile shares, 2023.

\bibitem[BK20]{BK20}
Siddharth Barman and Sanath~Kumar Krishnamurthy.
\newblock Approximation algorithms for maximin fair division.
\newblock {\em ACM Transactions on Economics and Computation (TEAC)},
  8(1):1--28, 2020.

\bibitem[BL16]{bouveret2016characterizing}
Sylvain Bouveret and Michel Lema{\^\i}tre.
\newblock Characterizing conflicts in fair division of indivisible goods using
  a scale of criteria.
\newblock {\em Autonomous Agents and Multi-Agent Systems}, 30(2):259--290,
  2016.

\bibitem[BNT21]{BabaioffNT2020}
Moshe Babaioff, Noam Nisan, and Inbal Talgam{-}Cohen.
\newblock Competitive equilibrium with indivisible goods and generic budgets.
\newblock {\em Mathematics of Operations Research}, 46(1):382--403, February
  2021.

\bibitem[BT96]{brams1996fair}
Steven~J Brams and Alan~D Taylor.
\newblock {\em Fair Division: From cake-cutting to dispute resolution}.
\newblock Cambridge University Press, 1996.

\bibitem[Bud11]{Budish11}
Eric Budish.
\newblock The combinatorial assignment problem: Approximate competitive
  equilibrium from equal incomes.
\newblock {\em Journal of Political Economy}, 119(6):1061--1103, 2011.

\bibitem[CGM20]{EFX3}
Bhaskar~Ray Chaudhury, Jugal Garg, and Kurt Mehlhorn.
\newblock {{EFX} Exists for Three Agents}.
\newblock In {\em {ACM} Conference on Economics and Computation {(ACM-EC)}},
  page 1–19, 2020.

\bibitem[CKM{\etalchar{+}}19]{CaragiannisKMPS19}
Ioannis Caragiannis, David Kurokawa, Herv{\'{e}} Moulin, Ariel~D. Procaccia,
  Nisarg Shah, and Junxing Wang.
\newblock {The Unreasonable Fairness of Maximum Nash Welfare}.
\newblock {\em {ACM} Trans. Economics and Comput.}, 7(3):12:1--12:32, 2019.

\bibitem[FGH{\etalchar{+}}19]{farhadi2019fair}
Alireza Farhadi, Mohammad Ghodsi, Mohammad~Taghi Hajiaghayi, Sebastien Lahaie,
  David Pennock, Masoud Seddighin, Saeed Seddighin, and Hadi Yami.
\newblock Fair allocation of indivisible goods to asymmetric agents.
\newblock {\em Journal of Artificial Intelligence Research}, 64:1--20, 2019.

\bibitem[FH22]{FH23arxiv}
Uriel Feige and Xin Huang.
\newblock On picking sequences for chores.
\newblock {\em CoRR}, abs/2211.13951, 2022.

\bibitem[FH23]{FH23}
Uriel Feige and Xin Huang.
\newblock On picking sequences for chores.
\newblock In {\em The {ACM} Conference on Economics and Computation
  {(ACM-EC)}}, pages 626--655, 2023.

\bibitem[FN22]{FN22}
Uriel Feige and Alexey Norkin.
\newblock Improved maximin fair allocation of indivisible items to three
  agents.
\newblock {\em CoRR}, abs/2205.05363, 2022.

\bibitem[FST21]{FST21}
Uriel Feige, Ariel Sapir, and Laliv Tauber.
\newblock A tight negative example for {MMS} fair allocations.
\newblock In {\em Web and Internet Economics {(WINE)} 2021}, 2021.

\bibitem[FT14]{FeigeT14}
Uriel Feige and Moshe Tennenholtz.
\newblock On fair division of a homogeneous good.
\newblock {\em Games Econ. Behav.}, 87:305--321, 2014.

\bibitem[GHS{\etalchar{+}}18]{GhodsiHSSY18}
Mohammad Ghodsi, Mohammad~Taghi Hajiaghayi, Masoud Seddighin, Saeed Seddighin,
  and Hadi Yami.
\newblock Fair allocation of indivisible goods: Improvements and
  generalizations.
\newblock In {\em The {ACM} Conference on Economics and Computation
  {(ACM-EC)}}, pages 539--556, 2018.

\bibitem[GM19]{GM19}
Laurent Gourv{\`{e}}s and J{\'{e}}r{\^{o}}me Monnot.
\newblock On maximin share allocations in matroids.
\newblock {\em Theor. Comput. Sci.}, 754:50--64, 2019.

\bibitem[GMT19]{garg2019approximating}
Jugal Garg, Peter McGlaughlin, and Setareh Taki.
\newblock Approximating maximin share allocations.
\newblock {\em Open access series in informatics}, 69, 2019.

\bibitem[GT20]{GT20}
Jugal Garg and Setareh Taki.
\newblock An improved approximation algorithm for maximin shares.
\newblock In {\em The {ACM} Conference on Economics and Computation
  {(ACM-EC)}}, pages 379--380, 2020.

\bibitem[HS23]{HuangS23}
Xin Huang and Erel Segal{-}Halevi.
\newblock A reduction from chores allocation to job scheduling.
\newblock In Kevin Leyton{-}Brown, Jason~D. Hartline, and Larry Samuelson,
  editors, {\em Proceedings of the 24th {ACM} Conference on Economics and
  Computation, {EC} 2023, London, United Kingdom, July 9-12, 2023}, page 908.
  {ACM}, 2023.

\bibitem[KPW18]{KurokawaPW18}
David Kurokawa, Ariel~D. Procaccia, and Junxing Wang.
\newblock Fair enough: Guaranteeing approximate maximin shares.
\newblock {\em J. {ACM}}, 65(2):8:1--8:27, 2018.

\bibitem[LMMS04]{LMMS04}
Richard~J. Lipton, Evangelos Markakis, Elchanan Mossel, and Amin Saberi.
\newblock On approximately fair allocations of indivisible goods.
\newblock In {\em Proceedings of the 5th {ACM} Conference on Electronic
  Commerce}, pages 125--131, 2004.

\bibitem[Mou04]{moulin2004fair}
Herv{\'e} Moulin.
\newblock {\em Fair division and collective welfare}.
\newblock MIT press, 2004.

\bibitem[PR20]{PR2020}
Benjamin Plaut and Tim Roughgarden.
\newblock {Almost Envy-Freeness with General Valuations}.
\newblock {\em SIAM Journal on Discrete Mathematics}, 34(2):1039--1068, 2020.

\bibitem[UF23]{bUF23}
Gilad~Ben Uziahu and Uriel Feige.
\newblock On fair allocation of indivisible goods to submodular agents.
\newblock {\em CoRR}, abs/2303.12444, 2023.

\bibitem[Var73]{varian1973equity}
Hal~R Varian.
\newblock {\em Equity, envy, and efficiency}.
\newblock Cambridge, MIT, 1973.

\end{thebibliography}
\newcommand{\etalchar}[1]{$^{#1}$}

\appendix

\section{The  AnyPrice share (APS)}\label{app:aps}

We present the definition of the AnyPrice share (APS) that is based on prices. 
Defines the share of an agent $i$ to be the value she can guarantee to herself whenever her budget is set to her entitlement  $b_i$ (when $\sum_i b_i=1$) and she buys her highest-value affordable set when items are {adversarially} priced with a total price of $1$.  
Let $\prices=\{(p_1,p_2,\ldots,p_m) | p_j\geq 0\ \forall j\in\items,\ \  \sum_{j\in \items} p_j=1\}$ be the set of item-price vectors that total to $1$.   
The price-based definition of AnyPrice share is the following:  
\begin{definition}[AnyPrice share]
	 \label{def:APS}
	Consider a setting in which agent $i$ with valuation $v_i$ has entitlement $b_i$ to a set of indivisible items $\items$.
	The \emph{AnyPrice share (APS)} of agent $i$, denoted by $\anypricei$, is the value she can guarantee herself whenever the items in $\items$ are adversarially priced with a total price of $1$, and she picks her favorite affordable bundle: 
	$$\anypricei = \min_{(p_1,p_2,\ldots,p_m)\in \prices}\ \ \max_{S\subseteq \items} \left\{v_i(S) \Big | \sum_{j\in S} p_j\leq b_i\right\}$$
\end{definition}

\section{Missing proofs from Section \ref{sec:domination}}\label{app:domination}

\subsection{Personalized shares}\label{app:domination-person}

\subsubsection{Goods}\label{sec:person-goods}

We first develop the notion of personalized shares for settings in which items are goods.

We start by considering the case in which the entitlement $b_i$ of an agent is of the form $\frac{1}{k}$ for some integer $k$ (as would happen in the equal entitlements case, where $k$ would be the number of agents). 
With every valuation function $v$ we associate a ``personalized" feasible (unrestricted) share that is the ``most desirable" feasible (unrestricted) share function for agents holding the valuation function $v$. The following definition applies to ex-post settings.

\begin{definition}
\label{def:personalMMS}
    For a valuation function $v$ and entitlement $b = \frac{1}{k}$ where $k$ is a positive integer, the value of the {\em personalized MMS} share function $MMS_{v,\frac{1}{k}}$ is computed as follows. We say that a bundle $S$ is an {\em acceptable $MMS(v,\frac{1}{k})$ bundle} if $v(S) \ge MMS(v,\frac{1}{k})$. Then, for every valuation $v'$ we have that $MMS_{v,\frac{1}{k}}(v',\frac{1}{k}) = \min_S[v'(S)]$, where $S$ ranges over all acceptable $MMS(v,\frac{1}{k})$ bundles. In particular, for $v'=v$ we have that $MMS_{v,\frac{1}{k}}(v,\frac{1}{k}) = MMS(v,\frac{1}{k})$.
\end{definition}

A way to think of the $MMS_{v,\frac{1}{k}}$ (unrestricted) share is as follows. The agent ``believes" that $v$ it the ``true" valuation function (e.g., because $v$ represents current prices of the items in the free market, or because this is what she was told by an appraiser that she trusts). Hence the allocation should be fair under this specific valuation $v$. The best that one can guarantee that each one of $k$ equally entitled agents can receive is a bundle that is an  acceptable $MMS(v,\frac{1}{k})$ bundle. Hence receiving such a bundle is ``fair". If some other agent has a different valuation function $v'$, then holding such a $v'$ (which is not the ``true" valuation, but a misguided one), does not make that other agent entitled to refuse to receive a bundle that is fair under $v$. 

\begin{proposition}
    \label{pro:personalMMS}
    For every valuation $v$, {for any class $C$ of valuations functions,} the personalized MMS share $MMS_{v,\frac{1}{n}}$  is feasible for $C$ (ex-post) for allocation instances with $n$ agents that have equal entitlement. 
\end{proposition}

\begin{proof}
    Consider an arbitrary $MMS(v,\frac{1}{n})$ partition. {For any valuations $(v'_1,\ldots,v'_n)$ of the agents, giving each} agent one of the bundles of this partition is an allocation that is acceptable under $MMS_{v,\frac{1}{n}}$. 
\end{proof}

For the ex-ante setting, a similar reasoning leads to the personalized MES share.

\begin{definition}
\label{def:personalMES}
    For a valuation function $v$ and entitlement $b = \frac{1}{k}$ where $k$ is a positive integer, the value of the {\em personalized MES} share function $MES_{v,\frac{1}{k}}$ is computed as follows. 
    For every valuation $v'$ we have that $MES_{v,\frac{1}{k}}(v',\frac{1}{k}) = \min_P[E_{S \sim P}[v'(S)]]$. Here, $P$ ranges over all $MES(v,\frac{1}{k})$ partitions, and $S\sim P$ denotes a bundle sampled uniformly at random among the $k$ bundles of the partition $P$. In particular, for $v'=v$ we have that $MES_{v,\frac{1}{k}}(v,\frac{1}{k}) = MES(v,\frac{1}{k})$. 
\end{definition}

If $v$ is an additive valuation, then all $k$-partitions give the same expected value. In this case, for every additive valuations $v$ and $v'$,  the $MES_{v,\frac{1}{k}}$ value of $v'$ equals the proportional share, $MES_{v,\frac{1}{k}}(v', \frac{1}{k}) = PS(v',\frac{1}{k})$.

\begin{proposition}
    \label{pro:personalMES}
    For every valuation $v$, {for any class $C$ of valuations functions,} the personalized MES share $MES_{v,\frac{1}{n}}$  is feasible for $C$ (ex-ante) for allocation instances with $n$ agents that have equal entitlements. 
\end{proposition}

\begin{proof}
    Consider an arbitrary $MES(v,\frac{1}{n})$ partition. {For any valuations $(v'_1,\ldots,v'_n)$ of the agents, giving every agent one of the bundles of this partition uniformly at random,} is a randomized allocation that is acceptable under $MES_{v,\frac{1}{n}}$.
\end{proof}

We now explain how we extend personalized shares to {arbitrary entitlements, those} that are not of the form $\frac{1}{k}$ for some integer $k$. For convenience, we restate here Definition~\ref{def:unit}.

\begin{definition}
    For entitlement $0 < b \le 1$, let $k = \lfloor \frac{1}{b} \rfloor$. 
    We define the {\em unit upper bound} of $b$ to be $\frac{1}{k}$, and denote it by $\hat{b}$. 
    We also define the ex-post share $\widehat{MMS}$ to be $\widehat{MMS}(v_i,b_i) = MMS(v_i,\hat{b}_i)$ and the ex-ante share $\widehat{MES}$ to be $\widehat{MES}(v_i,b_i) = MES(v_i,\hat{b}_i)$.
\end{definition}
{Observe that $k = \lfloor \frac{1}{b} \rfloor$ is the unique integer $k$ such that $\frac{1}{k+1}< b \le \frac{1}{k}$.}

The following definition extends Definition~\ref{def:personalMMS} to the case that the entitlement $b_i$ is not the inverse of an integer. 

\begin{definition}
\label{def:verypersonalMMS}
    For a valuation function $v_i$ and entitlement $b_i$ ($0 < b_i \le 1$), 
    the value of the {\em personalized MMS} share function $MMS_{v_i,b_i}$ is computed as follows. 
    Let integer $k$ be such that $\frac{1}{k+1} < b_i \le \frac{1}{k}$ (and hence $\hat{b}_i = \frac{1}{k}$). For every integer $\ell \in \{1, \ldots, k\}$, let ${\cal{F}}_{v_i,k,\ell}$ denote the set of all those bundles $B$ that have the property that $B$ is a disjoint union of $\ell$ bundles $B_1, \ldots, B_{\ell}$, where each of these $\ell$ bundles is an acceptable $MMS(v_i, \frac{1}{k})$ bundle.
    Consider arbitrary valuation $v_j$ and entitlement $b_j$. Let integer $f_j$ be such that $\frac{f_j}{k+1} < b_j \le \frac{f_j+1}{k+1}$. Then 
    $$MMS_{v_i,b_i}(v_j,b_j) = \min_{B \in {\cal{F}}_{v_i,k,f_j}} v_j(B)$$
    In particular, for $v_j=v_i$ and $b_j = b_i$ we have that $MMS_{v_i,b_i}(v_i,b_i) = \widehat{MMS}(v_i,b_i)$.
\end{definition}

Here is a way of interpreting Definition~\ref{def:verypersonalMMS} for $MMS_{v_i,b_i}$. 
For $b_i$ that is not the inverse of an integer, let integer $k$ be such that $\frac{1}{k+1} < b_i < \frac{1}{k}$. Agent $i$ thinks of $v_i$ as the ``true" valuation (as in the discussion following Definition~\ref{def:personalMMS}). The selection of the allocation is envisioned as being done in three steps. At first, each agent sends ``representatives" to a ``committee" that includes up to $k$ equally entitled representatives. Then, the committee selects an arbitrary $MMS(v_i,\frac{1}{k})$ partition of $\items$, and each representative gets one bundle of this partition. Finally, each agent receives the bundles of {all} her representatives. The question that remains is how many representatives should an agent with entitlement $b$ be allowed to send. To be ``fair", we want this number to be a non-decreasing function of $b$. This function should obey the constraint that the total number of representatives does not exceed $k$, regardless of the distribution of entitlements among the $n$ agents. Definition~\ref{def:verypersonalMMS} uses a step function, obtained by setting the ``price" of a representative to be just marginally above $\frac{1}{k+1}$ units of entitlements. The total entitlements of all agents combined allows them to ``purchase" at most $k$ representatives. Under this pricing, the number of representatives that agent $j$ can afford to purchase is precisely the value $f_j$ of Definition~\ref{def:verypersonalMMS}. 

\begin{proposition}
    \label{pro:verypersonalMMS}
    For every valuation $v$ and entitlement $b$, {for any class $C$ of valuation functions} the personalized MMS share $MMS_{v,b}$ is feasible for $C$ (ex-post). 
\end{proposition}

\begin{proof}
    Let $k$ be such that $\hat{b} = \frac{1}{k}$. Consider an arbitrary allocation instance, where the $n$ agents have valuations $(v_1,\ldots,v_n)$ and entitlements $(b_1, \ldots, b_n)$.  For every agent $j$, let integer $f_j$ be such that $\frac{f_j}{k+1} < b_j \le \frac{f_j+1}{k+1}$. Necessarily, $\sum_j f_j \le k$, because $\sum_j b_j = 1$.
    Consider an arbitrary $MMS(v,\frac{1}{k})$ partition of $\items$. For each agent $j$, give it $f_j$ distinct bundles from the $MMS(v,\frac{1}{k})$ partition. (If $\sum f_j < k$, then allocate the remaining items arbitrarily.) Thus, each agent $j$ gets a bundle that is the disjoint union of $f_j$ bundles, where each of these bundles is acceptable under $MMS(v,\frac{1}{k})$.  This satisfies the requirements of  Definition~\ref{def:verypersonalMMS} for being acceptable under $MMS_{v,b}$.  
\end{proof}

We can use a similar approach to extend the personalized MES (ex-ante) share to entitlements that are not the inverse of an integer. We omit the straightforward but tedious details of how this is done for arbitrary classes of valuations, and instead present only the special case of additive valuation, in which the definition simplifies to give a personalized proportional share.

\begin{definition}
\label{def:personalPS}
    For a valuation function $v_i$ and entitlement $b_i$, 
    the value of the {\em personalized PS} share function $PS_{v_i,b_i}$ is computed as follows. 
    Let integer $k$ be such that $\frac{1}{k+1} < b_i \le \frac{1}{k}$ (and hence $\hat{b}_i = \frac{1}{k}$). 
    Consider arbitrary additive valuation $v_j$ and entitlement $b_j$. Let integer $f_j$ be such that $\frac{f_j}{k+1} < b_j \le \frac{f_j+1}{k+1}$. Then 
    $$PS_{v_i,b_i}(v_j,b_j) = \frac{f_j}{k} v_j(\items)$$
    Observe that $PS_{v_i,b_i}(v_j,b_j)$ does not depend on $v_i$, and hence it does not matter whether $v_i$ is additive or not, as long as $v_j$ is additive. 
\end{definition}

We remark that for non-additive classes of valuations, the definition of personalized MES is such that the values $MES_{v_i,b_i}(v_j,b_j)$ do depend on $v_i$.

\begin{proposition}
    \label{pro:personalPS-goods}
    For every valuation $v$ and entitlement $b$, the personalized PS share $PS_{v,b}$ is feasible for the class of additive valuations {(ex-ante)}. 
\end{proposition}

\begin{proof}
    Let $k$ be such that $\hat{b} = \frac{1}{k}$. Consider an arbitrary allocation instance, where the $n$ agents have additive valuations $(v_1,\ldots,v_n)$ and entitlements $(b_1, \ldots, b_n)$.  For every agent $j$, let integer $f_j$ be such that $\frac{f_j}{k+1} < b_j \le \frac{f_j+1}{k+1}$. Necessarily, $\sum_j f_j \le k$, because $\sum_j b_j = 1$. Give each agent $j$ the grand bundle $\items$ with probability $\frac{f_j}{k}$. (If $\sum f_j < k$, then with the remaining probability select an arbitrary allocation.) This satisfies the requirements of  Definition~\ref{def:personalPS} for being acceptable ex-ante under $PS_{v,b}$.  
\end{proof}

\subsubsection{Chores}\label{sec:person-chores}

We now explain how the notion of personalized shares is adapted to settings in which items are chores. Our presentation for chores will be less detailed than our presentation for goods. We start by recalling Definition~\ref{def:unitLowerBound}.

\begin{definition}
    For {responsibility} $0 < b < 1$, let $k$ be such that $\frac{1}{k+1} \le b < \frac{1}{k}$.
    Then the {\em unit lower bound} on $b$, denoted by $\check{b}$, is $\frac{1}{k+1}$. Define the ex-post share $\overline{MMS}$ as $\overline{MMS}(c_i,b_i) = MMS(c_i,\check{b}_i)$ and the ex-ante share $\overline{MES}$ as $\overline{MES}(c_i,b_i) = MES(c_i,\check{b}_i)$.
\end{definition}

For additive cost functions MES is the same as the proportional share PS, and so $\overline{MES}$ can be replaced by $\overline{PS}(c_i,b_i) = PS(c_i,\check{b}_i)$.

The following definition is the adaptation of Definition~\ref{def:verypersonalMMS} to the case of chores.

\begin{definition}
\label{def:personalMMSchores}
    For a cost function $c_i$ and responsibility $b_i$ ($0 < b_i < 1$), 
    the value of the {\em personalized MMS} share function $MMS_{c_i,b_i}$ is computed as follows. 
    Let integer $k$ be such that $\frac{1}{k+1} \le b_i < \frac{1}{k}$ (and hence $\check{b}_i = \frac{1}{k+1}$). For every integer $\ell \in \{1, \ldots, k+1\}$, let ${\cal{F}}_{v_i,k,\ell}$ denote the set of all those bundles $B$ that have the property that $B$ is a disjoint union of $\ell$ bundles $B_1, \ldots, B_{\ell}$, where each of these $\ell$ bundles is an acceptable $MMS(c_i, \frac{1}{k+1})$ bundle.
    Consider arbitrary cost function $c_j$ and responsibility $b_j$. Let integer $f_j$ be such that $\frac{f_j-1}{k} \le b_j < \frac{f_j}{k}$. Then 
    $$MMS_{c_i,b_i}(c_j,b_j) = \max_{B \in {\cal{F}}_{c_i,k,f_j}} c_j(B)$$
    In particular, for $c_j=c_i$ and $b_j = b_i$ we have that $MMS_{c_i,b_i}(c_i,b_i) = \overline{MMS}(c_i,b_i)$.
\end{definition}

\begin{proposition}
    \label{pro:personalMMSchores}
    For every cost function $c$ and responsibility $b$, {for any class $C$ of cost functions} the personalized MMS share $MMS_{c,b}$ is feasible for $C$ (ex-post). 
\end{proposition}

\begin{proof}
    Let $k$ be such that $\check{b} = \frac{1}{k+1}$. Consider an arbitrary allocation instance, where the $n$ agents have cost functions $(c_1,\ldots,c_n)$ and responsibilitiess $(b_1, \ldots, b_n)$.  For every agent $j$, let integer $f_j$ be such that $\frac{f_j-1}{k} \le b_j < \frac{f_j}{k}$. Necessarily, $\sum_j f_j \ge k+1$, because $\sum_j b_j = 1$.
    Consider an arbitrary $MMS(c,\frac{1}{k+1})$ partition of $\items$. For each agent $j$, give it at most $f_j$ distinct bundles from the $MMS(c,\frac{1}{k+1})$ partition, exhausting all bundles in the $MMS(c,\frac{1}{k+1})$ partition. Thus, each agent $j$ gets a bundle that is the disjoint union of at most $f_j$ bundles, where each of these bundles is acceptable under $MMS(c,\frac{1}{k+1})$.  This satisfies the requirements of  Definition~\ref{def:personalMMSchores} for being acceptable under $MMS_{c,b}$.  
\end{proof}

The following definition is the adaptation of Definition~\ref{def:personalPS} to the case of chores.

\begin{definition}
\label{def:personalPSchores}
    For a cost function $c_i$ and responsibility $b_i$, 
    the value of the {\em personalized PS} share function $PS_{c_i,b_i}$ is computed as follows. 
    Let integer $k$ be such that $\frac{1}{k+1} \le b_i < \frac{1}{k}$ (and hence $\check{b}_i = \frac{1}{k+1}$). 
    Consider arbitrary additive cost functions $c_j$ and responsibility $b_j$. Let integer $f_j$ be such that $\frac{f_j-1}{k} \le b_j < \frac{f_j}{k}$. Then 
    $$PS_{c_i,b_i}(c_j,b_j) = \frac{f_j}{k+1} c_j(\items)$$
    Observe that $PS_{c_i,b_i}(c_j,b_j)$ does not depend on $c_i$, and hence it does not matter whether $c_i$ is additive or not, as long as $c_j$ is additive. 
\end{definition}

\begin{proposition}
    \label{pro:personalPS-chores}
    For every cost function $c$ and responsibility $b$, the personalized PS share $PS_{v,b}$ is feasible for the class of additive cost functions {(ex-ante)}. 
\end{proposition}

\begin{proof}
    Let $k$ be such that $\check{b} = \frac{1}{k+1}$. 
    Consider an arbitrary allocation instance, where the $n$ agents have additive cost functions $(c_1,\ldots,c_n)$ and responsibilities $(b_1, \ldots, b_n)$.  For every agent $j$, let integer $f_j$ be such that $\frac{f_j-1}{k} \le b_j < \frac{f_j}{k}$. Necessarily, $\sum_j f_j \ge k+1$, because $\sum_j b_j = 1$. Give each agent $j$ the grand bundle $\items$ with probability $0 \le p_j \le \frac{f_j}{k+1}$, for values of $p_j$ satisfying $\sum_j p_j = 1$.  This satisfies the requirements of  Definition~\ref{def:personalPSchores} for being acceptable ex-ante under $PS_{c,b}$.  
\end{proof}

\subsubsection{Nice shares for additive valuations}
\label{sec:nice}

Recall that we refer to a share as nice if it is name independent (see Definition~\ref{def:orderedShare}) and self maximizing (a property introduced in~\cite{BF22} and briefly discussed in Section~\ref{sec:model-share-props}). In this section we consider additive valuations {(both for goods and chores)}, and show that in this case all personalized shares are either nice, or can be replaced with nice shares that are feasible and are of the same value as the unrestricted personalized shares.

We start with the ex-ante case. In this case, the personalized share $PS_{v_i,b_i}$ is already name independent. As to being self maximizing, even though the proportional share is not self maximizing as an ex-post share~\cite{BF22}, it is self maximizing as an ex-ante share (when valuations are additive). This follows from the fact that to give an agent with valuation $v_j$ and entitlement $b_j$ {exactly her $PS_{v_i,b_i}$ share (ex-ante, according to her true $v_j$), one does not need to know $v_j$. It suffices to give her each of the items with probability $\hat{b}_j$ for goods, or $\check{b}_j$ for chores. Doing so is clearly an option even if the agent reports an incorrect $v'_j$. Hence, the agent cannot force the expected value that she receives (with respect to the true $v_j$) to be larger than $PS_{v_i,b_i}$, by misreporting of her valuation function.} 

Achieving name independence in the ex-post case is handled by a method of~\cite{bouveret2016characterizing}, that was also used for a similar purpose in~\cite{BF22}. {For completeness we {sketch} the basic idea here as well.} 
Instead of allocating $m$ items, allocate $m$ {\em coupons}. These coupons can later be used in order to determine a picking sequence among agents, where in every round $r$, the agent that gets to pick an item is the agent that holds coupon $r$. Hence the ``value" of coupon $r$ to agent $i$ is that of the $r$th most valuable item according to $v_i$, because the agent is guaranteed to be able to select such an item in exchange for her coupon, but not guaranteed to be able to select any better item. Coupon values are ``name independent" (a coupon $r$ has the same value for equivalent valuations). Interpreting the personalized share $MMS_{(v_i,b_i)}$ as determining which sets of coupons are acceptable (for example, if the bundle containing the three most valuable items under $v_i$ is acceptable, then the set of coupons $\{1,2,3\}$ is acceptable) transforms the unrestricted personalized share to a name independent share. We refer to this name independent share as $OMMS_{(v_i,b_i)}$, where $O$ stands for {\em ordered}, as the share involves ordering items by their values. {The fact that OMMS is self maximizing is a consequence of a result in~\cite{BF22} that states that $\omega$-picking-order shares are self maximizing. Further details are omitted.} 

\begin{corollary}
    \label{cor:nice}
    The following claims hold for personalized shares when agents have additive valuations and entitlement $b_i$. 
    \begin{enumerate}
        \item For any additive valuation $v_i$ over goods, the share 
        $OMMS_{(v_i,b_i)}$ is a nice (name independent and self maximizing) ex-post share that is feasible for additive goods, and satisfies $OMMS_{(v_i,b_i)}(v_i,b_i) = \widehat{MMS}(v_i,b_i)$.
        \item For any additive cost function $c_i$ over chores, the share $OMMS_{(c_i,b_i)}$ is a nice ex-post share that is a feasible for additive chores, and satisfies
        $OMMS_{(c_i,b_i)}(c_i,b_i) = \overline{MMS}(c_i,b_i)$.
        \item For any additive valuation $v_i$ over goods, the share $PS_{(v_i,b_i)}$ is a nice ex-ante share that is feasible for additive goods, and satisfies $PS_{(v_i,b_i)}(v_i,b_i) = \widehat{PS}(v_i,b_i)$ 
        \item For any additive cost function $c_i$ over chores, the share $PS_{(c_i,b_i)}(c_i,b_i)$ is a nice ex-ante share that is feasible for addiitve chores, and satisfies $PS_{(c_i,b_i)}(c_i,b_i) = \overline{PS}(c_i,b_i)$.
    \end{enumerate}
\end{corollary}

\subsection{Domination for Non-additive Valuations}\label{sec:non-add}

\begin{observation}\label{obs:MMS-non-add}
    There exists a class $C$ of valuation over goods (that is not additive) for which ${MMS}$ is \emph{not} the minimal ex-post share for equal entitlements that dominates every feasible \emph{name-independent} ex-post share (for class $C$). Additionally, for this class $C$ the share ${MES}$ is \emph{not} the minimal ex-ante share for equal entitlements that dominates every feasible \emph{name-independent} ex-ante share (for class $C$). 
\end{observation}

\begin{proof}
We show that for some class of valuations over goods and equal entitlements, it is possible to dominate every feasible share without dominating ${MMS}$. Consider the valuation $v$ defined over 4 items $\items=\{A,B,C,D\}$.  
The valuation is $v(\{A,B\})=v(\{C,D\})=1$, and $v(S)=1$ for any $S\subseteq \items$ such that $|S|>2$, and $0$ otherwise. 
Note that for two agent with equal entitlements $b=b_1=b_2=1/2$ it holds that ${MMS}(v,b)={MMS}(v,1/2)=1$.
Consider the class of valuations that includes the valuation $v$ and all valuations resulting from a permutation over the item names. We claim that the only {name independent} feasible share for two equally-entitled agents is the zero share.
Consider the valuation $v'$ that is defined as
$v'(\{A,C\})=v'(\{B,D\})=1$, and $v'(S)=1$ for any $S\subseteq \items$ such that $|S|>2$, and $0$ otherwise. 
As $v'$ is a result of a permutation over item names in $v$, for any name-independent share $s$ it holds that $s(v,1/2)=s(v',1/2)$ (in particular, ${MMS}(v',b)={MMS}(v',1/2)=1$). 
For any feasible share $s$, if $s(v,1/2)=s(v',1/2)>0$ then on the instance with two valuations $v$ and $v'$, it must give both agents positive value, but there is no allocation that gives both positive value.  
We conclude that the only feasible name-independent share for this class is the zero share, and it does not dominate the ${MMS}$.

The claim about ${MES}$ follow similarly, since ${MES}(v',b)={MES}(v',1/2)=1$ but for any feasible name-independent ex-ante share $s$, {it holds that $s(v,1/2)=s(v',1/2)$ and $s(v,1/2)+s(v',1/2) \le 1$, so  $s(v,1/2)=s(v',1/2)\leq \frac{1}{2}<1 = {MES}(v',b)={MES}(v',1/2)$.}
\end{proof}

\section{The  Bidding Game}\label{app:bidding}

\subsection{Myopic strategies are too weak}
\label{sec:myopic}

We show that no myopic strategy can guarantee $\frac{1}{2}$-$\widehat{TPS}$. 

Consider agent~1 with entitlement $0.51$ and agent~2 with entitlement $0.49$, both having the same additive valuation $v$ with $v(\items) = 1$. To get $\frac{1}{2}$-$\widehat{TPS}$, agent~1 needs to reach a value of at least $\frac{1}{2}$. What bid should agent~1 place if the most valuable item has value $0.34$? This depends on values of other items, and hence, the bidding strategy cannot be myopic.

If there are two other items of value $0.33$, agent~1 must bid not more than $0.27$, as otherwise agent~2 will let agent~1 win this bid, and agent~2 will collect the other two items by bidding $0.24$ on each of them. 

If there are $66$ other items each of value $0.01$, agent~1 must bid at least $0.3$, as otherwise agent~2 will bid $0.3$ and win the first item. Then agent~2 can bid $0.011$ on each of the remaining items, until she wins $17$ of them. Agent~2 has sufficient budget to afford~17 items, and agent~1 does not have sufficient budget to prevent this by winning more than $66 - 17 = 49$ items at these prices.

\subsection{Proof of Claim \ref{claim:k2}}\label{app:proof-k2}
Familiarity with the proof of Claim~\ref{claim:largek} can help in following the proof of Claim \ref{claim:k2} below. 

    \mbfuture{check this proof} 

    \begin{proof}
    The setting of the claim corresponds to entitlement $b_i$ satisfying $\frac{1}{k+1} < b_i \le \frac{1}{k}$ for $k = 2$.   
    The item value bounds then imply that $v_i(e_1) < \frac{1}{4}$, $v_i(e_2) + v_i(e_3) < \frac{1}{4}$, and $v_i(e_j) < \frac{1}{8}$ for all $j \ge 3$. We shall make use of one more item value bound, which is $v_i(e_8) < \frac{1}{12}$. This bound can be assumed to hold, because otherwise agent $i$ can bid $\frac{1}{9}$ on each of the first eight items. The adversary has a budget smaller than $\frac{2}{3}$, and can win an most five of these item. Hence, the agent can win three of the first eight items, for a value of at least $3 \cdot \frac{1}{12} = \frac{1}{4}$, as desired. Summarizing, in our case of $\frac{1}{3} < b_i \le \frac{1}{2}$, we have the following item value bounds.

 \begin{enumerate}
        \item $v_i(e_1) < \frac{1}{4}$.
        \item $v_i(e_2) + v_i(e_3) < \frac{1}{4}$.
        \item $v_i(e_j) < \frac{1}{8}$ for all $j \ge 3$. (This is implied by bound 2 above.) 
        \item $v_i(e_j) < \frac{1}{12}$ for all $j \ge 8$.
    \end{enumerate}

    To describe the bidding strategy of agent $i$, it will be convenient to represent the value of an item as $\frac{x}{72}$, where $x < 18$ (as all items have value smaller than $\frac{1}{4} = \frac{18}{72}$). The value of the bid will be represented as $\frac{f(x)}{72}$. The function $f(x)$ grows linearly with $x$, except around two special intervals of width $\frac{1}{72}$, in which it remains constant. These constant values are $\frac{1}{9}$ (corresponding to $f(x)=8$) and $\frac{1}{6}$ (corresponding to $f(x)=12$).  
    The function $f(x)$ is described in Table~\ref{tab:f(x)}.

    \begin{table}[hbt!]
        \centering
        \begin{tabular}{c|ccccc}
               & $0 < x < 8$ & $8 \le x < 9$ &  $9 \le x < 13$ & $13 \le x < 14$ & $14 \le x < 18$\\
               \hline
            $f(x)$ & $x$ & 8 & $x-1$ & 12 & $x-2$  
        \end{tabular}
        \caption{The bid $\frac{f(x)}{72}$ as a function of the value $\frac{x}{72}$ of the item.}
        \label{tab:f(x)}
        \end{table}

    There are two natural exceptions to using bids as in Table~\ref{tab:f(x)}. In both exceptions, agent $i$ bids all her remaining budget.

    \begin{itemize}
    \item If at a round $j$, the remaining budget of the agent is too small in order to place the bid dictated by Table~\ref{tab:f(x)}, bid the entire remaining budget in round $j$.
         \item If at a round $j$, winning item $e_j$ suffices for the agent in order to reach the desired value of $\frac{1}{4}$, bid the entire remaining budget in round $j$. 
    \end{itemize}
    


    Recall the notion of the surplus {introduced before} Claim~\ref{claim:largek} (the access value according to $v_i$ that the adversary consumes, compared to the adversary's payment) 
    {and the fact that} it suffices to show that the adversary cannot accumulate a surplus larger than $b_i - \frac{1}{2k}$. {In our case of $k=2$, this means that in suffices to bound the surplus by} ${\frac{1}{3}- \frac{1}{4} = } \frac{1}{12} = \frac{6}{72}$. We will not be able to show that the adversary's surplus is at most $\frac{1}{12}$. Instead we will show that either the adversary's surplus is at most $\frac{1}{12}$, or that after some items are consumed, the resulting sub-instance that remains is one for which there is a strategy that (combined with the items that the agent already won) gives the agent a value of at least $\frac{1}{4}$.

    \mbfuture{CHECK FROM HERE}

    There are several types of events in which the adversary can gain surpluses. We first list the types of surplus events that may happen when the agent has sufficient budget in order to bid according to Table~\ref{tab:f(x)}.

    \begin{enumerate}

    \item Type $S$ (where $S$ stands for Single surplus): an item $e_j$ had value $\frac{8}{72} < v_i(e_j) \le \frac{13}{72}$, the agent bid according to Table~\ref{tab:f(x)}, the adversary won the item and paid the bid of the agent. In this event, the adversary paid at least $\frac{8}{72}$, and the surplus is at most $\frac{1}{72}$. Observe that by the item value bounds, event $S$ can only happen for items $e_1, \ldots, e_7$, as $v_i(e_8) < \frac{1}{12} = \frac{6}{72} < \frac{8}{72}$.
        
        \item Type $D$ (where $D$ stands for Double surplus): an item $j$ had value $\frac{13}{72} < v_i(e_j) \le \frac{18}{72}$, the agent bid according to Table~\ref{tab:f(x)}, the adversary won the item and paid the bid of the agent. In this event, the adversary paid at least $\frac{12}{72}$, and the surplus is at most $\frac{2}{72}$. Observe that by the item value bounds, event $D$ can only happen for items $e_1$ and/or $e_2$, as $v_i(e_3) < \frac{1}{8} = \frac{9}{72} < \frac{13}{72}$. When it happens for item~1 we refer to it as $D_1$, and when it happens for item~2 we refer to it as $D_2$. 
        
    \end{enumerate}

We now list the types of surplus events that may happen when the agent already won some items, and does not have sufficient budget left in order to bid according to Table~\ref{tab:f(x)}.

    \begin{enumerate}

         \item Type $O$ (where $O$ stands for One item): the agent who previously won one item $e$ does not have sufficient budget left in order to bid on an item $e'$ according to Table~\ref{tab:f(x)} (in particular, the budget left is smaller than $v_i(e')$), the adversary wins $e'$ and pays less than $v_i(e')$.  An event $O$ may only happen if $e = e_1$, and $e' = e_2$. This is because according to Table~\ref{tab:f(x)}, the agent never pays more than $\frac{16}{72}$ for a single item, leaving the agent a budget of at least $\frac{8}{24}$. By the item value bounds, $v_i(e_3) < \frac{1}{8} = \frac{9}{72}$, and by Table~\ref{tab:f(x)} the bid in this case is at most $\frac{8}{72}$, so the agent can afford the bid. The surplus {from item $e_2$ when $O$ happens} 
        can be as large as $v_i(e_2) - \frac{1}{9}$. 
        As $v_i(e_2)$ might have value of nearly $\frac{1}{4}$, this surplus might be larger then our limit of $\frac{1}{12}$. 
        

        \item Type $T$ (where $T$ stands for Two items): the agent who previously won two items (but still did not reach a value of $\frac{1}{4}$) does not have sufficient budget in order to bid on an item $e'$ according to Table~\ref{tab:f(x)} (in particular, the budget left is smaller than $v_i(e')$), the adversary wins $e'$ and pays less than $v_i(e')$.  
        An event $O$ may only happen if $e' \in \{e_3, \ldots, e_7\}$. 
        The fact that $e' \not\in \{e_1, e_2\}$ holds because $O$ happens only if the agent previously won two items.
        The claim that $e' \not\in \{e_8, \ldots, e_m\}$ follows from the item value bound $v_i(e_8) \le \frac{1}{12}$. The agent who won two items and has not reached a value of $\frac{1}{4}$ still has a budget of at least $\frac{1}{3} - \frac{1}{4} = \frac{1}{12}$ left (because the agent never pays for an item more than its value), and hence can afford to bid on $e_8$.
        
        The surplus {from item $e'$ when $T$ happens} is at most $\frac{1}{72}$. Let us explain how we deduce this upper bound. For $y > 0$ and $z \ge 0$, let $\frac{8+y}{72}$ and $\frac{8-z}{72}$ denote the payments of the agent on the first item and second items that she won, respectively. Note that necessarily $y > 0$, as otherwise the agent can afford to win three items (item values are non-increasing), and that necessarily $z \ge 0$, because if $z < 0$ each of the items that the agent won has value at least $\frac{9}{72}$ (by inspection of Table~\ref{tab:f(x)}), and the agent already has a value of $\frac{1}{4}$. As the value that the agent has is assumed to be at most $\frac{1}{4}$, and the value of the first item is at least $\frac{8+y}{72} + \frac{1}{72}$ (by inspection of Table~\ref{tab:f(x)}). we have that $\frac{8+y}{72} + \frac{1}{72} + \frac{8-z}{72} < \frac{18}{72}$, implying that $y-z < 1$. As the bid on $e'$ is at least $\frac{1}{3} - \frac{8+y}{72} - \frac{8-z}{72} = \frac{8 - y + z}{72}$ and $v_i(e') \le \frac{8-z}{72}$ (as item values are non-increasing), the surplus is at most $\frac{8-z}{72} - \frac{8 - y + z}{72} = \frac{y - 2z}{72} < \frac{y - z}{72} < \frac{1}{72}$.
    \end{enumerate}

    Observe that there is no surplus event in which the agent previously won three or more items (but still did not reach a value of $\frac{1}{4}$) and does not have sufficient budget in order to bid on an item $e'$ according to Table~\ref{tab:f(x)}. This is because not reaching a value of $\frac{1}{4}$, the agent has a budget of at least $\frac{1}{3} - \frac{1}{4} = \frac{1}{12}$ left, whereas $v_i(e') < \frac{1}{12}$. (The fact that item values are non-increasing implies that the agent paid at most $\frac{1}{12}$ for the third item that she won, and $e'$ cannot be more valuable than this third item.)

    To reach surplus above $\frac{1}{12}$, the adversary can benefit from various combinations of the above types of events. We do a case analysis over all possible combinations. At its top level, the case analysis splits based on whether $D_1$ happens or not. At its next level, if $D_1$ happens, then the case analysis splits based on whether $D_2$ happens or not, whereas if $D_1$ does not happen, the case analysis splits based on whether $O$ happens or not. (If neither $D_1$ nor $O$ happen then there is an additional split, based on whether $D_2$ happens. This split is insignificant, because conditioned on $D_1$ not happening, only one of $D_2$ and $O$ might happen, and the case that $D_2$ happens is treated by the same argument that handles the case that $O$ happens.)
    

    \begin{enumerate}
            \item {There is both a $D_1$ event and a $D_2$ event.} Then, the adversary wins both $e_1$ and $e_2$, and pays 
        at least $\frac{1}{6} + \frac{1}{6}= \frac{1}{3}$. 
        At that point, the agent will have higher budget than the adversary {(as $b_i > \frac{1}{3} > (1 - b_i) - \frac{1}{3}$)}, and the total value of the items that remain is at least $\frac{1}{2}$ {(by the item value bounds)}. The agent can switch to the strategy for $k=1$ and get a value of at least $\frac{1}{2} \cdot \frac{1}{2} = \frac{1}{4}$, as desired.  
        
        \item There is a $D_1$ event but not a $D_2$ event. 
        {There cannot be an $O$ event (as it conflicts $D_1$)}. 
        There can be at most four $S$ events, because the $D_1$ event and five $S$ events would cost the adversary at least $\frac{1}{6} + 5 \cdot \frac {1}{9} > \frac{2}{3}$, which is more than its budget. Likewise, there can be at most four $T$ events, as the agent needs to first win two items, and so the $T$ event can only happen on the items $\{e_4,e_5,e_6,e_7\}$. A similar argument shows that the sum of $S$ events and $T$ events is at most four. As each such event contributes a surplus of at most $\frac{1}{72}$, the surplus is at most $\frac{1}{36} + 4 \cdot \frac{1}{72} = \frac{1}{12}$. 

        \item There is no $D_1$ event but there is a $D_2$ event. Our treatment of the case that there is a $D_1$ event and an $O$ event handles also this case (in both cases the agent wins $e_1$ and the adversary wins $e_2$). See item~5.

        \item There is no $D_1$ event and no $D_2$ event and no $O$ event. 
        There can be at most five $S$ events, because six $S$ events would cost the adversary at least $6 \cdot \frac {1}{9} = \frac{2}{3}$, which is more than its budget. Likewise, there can be at most five $T$ events, as the agent needs to first win two items, and so the $T$ events can only happen on the items $\{e_3, e_4,e_5,e_6,e_7\}$. A similar argument shows that the sum of $S$ events and $T$ events is at most five. As each such event contributes a surplus of at most $\frac{1}{72}$, the surplus is at most $5 \cdot \frac{1}{72} < \frac{1}{12}$.

        \item  
        There is no $D_1$ event and there is an $O$ event. By definition of $O$, this implies that the agent wins $e_1$ and the adversary wins $e_2$ with a surplus. The $O$ event implies that no $D_2$ event happens (in both events the adversary wins $e_2$, and the difference is only in whether the agent could bid according to Table~\ref{tab:f(x)} or not). The analysis here applies also in the complement case in which $D_2$ happens rather than $O$.  

        Let us first note that the fact that an $O$ event happens implies that no $T$ event will ever happen. This is because the agent pays at least $\frac{1}{6}$ for $e_1$. Hence if the agent wins a second item and does not yet reach value of $\frac{1}{4}$, she pays at most $\frac{1}{12}$ for the second item, and has a budget of at least $\frac{1}{12}$ left for any single additional item, whereas the additional item cannot require a higher bid (as item values are non-increasing). 
        
        We conclude that no $D_1$ even happens, no $D_2$ event happens and no $T$ event happens. The only way in which the adversary can gain surplus (beyond the $O$ event) is by $S$ events. We now analyse two cases, depending on the value of $v_i(e_1) + v_i(e_3)$.

        Suppose first that $v_i(e_1) + v_i(e_3) \ge \frac{1}{4}$. In this case the agent bids her entire remaining budget of $e_3$ (as it did on $e_2$), and the adversary is forced to win it. As the agent bid at most $\frac{16}{72} = \frac{2}{9}$ on $e_1$, her bids on each of $e_2$ and $e_3$ are at least $\frac{1}{9}$. The item value bounds imply that $v_i(e_2) + v_i(e_3) < \frac{1}{4}$, and the combined surplus of the two items is at most $\frac{1}{4} - \frac{2}{9} = \frac{1}{36}$. As additional $S$ events can each offer the adversary a surplus of at most $\frac{1}{72}$, and the budget of the adversary suffices for only three additional $S$ events, the total surplus is at most $\frac{1}{36} + 3 \cdot \frac{1}{72} < \frac{1}{12}$, as desired.

        It remains to consider the case that $v_i(e_1) + v_i(e_3) < \frac{1}{4}$. In this case, let us consider the instance $I'$ that remains after the agent wins $e_1$ and the adversary wins $e_2$. Denote the payment of the agent on $e_1$ by $\frac{1}{6} + s$, noting that $0 < s < \frac{4}{72}$ (the upper bound on $s$ is from Table~\ref{tab:f(x)}). Then the budget that the agent has in $I'$ is at least $\frac{1}{6} - s$. We also have that $v_i(e_1) = \frac{1}{6} + s + \frac{2}{72}$, and so the value that agent $i$ needs to achieve in $I'$ is $\frac{1}{4} - v_i(e_1) = \frac{4}{72}-s$, and no single item in $I'$ has such a value. The payment of the agent on $e_1$ plus the payment of the adversary on $e_2$ are at least $b_i$ (as the agent bids all her remaining budget on $e_2$), implying that the total budget in $I'$ is at most $\frac{2}{3}$. Scaling the total budget to~1, the budget $b'_i$ of the agent in $I'$ is at least  $\frac{3}{2}\cdot (\frac{1}{6} - s) = \frac{1}{4} - \frac{3}{2} \cdot s$. As $v_i(e_2) \le v_i(e_1) < \frac{1}{4}$, the total value of all items in $I'$ is at least $\frac{1}{2}$. Scaling values of item so that the total value is~1, the target value that the agent needs to achieve in $I'$ is at most $2\cdot(\frac{4}{72}-s) = \frac{8}{72} - 2s$, and no item has value higher than this. 
        This target value is less than $\frac{1}{2} b'_i$ (regardless of the value of $s$). 
        The agent can now switch to the ``bid your value" strategy for $I'$, as it is known that this strategy guarantees a value of at least $\frac{1}{2} \cdot b'_i$, if no single item has value larger than $b'_i$.  
        \end{enumerate}
        
    \end{proof}

\begin{remark}
\label{ref:lookahead}
    The bidding strategy described in the proof of Claim~\ref{claim:k2} involves substantial ``look ahead". The bid on $e_1$ depends on the value of $e_8$. If the item value bound $v_i(e_8) < \frac{1}{12}$ holds then the bid is according to Table~\ref{tab:f(x)}, whereas if $v_i(e_8) > \frac{1}{12}$ the bid is $\frac{1}{9}$, independently of $v_i(e_1)$. The use Table~\ref{tab:f(x)} while ignoring this look ahead to $v_i(e_8)$ does not guarantee a value of $\frac{1}{4}$. 
    
    For positive $\epsilon < \frac{1}{12}$, consider an example with $b_i = \frac{1}{3} + \frac{\epsilon}{72}$ and eight items whose $v_i$ values are $(14 + \epsilon, 10-\epsilon, 8, 8, 8, 8, 8, 8)$, all scaled down by a factor of~72 (so that the sum of item values is~1). This valuation satisfies all item value bounds except for the one for $e_8$, as $v_i(e_8) = \frac{8}{72} > \frac{1}{12}$. Suppose that the item value bound for $e_8$ is ignored and the agent uses Table~\ref{tab:f(x)}. In this case, the adversary might win $e_1$ paying $\frac{12 + \epsilon}{72}$, and the agent might win $e_2$  and $e_3$, paying $\frac{9 - \epsilon}{72}$ and $\frac{8}{72}$, respectively. At this point, the budget that the agent still holds is $\frac{7+2\epsilon}{72}$. The adversary can win each of the remaining five items, paying $\frac{7 + 2\epsilon}{72}$ on each one. The total payment of the adversary is $\frac{(12 + \epsilon) + 5\cdot(7 + 2\epsilon)}{72} < \frac{2}{3} - \frac{\epsilon}{72}$ (the inequality holds because $\epsilon < \frac{1}{12}$), implying that the adversary has sufficient budget for all these payments. The agent receives a value of $\frac{10-\epsilon}{72} + \frac{8}{72} < \frac{1}{4}$.
\end{remark}

\section{Proofs from Section \ref{sec:chores}}\label{app:chores}

\subsection{Proof of Proposition \ref{prop:rounding-chores}}
In this section we prove Proposition \ref{prop:rounding-chores}.
The proof directly follows from Lemma \ref{lem:chores-MMS-no-add}, Lemma \ref{lem:chores-MES} and Lemma \ref{lem:chores-MMS-additive} presented below. 

\begin{lemma}\label{lem:chores-MMS-no-add}
    For any class $C$ of cost functions  over chores,  $\overline{MMS}$ dominates every feasible unrestricted ex-post share for arbitrary responsibilities  (for class $C$). Moreover, it is the maximal ex-post share for arbitrary responsibilities that dominates every feasible unrestricted ex-post share (for class $C$).
\end{lemma}
\begin{proof}
Fix a class $C$ of cost functions  over chores.  We first show that  $\overline{MMS}$ dominates every feasible unrestricted ex-post share for arbitrary responsibilities,  that is, for any feasible  unrestricted ex-post share $s$ it holds that $\overline{MMS}(c,b)\leq s(c,b)$ for any cost function $c$ and any responsibility $b$. 
Assume in contradiction that for such a feasible share $s$ there exist a cost function $c$ and a responsibility $b$ such that $\overline{MMS}(c,b)> s(c,b)$.
Let {$k+1=(\check{b})^{-1}$}, 
for 
$\check{b}$ which is the  unit {lower} 
bound on $b$.
Consider an instance with $k$ agents having a responsibility $b$, and one agent with responsibility $B= 1-b\cdot k \le 1-\frac{k}{k+1}= \frac{1}{k+1}={\check{b}}$. {All agents have the same cost function $c$.}
By monotonicity of the share $s$ it holds that $s(c,B)\leq s(c,b)$.
As the share $s$ is feasible, there must be an assignment that assigns each one of the $k+1$ agents chores of cost at most $s(c,b)< \overline{MMS}(c,b)= MMS(c, \check{b})= MMS(c, \frac{1}{k+1})$. But this contradicts the definition of $MMS(c, \frac{1}{k+1})$ as the minimum cost that can be suffered by every one of  $k+1$ equally responsible agents with cost function $c$. 

 We next show that $\overline{MMS}$ is the maximal share for arbitrary responsibilities that dominates every feasible unrestricted ex-post share (for class $C$).
 That is, we show that a share that dominates every feasible unrestricted share must dominate $\overline{MMS}$ . 
 Let $\bar{s}$ be an ex-post share that dominates every feasible unrestricted ex-post share for arbitrary responsibilities (for class $C$). 
 Assume in contradiction that $\bar{s}$ does not dominate $\overline{MMS}$. Then there exist a cost function $c\in C$ and a responsibility $b$ such that $\overline{MMS}(c,b)< \bar{s}(c,b)$. By Proposition \ref{pro:personalMMSchores}, the personalized MMS share $MMS_{c,b}$ is feasible for $C$ (ex-post). For this share it holds that $MMS_{c,b}(c,b) = \overline{MMS}(v,b)$. As $\overline{MMS}(c,b)< \bar{s}(c,b)$, the share $\bar{s}$ does not dominate the feasible share  $MMS_{c,b}$, a contradiction. 
\end{proof}

A similar claim can also be proven for MES and ex-ante shares. As the proof is essentially the same as the proof for the ex-post case, we omit it. 

\begin{lemma}\label{lem:chores-MES}
    For any class $C$ of cost functions  over chores,  $\overline{MES}$ dominates every feasible unrestricted ex-ante share for arbitrary responsibilities  (for class $C$). Moreover, it is the maximal ex-ante share for arbitrary responsibilities that dominates every feasible unrestricted ex-ante share (for class $C$). 
\end{lemma}


    \begin{lemma}\label{lem:chores-MMS-additive}
    For the class of additive cost functions  over chores and arbitrary responsibilities, the share $\overline{MMS}$ is the maximal ex-post share  that dominates every feasible {nice} 
    ex-post share, and the share $\overline{PS}$ is the maximal ex-ante share  that dominates every feasible  {nice} 
    ex-ante share. 
    \end{lemma}
    \begin{proof}
   {The proof is similar to the proof of the corresponding claim for goods, presented in Proposition \ref{prop:goods-additive-minimal}, using the corresponding personalized shares defined in  Appendix \ref{sec:nice}. We omit the details.}
        \end{proof}


\end{document}